\documentclass[10pt]{amsart}

\usepackage{amssymb}
\usepackage{verbatim} % for comment
\usepackage[colorlinks=true, citecolor=black, filecolor=black, linkcolor=black, urlcolor=black]{hyperref}

\def\SS{\textbullet~}
\newcommand\introsubsec[1]{\smallskip\noindent\SS\textbf{#1.}~}
\newcommand\subsubsec[1]{\smallskip\noindent\SS\emph{#1.}~}

\def\bp{{\bar\partial}}

\def\bfs{\boldsymbol}
\def\ms{\medskip}

\def\pa{\partial}
\def\PM{\tau} % period matrix 
\def\sm{\setminus}

\def\ti{\tilde}
\def\wh{\widehat}
\def\wt{\widetilde}
\def\ve{\varepsilon}

\def\div{\mathrm{div}}

\def\ee{\mathrm{e}}
\def\dd{\mathrm{d}}
\def\hol{\mathrm{hol}}
\def\id{\mathrm{id}}
\def\sing{\mathrm{sing}}
\def\supp{\mathrm{supp}}
\def\zero{\mathrm{zero}}

\def\Im{\mathrm{Im}}
\def\Jac{\mathrm{Jac}}
\def\NC{\mathrm{NC}}

\def\PS{\mathrm{PS}}
\def\PPS{\mathrm{PPS}}
\def\PS{\mathrm{PS}}
\def\Re{\mathrm{Re}}

\DeclareMathOperator{\Res}{Res}
\DeclareMathOperator{\Sing}{Sing}

\def\dd{\mathrm{d}}
\def\AA{\mathcal{A}}
\def\CC{\mathcal{C}}
\def\DD{\mathcal{D}}
\def\EE{\mathcal{E}}
\def\FF{\mathcal{F}}
\def\HH{\mathcal{H}}
\def\LL{\mathcal{L}}
\def\OO{\mathcal{O}}
\def\PP{\mathcal{P}}
\def\VV{\mathcal{V}}
\def\WW{\mathcal{W}}
\def\XX{\mathcal{X}}
\def\YY{\mathcal{Y}}
\def\ZZ{\mathcal{Z}}

\def\C{\mathbb{C}}
\def\D{\mathbb{D}}
\def\E{\mathbf{E}}
\def\H{\mathbb{H}}
\def\P{\mathbf{P}}
\def\R{\mathbb{R}}
\def\T{\mathbb{T}}
\def\Z{\mathbb{Z}}

\def\wh{\widehat}

\theoremstyle{plain}
\numberwithin{equation}{section}
\newtheorem{thm}{Theorem}[section]
\newtheorem{lem}[thm]{Lemma}
\newtheorem{cor}[thm]{Corollary}
\newtheorem{prop}[thm]{Proposition}
\newtheorem*{claim*}{Claim}
\newcounter{tmp}

\theoremstyle{definition}
\newtheorem*{eg*}{Example}\newtheorem*{egs*}{Examples}
\newtheorem*{q*}{Question}
\newtheorem*{def*}{Definition}

\theoremstyle{remark}
\newtheorem*{rmk*}{Remark}
\newtheorem*{rmks*}{Remarks}

\makeindex

\begin{document}
%\frontmatter
\title{Calculus of conformal fields on a compact Riemann surface}
%%%%%%%%%%%%%%%%%%%%%% author %%%%%%%%%%%%%%%%%%%%%%%%%

\author{Nam-Gyu Kang} 
\address{School of Mathematics, Korea Institute for Advanced Study, \newline Seoul, 02455, Republic of Korea} 
\email{namgyu@kias.re.kr} 
\thanks{The first author was partially supported by Samsung Science and Technology Foundation (SSTF-BA1401-01).} 

\author{Nikolai~G.~Makarov}
\address{Department of Mathematics, California Institute of Technology, \newline Pasadena, CA 91125, USA} 
\email{makarov@caltech.edu}
\thanks{
The second author was supported by NSF grant no. 1500821.}
%%%%%%%%%%%%%%%%%%%%%% author %%%%%%%%%%%%%%%%%%%%%%%%%

%%%%%%%%%%%%%%%%%%% MSC, Key words and phrases %%%%%%%%%%%%%%%%%%
%%%%%%%%%%%%%%%%%%% MSC, Key words and phrases %%%%%%%%%%%%%%%%%%

\begin{abstract} 
We present analytical implementation of conformal field theory on a compact Riemann surface.
We consider statistical fields constructed from background charge modifications of the Gaussian free field and derive Ward identities which represent the Lie derivative operators in terms of the Virasoro fields and  the puncture operators associated with the background charges.  
As applications, we derive Eguchi-Ooguri's version of Ward's equations and certain types of BPZ equations on a torus. 
\end{abstract}

\maketitle
%\vspace{-3em}
%\setcounter{page}{2}
\tableofcontents
%\mainmatter

\section{Introduction and main results} 

\begingroup
\setcounter{tmp}{\value{thm}}
\setcounter{thm}{0} 
\renewcommand\thethm{\Alph{thm}}

Since Belavin, Polyakov, and Zamolodchikov \cite{BPZ84} introduced an operator algebra formalism of conformal fields in the complex plane and classified them in terms of the representation theory of Virasoro algebra, conformal field theory has been applied with great success in string theory, condensed matter physics, and statistical physics.  
For example, it has been used to derive several exact results for the conformally invariant critical clusters in two-dimensional lattice models. 
In mathematics, conformal field theory inspired development of algebraic theories such as the theory of vertex algebras.
Such formalism has been extended to a compact Riemann surface in the physics literature, \cite{EO87b,EO87,Knizhnik,VV87}.  
See also expository papers, \cite{DP88, SCS17}.

One of goals in this paper is to construct several equations in function theory on a compact Riemann surface based on conformal field theoretic approach and show that they are satisfied by correlation functions of certain conformal Fock space fields.
Examples of such equations include Ward's equations, Eguchi-Ooguri equations, and BPZ equations.
The other goal is to develop basic resources for the study of relations between conformal field theory and  SLE (Schramm-Loewner evolution) theory in various settings like a doubly connected domain (\cite{BKT17}), arbitrary non-random pre-pre-Schwarzian modifications, and various insertions (e.g. $N$-leg operators). 
Such correspondences in some conformal settings are well known in the physics literature, e.g. see \cite{BBK05,HBB10}.

\introsubsec{Gaussian free field on a compact Riemann surface}
The fields we consider in this paper are the statistical fields constructed from the Gaussian free field $\Phi$ on a compact Riemann surface of genus $g.$ 
We define the Gaussian free field both as a distributional field $f\mapsto \Phi(f) $ on the space of test functions $f$ satisfying a certain neutrality condition and as a bi-variant Fock space field $(z,z_0)\mapsto\Phi(z,z_0).$
The correlation function of Gaussian free field is given by bipolar Green's function and computed in terms of the Riemann theta function and Abel-Jacobi map in the higher genus case.
Given a conformal metric $\rho$, we also define the bosonic field $\Phi_\rho(z)$ associated with $\rho$ as a single-variable Fock space field  and compute its correlation functions in terms of the resolvent kernel. 
If one considers its differences $\Phi_\rho(z)-\Phi_\rho(z_0)(=\Phi(z,z_0))$, then its dependence on $\rho$ disappears.

\introsubsec{OPE family of Gaussian free field} 
We construct some basic fields including the current field $J:=\pa \Phi,$ the Virasoro field $T$ via operator product expansion (OPE) of $J$, and the multi-vertex fields as the OPE exponentials of the Gaussian free field. 
The correlations of Fock space fields provide the coordinate-free functions in the function theory on a compact Riemann surface. 
Neither a basis for the homology nor a basis for the space of all holomorphic 1-differentials needs to be fixed. 
For example, the correlation $\E\,J(z)\overline{J(z_0)}$ represents the inverse of the period matrix and does not require a specific choice of basis for homology or cohomology. 
It is important that the correlations are independent of the choice of an element $e$ in the theta divisor. 
The Bergman projective connection appears in the 1-point function $\E\,T$ and its independence of $e$ is obvious from the conformal field theoretic approach. 
We consider the OPE family $\FF$ of the Gaussian free field with central charge $c=1$ where the Lie derivative operators can be represented by the Virasoro field $T$ within correlations,
$$\LL_v^+X(z) = \frac1{2\pi i} \oint_{(z)} vTX(z), \qquad X \in \FF,$$
where $\LL_v^+$ is the $\C$-linear part of Lie derivative operator $\LL_v,$ i.e. $\LL_v^+:=\frac12(\LL_v - \LL_{iv}).$
This family $\FF$ includes the Gaussian free field, the current field, the Virasoro field, and the multi-vertex fields.

\introsubsec{Central charge modifications and background charge modifications}
We modify the Gaussian free field by adding a non-random pre-pre-Schwarzian (PPS) form $\varphi$ of order $(ib,ib)$ to implement a version of conformal field theory with central charge $c = 1-12b^2.$
A PPS form $\varphi$ we consider mostly in this paper is single-valued and harmonic except for finitely many logarithmic singularities. 
Such a simple PPS form $\varphi$ determines a background charge $\bfs\beta = i\pa\bp\varphi/\pi.$
We remark that the metric we consider is not a smooth one, but a flat one with finitely many conical singularities. 
An extension of the Gauss-Bonnet theorem to the singular metric leads to the neutrality condition $(\NC_b):$ 
$$\sum \beta_ k =b\,\chi(M)$$
for a background charge $\bfs\beta = i\pa\bp\varphi/\pi = \sum \beta_k \delta_{q_k}.$
Here $\chi(M)$ is the Euler characteristic of $M.$
Given a divisor $\bfs\beta$ with the neutrality condition $(\NC_b),$ we construct a PPS form $\varphi_{\bfs\beta}$ such that $\bfs\beta = i\pa\bp\varphi_{\bfs\beta}/\pi.$
Such a PPS form $\varphi_{\bfs\beta}$ is unique up to an additive constant.
For an imaginary $b$ ($c>1$), we interpret such a modification as a Liouville field with the background charge $\bfs\beta$ and zero cosmological constant. 
See \cite{DKRV16,DRV16} for construction of Liouville quantum field theory with non-zero cosmological constant on the Riemann sphere and the complex tori, respectively. 
 
The OPE family $\FF_{\bfs\beta}^*(\equiv \FF_{\bfs\beta}(M^*))$ on the punctured surface $M^*(:=M\setminus\supp\,\bfs\beta$) of the modified field $\Phi_{\bfs\beta}$ has the central charge $c = 1-12b^2$ and the Virasoro field 
$$T_{\bfs\beta} = -\frac12 J_{\bfs\beta} *J_{\bfs\beta} +ib\pa J_{\bfs\beta}, \quad J_{\bfs\beta} = J + j_{\bfs\beta}, \quad  j_{\bfs\beta} = \pa \varphi_{\bfs\beta}.$$
The reason that we consider a PPS form in the definition of $\Phi_{\bfs\beta}$ and that we construct $T_{\bfs\beta}$ in this way is clear from the algebraic point of view. 
It is well known in the algebraic literature that if the generators $\ti L_n$ are constructed as 
$$\ti L_n = L_n - ib(n+1)J_n,$$ 
where $L_n$'s are the Virasoro generators, the modes of $T$ and $J_n$'s are the modes of the current field $J,$ then the generators $\ti L_n$ represent the Virasoro algebra with central charge $c= 1-12b^2.$
One can identify $\ti L_n$ with the modes of $T_{\bfs\beta}.$
Later, we extend the OPE family by adding all renormalized fields (for example $J_q:=J*1_{q}$)  rooted at the points $q\in\supp\,\bfs\beta$ and denote it by $\FF_{\bfs\beta}\equiv\FF_{\bfs\beta}(M).$ 

\introsubsec{Liouville action and partition function}
The bosonic field $\Phi_\rho$ associated with a conformal metric $\rho$ can be produced as a random distribution by the Dirichlet action $\DD:$
$$\DD(\psi) = \frac1{4\pi} (\psi,\Delta\psi)_{L^2(\rho)}$$
on the Sobolev space $\WW(\rho).$
More precisely, for a random function $\psi_n = \sum_{j=1}^n a_j e_j,$ 
($e_j$'s are normalized eigenfunctions of the positive Laplace operator $\Delta_\rho$ associated to the eigenvalues $\lambda_j$ with $\lambda_0 = 0, e_0 =1$)
the function $\ee^{-\DD(\psi_n)}$ gives rise to a joint distribution measure $\mu_n$ for random coefficients $a_j$ and its partition function $Z_n:$
$$\dd\mu_n = \frac1{Z_n} \ee^{-\DD(\psi_n)} \dd a_1\,\cdots\dd a_n, \qquad  Z_n =  \prod_{j=1}^n \sqrt{\frac{2\pi}{\lambda_j}}.$$
We identify the limiting field $\lim_{n\to\infty} \psi_n$ with $\Phi_\rho,$ see Subsection~\ref{ss: Dirichlet action}. 
Of course, a sequence $Z_n$ of partition functions diverges as $n\to\infty.$ 
However, by means of the zeta function regularization of the determinant of $\Delta_\rho,$ one can define the partition function $Z$ associated with the bosonic field $\Phi_\rho.$ 
For example, it is well known (e.g. see \cite{CFT}) that
$$Z \equiv Z(\tau) = \frac1{\sqrt{\Im\,\tau}|\eta(\tau)^2|}$$
on the torus $\T_\Lambda:=\C/\Lambda, (\Lambda = \Z + \tau\Z, \Im\,\tau>0).$ 
Here, $\eta$ is the Dedekind $\eta$-function, $\eta(\tau) = q^{1/24}\prod_{n=1}^\infty(1-q^n), q= \ee^{2\pi \tau}.$
The relation between the partition function and the Virasoro field is also well known in the physics literature:
\begin{equation} \label{eq: Ooguri4Z}
2\pi i\, \frac\pa{\pa\tau} \log Z(\tau) = \E\,T \big(\equiv \eta_1 - \frac\pi{2\,\Im\,\tau}\big)
\end{equation}
in the $\T_\Lambda$-uniformization, see \cite{EO87}. 
Here, $\eta_1 = \zeta(1/2)$ and $\zeta$ is the Weierstrass $\zeta$-function.   

We extend the above well-known results to more general cases. 
For an imaginary modification parameter $b$ (i.e. in the theory with central charge $c> 1$), we show that the background charge modification $\Phi_{\bfs\beta}$ can be obtained by the Liouville action associated to $\bfs\beta:$
$$S(\varphi) = D(\psi) - i\int\psi\bfs\beta \big( = D(\psi)  +\frac1\pi\int\psi\pa\bp\varphi_{\bfs\beta}\big), \qquad \varphi = \varphi_* + \psi,$$
where $\varphi_*$ is a fixed reference $\PPS(ib,ib)$-form. 
See Theorem~\ref{Liouville action}.
By means of a proper regularization, the partition function $Z_{\bfs\beta}$ associated with the modification field $\Phi_{\bfs\beta}$ is defined by 
$$Z_{\bfs\beta} = Z \PP_{\bfs\beta},$$
where $\PP_{\bfs\beta}$ is the puncture operator defined in \eqref{eq: puncture} below.  
We extend the identity~\eqref{eq: Ooguri4Z} to the case $\bfs\beta\ne \bfs 0:$ in the $\T_\Lambda$-uniformization 
$$2\pi i\, \frac\pa{\pa\tau} \log Z_{\bfs\beta}(\tau) = \oint_{[0,1]}\E\,T_{\bfs\beta}(\xi)\,\dd\xi,$$ 
see Corollary~\ref{dtau log Zbeta}.

\introsubsec{Vertex operators and Coulomb gas correlation functions}
Given a background charge $\bfs\beta$ with the neutrality condition $(\NC_b)$ and a divisor $\bfs\tau$ with the neutrality condition $(\NC_0),$ we construct the OPE exponential $\OO_{\bfs\beta}[\bfs\tau]$ of the bosonic field $i\Phi_{\bfs\beta}[\bfs\tau].$  
If $\bfs\beta$ and $\bfs\tau$ do not have an overlapping support, then its construction is straightforward and it is in the OPE family $\FF_{\bfs\beta}^*.$
If they do, then the construction of $\OO_{\bfs\beta}[\bfs\tau]$ resorts to a certain renormalization procedure and it belongs to the (extended) OPE family $\FF_{\bfs\beta}.$  
Wick's exponential $\VV^\odot[\bfs\tau]:=\ee^{\odot i\Phi[\bfs\tau]}$ is a well-defined Fock's space functionals and we have 
$$\OO_{\bfs\beta}[\bfs\tau] = \VV^\odot[\bfs\tau]\,\E\,\OO_{\bfs\beta}[\bfs\tau].$$
One of our main goals is to compute the correlation $\E\,\OO_{\bfs\beta}[\bfs\tau].$
(A more traditional notation for the correlation function is $\langle \cdot \rangle.$)
We consider the genus zero case first. 
For this purpose, we introduce the Coulomb gas correlation functions $\CC_{(b)}[\bfs\sigma]$ on $\wh\C$ as the differentials 
with conformal dimensions $(\lambda_j,\lambda_j),$ 
$$\lambda_j = \frac12{\sigma_j^2} - \sigma_j b$$
at $z_j$'s (including infinity) and with values 
$$\prod_{\substack{j<k\\z_j,z_k\ne\infty}}|z_j-z_k|^{2\sigma_j\sigma_k},\qquad (z_j\in\wh\C)$$
in the identity chart of $\C$ and the usual chart $z\mapsto -1/z$ at infinity.
In the genus zero case, we have 
$$\E\,\OO_{\bfs\beta}[\bfs\tau] = \frac{\CC_{(b)}[\bfs\tau+\bfs\beta]}{\CC_{(b)}[\bfs\beta]}.$$
This formula can be extended to the higher genus case, see Theorem~\ref{main: EO} below. 
If $g\ge 2,$ then the Coulomb gas correlation functions $\CC_{(b)}[\bfs\sigma]$ can be expressed in terms of the theta function $\Theta$  in $\C^g$ and Abel-Jacobi map $\AA:$ for $\bfs\sigma = \sum \sigma_j\cdot z_j,$  
\begin{align*}
\CC_{(b)}[\bfs\sigma] = C_* & \prod_j |\omega(z_j)|^{2\lambda_j}
\prod_{j<k} |\theta(z_j-z_k)|^{2\sigma_j\sigma_k}
\prod_{j,l} |\theta(z_j-p_l)|^{2b\sigma_j} \\
& \exp \big(2\pi(\Im\,\PM)^{-1}\Im\,\AA[\bfs\sigma]\cdot\Im\,\AA[\bfs\sigma-2\bfs \beta_0]\big),
\end{align*}
where $\omega:=\sum_j \pa_j\Theta(e)\omega_j,$ $\{\omega_j\}$ is a canonical basis for the cohomology, $\theta(z) = \Theta(\AA(z)-e)$, $p_l$'s are the zeros of $\omega,$ $\bfs\beta_0 = -b\cdot(\omega),$ and $(\omega)$ is the divisor of $\omega,$ i.e. $\bfs\beta_0 = -\sum_{l=1}^{2g-2} b\cdot p_l.$

\begin{thm} \label{main: EO}
There is a unique map $\CC_{(b)}:\bfs\sigma\mapsto \CC_{(b)}[\bfs\sigma]$ (up to a multiplicative constant) which assigns a divisor $\bfs\sigma$ satisfying the neutrality condition $(\NC_b)$ to a differential $\CC_{(b)}[\bfs\sigma]$ such that 
for any background charge $\bfs\beta$ and any divisor $\bfs\tau$ with the neutrality condition $(\NC_0),$ we have 
\begin{equation} \label{eq: EO}
\OO_{\bfs\beta}[\bfs\tau] = \frac{\CC_{(b)}[\bfs\tau+\bfs\beta]}{\CC_{(b)}[\bfs\beta]}\,\VV^\odot[\bfs\tau], \quad \VV^\odot[\bfs\tau]:=\ee^{\odot i\Phi[\bfs\tau]}.
\end{equation}
\end{thm}

OPE exponentials we consider in this paper are single-valued or non-chiral. 
We study chiral multi-vertex fields in a forthcoming paper. 

\introsubsec{Insertion formulas}
We explain how the insertion of Wick's exponential of the Gaussian free field gives rise to the change of background charge modifications. 
A similar type of insertion procedure plays an important role in establishing the relation between the chordal SLE theory and conformal field theory in a simply connected domain with two marked boundary points in \cite{KM13}.  
See \cite{KM12} for its radial version with one marked boundary point and one marked interior point.

\begin{thm} \label{main: Insertion}
Given two background charges $\bfs\beta_1, \bfs\beta_2$ with the neutrality condition $(\NC_b),$
the image of $\FF_{\bfs\beta_1}$ under the insertion of $\VV^{\odot}[\bfs\beta_2-\bfs\beta_1]$ is $\FF_{\bfs\beta_2}.$ 
%on $M\setminus(\bfs q_1 \cup \bfs q_2),$ where $\bfs q_j = \supp\,\bfs\beta_j.$
\end{thm}

\introsubsec{Puncture operator and Ward identity for the extended OPE family}
Given a non-random meromorphic vector field $v$ with poles at $\xi_k$'s, we define the Ward functional $W_v^+$ by 
$$W_v^+:=-\frac1{2\pi i}\sum_k \oint_{(\xi_k)} v T.$$
Given a background charge $\bfs\beta,$ we define the puncture operator $\PP_{\bfs\beta}$ by 
\begin{equation} \label{eq: puncture}
\PP_{\bfs\beta}=\CC_{(b)}[\bfs\beta].
\end{equation}

We represent the Ward functionals as the Lie derivative operators up to the conjugation by the puncture operators $\PP_{\bfs\beta}.$ 
The following Ward identity holds for any compact Riemann surface of genus $g.$

\begin{thm} \label{main: Ward identity} \index{Ward identity}
In correlations with any string of fields in the extend OPE family $\FF_{\bfs\beta},$ we have
$$W_v^+ = \PP_{\bfs\beta}^{-1}\LL_v^+ \PP_{\bfs\beta}.$$
\end{thm}

Whereas the above theorem is somewhat abstract, its application to a specific choice of vector field gives rise to the rather concrete equations, the so-called Ward's equations. In this paper, we study Ward's equations in the cases of genus zero and one. 
We now state various Ward's equations on the torus $\T_\Lambda.$
In the theorem below, we consider the vector fields $v_{\xi,\xi_0}$ and $v_\xi$ in $\T_\Lambda$  given in terms of the Jacobi-theta function $\theta$ and the Weierstrass $\wp$-function as 
$$v_{\xi,\xi_0}(z) = \frac{\theta'(\xi-z)}{\theta(\xi-z)} - \frac{\theta'(\xi_0-z)}{\theta(\xi_0-z)}, \qquad v_\xi(z) = -\wp(\xi-z).$$

\begin{thm} \label{main: KM2} \index{Ward's equations}
For  any tensor product $\XX_{\bfs\beta}$ of fields in the OPE family $\FF_{\bfs\beta},$
\begin{align}
\E\,(T_{\bfs\beta}(\xi)-T_{\bfs\beta}(\xi_0) )\XX &=  \E\,(T_{\bfs\beta}(\xi)-T_{\bfs\beta}(\xi_0))\,\E\,\XX + \E\, \LL_{v_{\xi,\xi_0}}^+\XX, \nonumber\\
\E\,\pa T_{\bfs\beta}(\xi) \XX &=  \E\,\pa T_{\bfs\beta}(\xi)\,\E\,\XX + \E\, \LL_{v_\xi}^+\XX  \label{eq: Ward4paT}
\end{align}
in the $\T_\Lambda$-uniformization.
\end{thm}

\begin{eg*} 
Applying Ward's equation~\eqref{eq: Ward4paT} to $\XX = J(z)J(z_0)$ (with $b=0, \bfs\beta=\bfs0$), we derive the addition theorem for Weierstrass $\wp$-function:
$$\begin{vmatrix}
1 & \wp(z)  & \wp'(z) \\
1& \wp(w) & \wp'(w) \\
1 &\wp(z+w)& -\wp'(z+w) \\
\end{vmatrix} = 0.$$
\end{eg*}

\begin{eg*} 
Ward's equation~\eqref{eq: Ward4paT} with $\XX = T(z)$ (with $b=0, \bfs\beta=\bfs0$) gives rise to the following differential equation for $\wp:$
$$\wp\wp' = \frac1{12}\wp'''.$$
\end{eg*}

\begin{eg*} 
Applying Ward's equation~\eqref{eq: Ward4paT} to $\XX = \OO[\tau\cdot z_1- \tau\cdot z_2]$ (with $b=0, \bfs\beta=\bfs0$), we derive the addition theorem for Weierstrass $\zeta$-function:
$$\zeta(z+w) = \zeta(z) + \zeta(w) + \frac12\,\frac{\zeta''(z)-\zeta''(w)}{\zeta'(z)-\zeta'(w)}.$$
\end{eg*}

\begin{eg*} 
Ward's equation~\eqref{eq: Ward4paT} with $\XX_{\bfs\beta} = J_{\bfs\beta}(z), (\bfs\beta = 1\cdot q_1 - 1\cdot q_2)$ gives rise to
\begin{align}\label{eq: new addition0}
&\wp(z_1)\wp(z_2) - \big(\wp(z_1)+\wp(z_2)\big)\wp(z_1+z_2)-\frac12 \frac{\wp'(z_1)-\wp'(z_2)}{\wp(z_1)-\wp(z_2)}\wp'(z_1+z_2)=\\
&\wp(\ti z_1)\wp(\ti z_2) - \big(\wp(\ti z_1)+\wp(\ti z_2)\big)\wp(\ti z_1+\ti z_2)-\frac12 \frac{\wp'(\ti z_1)-\wp'(\ti z_2)}{\wp(\ti z_1)-\wp(\ti z_2)}\wp'(\ti z_1+\ti z_2) \nonumber
\end{align}
if $z_1+z_2 = \ti z_1+\ti z_2.$
It is clear that \eqref{eq: new addition0} follows from \eqref{eq: new addition} below. 
\end{eg*}

\introsubsec{Eguchi and Ooguri's version of Ward's equations}
To obtain Ward's equations on a torus for the insertion of the Virasoro field with a single node, one needs to consider a multivalued meromorphic vector field with a single simple pole. 
Let 
\begin{equation} \label{eq: Ooguri vector field}
\ti v_\xi(z) = \zeta(\xi-z) + 2\eta_1 z
\end{equation}
in the $\T_\Lambda$-uniformization. 

\begin{thm} \label{main: KM2'} \index{Ward's equations}
For  any tensor product $\XX_{\bfs\beta}$ of fields in the OPE family $\FF_{\bfs\beta},$ \begin{equation} \label{eq: EOKM_main} 
\E\,T_{\bfs\beta}(\xi)\XX_{\bfs\beta} =  \E\,T_{\bfs\beta}(\xi)\E\,\XX_{\bfs\beta} + \LL_{\ti v_\xi}^+\E\,\XX_{\bfs\beta}+2\pi i \, \frac\pa{\pa\tau} \,\E\,\XX_{\bfs\beta}
\end{equation}
in the $\T_\Lambda$-uniformization.
\end{thm}

The above theorem for the tensor product $\XX_{\bfs\beta}$ of differentials $X_j$ in the OPE family $\FF_{\bfs\beta}$ with $\bfs\beta = \bfs 0$ first appeared in the paper \cite{EO87} by Eguchi and Ooguri:
\begin{align} \label{main: EO87}
\E\,T_{\bfs0}(\xi)\XX_{\bfs0} &=  \E\,T_{\bfs0}(\xi)\,\E\,\XX_{\bfs0} +2\pi i \, \frac\pa{\pa\tau} \,\E\,\XX_{\bfs0}
 \\
&+\sum_j  \big((\zeta(\xi-z_j) +2\eta_1z_j)\pa_j  + \lambda_j( \wp(\xi-z_j) + 2\eta_1)\big)\E\,\XX_{\bfs0}, \nonumber
\end{align}
where $\XX_{\bfs0} = X_1(z_1)\cdots X_n(z_n).$

\begin{eg*} Applying $\XX_{\bfs0} = J(z)J(z_0)$ $(z\ne z_0)$ in the case $b=0$ to \eqref{main: EO87}, we obtain 
\begin{align}\label{eq: new addition}
\big(\wp(z_1+z_2)&-\wp(z_1)\big) \big(\wp(z_1+z_2)-\wp(z_2)\big) - \frac12\frac{\wp'(z_1)-\wp'(z_2)}{\wp(z_1)-\wp(z_2)}\,\wp'(z_1+z_2)\\ &= \frac12\wp''(z_1+z_2). \nonumber
\end{align}
\end{eg*}

It is quite remarkable to observe that this choice of $\ti v_\xi$ gives rise to the so-called Teichm\"uller term
$\pa_\tau \E\,\XX_{\bfs\beta}$ in \eqref{eq: EOKM_main}. 
In the physics literature, its appearance is explained in term of the Virasoro generator $L_0,$ see \cite{EO87}.
We rather present the following theorem (Eguchi-Ooguri equations) and prove it by using the pseudo-addition theorem for the Weierstrass $\zeta$-function.    

\begin{thm} \label{main: KM} \index{Eguchi-Ooguri equations}
For any tensor product $\XX_{\bfs\beta}$ of fields in the OPE family $\FF_{\bfs\beta},$
$$%\begin{equation} \label{eq: int T}
\frac1{2\pi i}\oint_{[0,1]} \E\,T_{\bfs\beta}(\xi)\XX_{\bfs\beta}\,\dd \xi =  \E\,T_{\bfs\beta}(\xi)\E\,\XX_{\bfs\beta} + \frac{\pa}{\pa\tau}\,\E\,\XX_{\bfs\beta}.
$$%\end{equation}
\end{thm}

\begin{eg*} 
Applying Theorem~\ref{main: KM} to $\XX = J(z)\overline{J(z)},$ we obtain the following well-known relation:
$$\eta_1 = -\frac16\frac{\theta'''(0)}{\theta'(0)}.$$  
\end{eg*}

\begin{eg*} 
A simple application of Theorem~\ref{main: KM} to $\XX = \Phi*\Phi(z)$ gives rise to 
$$ 2\pi i \, \frac \pa {\pa_\tau} \log {\theta'(0)} =-3\eta_1. $$
\end{eg*} 

\begin{eg*}
Applying Theorem~\ref{main: KM} to $\XX = J(z)J(z_0)$ with $z\ne z_0,$ we obtain the following well-known formula:
$$2\pi i \, \frac{\pa}{\pa\tau} \, \wp(z) = 2\wp(z)^2 + \zeta(z)\wp'(z) - \frac{g_2}3 - 2\eta_1(2\wp(z)+z\wp'(z)),$$
where $g_2$ appears in the expansion of $\wp$ around $z=0,$
$$\wp(z) = \frac1{z^2} + \frac{g_2}{20}z^2 + O(z^4).$$
\end{eg*}

\begin{eg*}
A simple application of Theorem~\ref{main: KM} to $\XX = T(z)$ gives rise to 
$$2 \pi i\, \frac{\pa}{\pa\tau}\eta_1 = -2\eta_1^2  + \frac1{24} g_2.$$
\end{eg*}

\introsubsec{Singular vectors and BPZ equations}
Given a modification parameter $b,$ let $a$ be one of the solutions to the quadratic equation $2x(x+b) = 1$ for $x.$ 
In terms of the action of Virasoro generators $L_n$ (the modes of $T_{\bfs\beta}$) defined by 
$$L_n(z):=\frac1{2\pi i}\oint_{(z)}(\zeta-z)^{n+1} T_{\bfs\beta}(\zeta)\,\dd \zeta,$$
the OPE exponentials $\OO_{\bfs\beta}(z)\equiv \OO_{\bfs\beta}^{(a,\bfs\tau)}(z) := \OO_{\bfs\beta}[a\cdot z + \bfs \tau]$ $(\bfs\tau = \sum \tau_j z_j, j\le0)$ with the neutrality condition $a + \sum \tau_j = 0$ satisfy the level two degeneracy equations 
$$\big(L_{-2}(z)-\frac1{2a^2} L_{-1}^2(z)\big)\OO_{\bfs\beta}(z)=0,$$
see Proposition~\ref{level2degeneracy}.
Combining these equations with Ward's equations~\eqref{eq: EOKM_main}, we derive the following types of Belavin-Polyakov-Zamolodchikov equations (BPZ equations) on a torus. \index{BPZ equations}
See Theorem~\ref{BPZ0} for the BPZ equations on the Riemann sphere. 

Let the vector field $\ti v_z$ be given by \eqref{eq: Ooguri vector field}.

\begin{thm} \label{BPZ1}
Let $\OO_{\bfs\beta}(z) = \OO_{\bfs\beta}^{(a,\bfs\tau)}(z)$ and $\lambda = \frac12 a^2 -ab.$ 
If $z\notin\supp\,\bfs\beta \cup \supp\,\bfs\tau,$ then for any tensor product $\XX_{\bfs\beta}= X_1(z_1)\cdots X_n(z_n)$ of fields $X_j$ in $\FF_{\bfs\beta},$ we have 
\begin{align*}
\frac1{2a^2}\pa_z^2\E\,\OO_{\bfs\beta}(z)\XX_{\bfs\beta} &=  \E\,T_{\bfs\beta}(z) \E\,\OO_{\bfs\beta}(z)\XX_{\bfs\beta}+ 2\pi i\,\frac{\pa}{\pa\tau} \E\,\OO_{\bfs\beta}(z)\XX_{\bfs\beta}\\
&+(2\eta_1z\pa_{z} + 2\lambda \eta_1)\E\,\OO_{\bfs\beta}(z)\XX_{\bfs\beta}+ \check\LL_{\ti v_z}^+ \E\,\OO_{\bfs\beta}(z)\XX_{\bfs\beta}
\end{align*}
in the $\T_\Lambda$-uniformization.
Here, the Lie derivative operator $\check\LL_{\ti v_z}^+$ does not apply to the $z$-variable. 
\end{thm}

For example, if $X_j$'s are $(\lambda_j,\lambda_{*j})$-differentials and if $\supp\,\bfs\beta \cap \supp\,\bfs\tau = \emptyset,$ then $\check\LL_{\ti v_z}^+ \E\,\OO_{\bfs\beta}(z)\XX_{\bfs\beta}$ reads as 
\begin{align*}
\check\LL_{\ti v_z}^+ \E\,\OO_{\bfs\beta}(z)\XX_{\bfs\beta} &= \sum_j  \big((\zeta(z-z_j) +2\eta_1z_j)\pa_j  + \lambda_j( \wp(z-z_j) + 2\eta_1)\big) \E\,\OO_{\bfs\beta}(z)\XX_{\bfs\beta} \\
&+ \sum_k (\zeta(z-q_k) +2\eta_1q_k)\pa_{q_k} \E\,\OO_{\bfs\beta}(z)\XX_{\bfs\beta} .
\end{align*}
To the best of our knowledge, the above BPZ equation for the differentials in the case of the Ising model ($a^2 = 2/3, \lambda =1/2$) with $\bfs\beta = \bfs0$ first appeared in \cite{EO87}. 

It follows from the so-called fusion rule (e.g. see \cite{CFT,Dubedat15}) that the OPE exponentials $\OO_{\bfs\beta}(z) \equiv \OO_{\bfs\beta}^{(2a,\bfs\tau)}(z):=\OO_{\bfs\beta}[2a\cdot z + \bfs\tau]$  satisfy the level three degeneracy equations 
$$\Big(\frac1{8a^2}L_{-1}^3 -  L_{-2}L_{-1} + (2a^2-\frac12)L_{-3}\Big)\OO_{\bfs\beta}(z),$$
see Proposition~\ref{level3degeneracy}.
We combine Ward's equations with the above level three degeneracy equations to derive certain types of BPZ equations in the genus zero case (Theorem~\ref{BPZl3g0}) and in the genus one case (Theorem~\ref{BPZl3g1}). 
See \cite{EO88,EO89} for the BPZ equations associated with the level three/six degeneracy equations for the differentials in the case of the Ising model.

\endgroup
\setcounter{thm}{\thetmp}

\section{Gaussian free field on a compact Riemann surface} \label{s: GFF}

We define the Gaussian free field $\Phi$ on a compact Riemann surface $M$ of genus $g$ as a bi-variant Fock space field. We extend this definition to a divisor with a certain neutrality condition. 
The Gaussian free field is the simplest example of distributional Fock space fields. 
It can be represented by a linear map $f \mapsto \Phi(f)$ from a space of test functions $f$ with the neutrality condition $\NC_{(0)},$ or with integral zero, to the $L^2$ space of random variables. 

\subsection{Bipolar Green's function}
In general, a Green's function does not exist for a compact Riemann surface $M.$
However, it has a bipolar Green's function. \index{bipolar Green's function}
Let $p,q$ be distinct marked points of $M.$
By definition, \emph{bipolar Green's function} $z\mapsto G_{p,q}(z)$ with poles at $p$ and $q$ is harmonic on $M\setminus\{p,q\},$ 
and satisfies
\begin{align*}
G_{p,q}(z) &= \phantom{-}\log\frac1{|z-p|} + O(1) \qquad (z\to p),\\
 G_{p,q}(z) &= -\log\frac1{|z-q|} + O(1) \qquad (z\to q)
\end{align*}
(in some/any chart $\phi.$).
We often use the same letter for the point $p$ on $M$ and $\phi(p)$ on the coordinate plane.
Note that a bipolar Green's function is not uniquely determined. 
However, it is unique up to additive constants.

We denote by $\DD^{(0,0)}$ the collection of non-random $(0,0)$-differentials. 
See Subsection~\ref{ss: Conformal Fock space fields} for basic properties of differentials. 
The \emph{Sobolev space} \index{Sobolev space} $\WW(M)$ is defined as the completion of $\DD^{(0,0)} / \C$ with respect to the scalar product
$$(f,g)_{\nabla} = \int \pa f \overline{\pa g} +\int \bp f \pa \bar g.$$
Since $M$ is closed, the integrals in the right-hand side are equal and 
$$(f,g)_{\nabla} = 2\int f(-\pa\bp\bar g).$$

We write $\DD_{0}^{(1,1)}$ for the space of $(1,1)$-differentials with integral zero.
The \emph{energy space} \index{energy space} $\EE(M)$ is the completion of the space $\DD_{0}^{(1,1)}$ with respect to scalar product inherited from $ \WW(M)$ such that 
$$-2\pa\bp : \WW(M) \to \EE(M) $$
is a unitary operator. 

We can think of $G_{p,q}$ as ``generalized" elements of $\WW(M)$ so $\pa\bp$ maps it to the ``generalized" elements of $\EE(M).$
Indeed, we have 
\begin{equation} \label{eq: pabpG}
-\pa\bp G_{p,q} = \frac\pi2(\delta_p-\delta_q),
\end{equation}
in the sense of distribution. 
Bipolar Green's function satisfies the bilinear relation:
\begin{equation} \label{eq: bilinear relation}
\Big(\,\frac1\pi(G_{p,q},G_{\ti p,\ti q})_\nabla =\Big) G_{p,q}(\ti p)-G_{p,q}(\ti q) = G_{\ti p,\ti q}(p)-G_{\ti p,\ti q}(q).
\end{equation}
Again, we formally take the scalar product of ``generalized" elements $G_{p,q}$ and $G_{\ti p,\ti q}$ of $\WW(M)$.

\subsubsec{The genus zero case} 
Bipolar Green's function for the Riemann sphere $M = \widehat \C$ is given by 
$$G_{p,q}(z) = \log\left|\frac{z-q}{z-p}\right|, \quad G_{p,\infty}(z) = \log \frac1{|z-p|}$$
up to an additive constant.
Note that 
$(G_{p,q},G_{\ti p,\ti q})_\nabla/\pi$ is the logarithm of the modulus of the cross-ratio $\lambda:$
$$\lambda(p,q;\ti p,\ti q)=\frac{(\ti p -q)(\ti q-p)}{(\ti p-p)(\ti q-q)}.$$

\subsubsec{The genus one case} 
Let us consider the complex torus $\T_\Lambda:=\C/\Lambda$ of genus one, where $\Lambda = \Z + \tau\Z$ is the group generated by $z\mapsto z+1, z\mapsto z + \tau.$ 
Here, the modular parameter $\tau$ is in the upper half-plane $\H:=\{z\in\C\,|\, \Im\, z > 0\}.$
Bipolar Green's function for $M = \T_\Lambda$ is given by 
$$G_{p,q}(z) = \log \left|\frac{\theta(z-q)}{\theta(z-p)}\right| -2 \pi\,\frac{\Im(p-q)\,\Im\,z}{\Im\,\tau},$$
up to an additive constant, see \cite{EO87}.
Here $\theta$ is one of the four \emph{Jacobi theta functions}: \index{Jacobi theta function}
$$\theta(z) \equiv \theta_1(z\,|\,\tau) = 2\sum_{n=1}^\infty (-1)^{n-1} \ee^{\pi i\tau(n-\frac12)^2}\sin(2n-1)\pi z.$$
In the $\T_\Lambda$-uniformization, 
$$\frac1\pi\,(G_{p,q},G_{\ti p,\ti q})_\nabla = \log |\lambda(p,q;\ti p,\ti q)| - 2 \pi\,\frac{\Im(p-q)\Im(\ti p-\ti q)}{\Im\,\tau},$$
where the cross-ratio $\lambda(p,q;\ti p,\ti q)$ is defined by 
\begin{equation}\label{eq: cross-ratio1}
\lambda(p,q;\ti p,\ti q)=\frac{\theta(\ti p -q)\theta(\ti q-p)}{\theta(\ti p-p)\theta(\ti q-q)}.
\end{equation}

\subsubsec{The higher genus case}
The previous representation for $(G_{p,q},G_{\ti p,\ti q})_\nabla/\pi$ can be extended to a compact Riemann surface $M$ of genus $g\ge 2.$
We review some basic properties of Riemann theta function $\Theta$ in Appendix~\ref{sec: Theta}. 
We consider the lattice $\Lambda = \Z^g + \PM \Z^g$ in $\C^g$ associated with the period matrix $\PM$ of $M$ and set 
$$\T_\Lambda \equiv \T_\Lambda^{g}:=\C^{g}/\Lambda.$$
We now express bipolar Green's function for $M$ in terms of the Riemann theta function $\Theta$: up to an additive constant, 
$$G_{p,q}(z) =\log\Big|\frac{\Theta(\AA(z) - \AA(q)-e)}{\Theta(\AA(z) - \AA(p)-e)}\Big| - 2\pi (\Im\,\PM)^{-1}\Im(\AA(p) - \AA(q))\cdot \Im\,\AA(z),$$
see Theorem~\ref{Green M}. 
Here, $\AA:M\to \T_\Lambda$ is the Abel-Jacobi map. 
The point $e \in \T_\Lambda$ is chosen to be an arbitrary element of the theta divisor, i.e. $\Theta(e) = 0$ such that neither of maps $z\mapsto\Theta(\AA(z)-\AA(p)-e), z\mapsto\Theta(\AA(z)-\AA(q)-e)$ is identically zero. 

It is well known that as a multivalued function, the generalized cross-ratio 
\begin{equation}\label{eq: cross-ratio2}
\lambda(p,q;\ti p,\ti q) = \frac{\Theta(\AA(\ti p) - \AA(q)-e)\Theta(\AA(\ti q) - \AA(p)-e)}{\Theta(\AA(\ti p) - \AA(p)-e)\Theta(\AA(\ti q) - \AA(q)-e)}
\end{equation}
of four points $p,q,\ti p,\ti q$ on $M$ does not depend on the choice of $e,$ see \cite[VII]{FK}. 
Let
$$\theta(z) = \Theta(\AA(z) - e).$$
Since the Abel-Jacobi map can be extended to the group of divisors on $M,$ we rewrite the generalized cross-ratio in terms of $\theta:$
$$\lambda(p,q;\ti p,\ti q)=\frac{\theta(\ti p-q)\theta(\ti q-p)}{\theta(\ti p-p)\theta(\ti q-q)}.$$ 
With this notation, we have 
$$\frac1\pi\,(G_{p,q},G_{\ti p,\ti q})_\nabla = \log |\lambda(p,q;\ti p,\ti q)| - 2 \pi\,(\Im\,\PM)^{-1}\, \Im(\AA(p) - \AA(q))\cdot\Im(\AA(\ti p)-\AA(\ti q)).$$

We can formally think of bipolar Green's function $G_{p,q}$ as the difference of two formal elements $G_p$ and $G_q,$
$$G_{p,q}(z) = G_p(z) - G_q(z).$$
For example, on the Riemann sphere $\wh\C,$ one can choose 
$$G_p(z) = -\log|z-p|.$$
In the torus $ \T_\Lambda, (\Lambda = \Z + \tau\Z),$ one can choose 
$$G_p(z) = \log \frac1{|\theta(z-p)|}-2\pi\, \frac
{\Im\,p\,\Im\,z}{\Im\,\tau}.$$
In the higher genus case, 
$$G_p(z) = \log \frac1{|\Theta(\AA(z) - \AA(p)-e)|}- 2\pi\,(\Im\,\PM)^{-1}\Im\,\AA(p)\cdot\Im\,\AA(z)$$ 
with a generic $e$ in the theta divisor. 

The bilinear relation 
$$G_{p,q}(\ti p)-G_{p,q}(\ti q) = G_{\ti p,\ti q}(p)-G_{\ti p,\ti q}(q)$$
for bipolar Green's function 
can be verified by formal representation $G_{p,q}(z) = G_p(z) - G_q(z).$
In a similar way, the Jacobi identity 
$$G_{p,q}(z)+G_{q,z}(p)+G_{z,p}(q)=0$$
for bipolar Green's function can be verified. 

For all $\mu\in\EE(M), q\in M,$ we have 
\begin{equation}\label{eq: E-norm}
\|\mu\|_{\EE(M)}^2 = \iint_{(z,p)} 2 G_{p,q}(z) \mu(z)\overline{\mu(p)},
\end{equation}
where the integral is taken over the variables over $z$ and $p.$
To verify this, let $\mu = -2\pa\bp f$ for some $f\in\WW(M).$
It follows from \eqref{eq: pabpG} that 
$$\iint_{(z,p)} 2 G_{p,q}(z) \mu(z)\overline{\mu(p)} = \int f(p)\overline{\mu(p)} = 2\int f(-\pa\bp \bar f).$$
Clearly, we have $\|f\|_{\WW(M)}^2 = \|\mu\|_{\EE(M)}^2.$

\subsection{Gaussian free field} \label{ss: GFF}

Recall that the energy space $\EE=\EE(M)$ is the completion of the space $\DD_0^{(1,1)}(M)$ 
of $(1,1)$-differentials $\mu$ satisfying the neutrality condition $(\NC_0):$
$$\int\mu = 0$$
with respect to the norm $\|\cdot\|_{\EE}:$
$$\|\mu\|_{\EE}^2 = \iint 2 G_{\zeta,\eta}(z) \mu(z)\overline{\mu(\zeta)}$$
for all $\eta\in M.$

The \emph{Gaussian free field} \index{Gaussian free field} $\Phi$ is a Gaussian field indexed by the energy space $\EE(M),$
$$\Phi:\EE(M) \to L^2(\Omega,\P),$$
where $(\Omega,\P)$ is some probability space.
By definition, $\Phi$ is an isometry such that the image consists of centered Gaussian random variables. 

We introduce the Fock space functionals $\Phi(z,z_0)$ as ``generalized" elements of Fock space 
$$\Phi(z,z_0) = \Phi(\delta_z - \delta_{z_0}),$$
where $\delta_z - \delta_{z_0}$ is the ``generalized" elements of $\EE(M).$
For each $z\in M,$ we choose a family of test functions $f_{\ve,z}$ supported in a disc of radius $\ve$ about $z$ such that $f_{\ve,z} \to \delta_z$ as $\ve\to 0.$
We now define Gaussian random variables
$$\Phi_\ve(z,z_0) = \Phi(f_{\ve,z}-f_{\ve,z_0}).$$
Then they are centered and their covariances are 
$$\E[\Phi_\ve(p,q)\Phi_\ve(\ti p,\ti q)] = (f_{\ve,p}-f_{\ve,q},f_{\ve,\ti p}-f_{\ve,\ti q})_\EE$$
for $\ti p, \ti q \notin\{p,q\}.$
It follows from \eqref{eq: E-norm} and the bilinear relation~\eqref{eq: bilinear relation} that
$$\E[\Phi_\ve(p,q)\Phi_\ve(\ti p,\ti q)] \to 2(G_{p,q}(\ti p)-G_{p,q}(\ti q))= \frac2\pi\,(G_{p,q},G_{\ti p,\ti q}) _\nabla$$
as $\ve\to0.$
We now define the correlation function of Gaussian free field by
\begin{equation} \label{eq: correlator of Phi}
\E[\Phi(p,q)\Phi(\ti p,\ti q)] = 2(G_{p,q}(\ti p)-G_{p,q}(\ti q)),\quad (\ti p, \ti q \notin\{p,q\})
\end{equation}
so that correlation functionals can be approximated by genuine random variables.

\subsubsec{The genus zero case} 
In the $\wh\C$-uniformization, we have 
\begin{equation} \label{eq: EPhi0}
\E\, \Phi(p,q)\Phi(\ti p,\ti q) = \log|\lambda(p,q;\ti p,\ti q)|^2 = 2\log\left|\frac{(\ti p -q)(\ti q-p)}{(\ti p-p)(\ti q-q)}\right|.
\end{equation}
\subsubsec{The genus one case} 
In the $\T_\Lambda$-uniformization, we have 
\begin{equation}\label{eq: EPhi1}
\E\, \Phi(p,q)\Phi(\ti p,\ti q) = \log|\lambda(p,q;\ti p,\ti q)|^2 -4\pi\,\frac{\Im(p-q)\,\Im\,(\ti p-\ti q)}{\Im\,\tau},
\end{equation}
where $\lambda(p,q;\ti p,\ti q)$ is the cross-ratio, see \eqref{eq: cross-ratio1}.

\subsubsec{The higher genus case} We have 
\begin{align}\label{eq: EPhi2}
\E\, \Phi(p,q)\Phi(\ti p,\ti q) &= \log|\lambda(p,q;\ti p,\ti q)|^2 \\
&-4\pi (\Im\,\PM)^{-1}\Im(\AA(p) - \AA(q))\cdot \Im\,(\AA(\ti p) - \AA(\ti q)), \nonumber
\end{align}
where $\lambda(p,q;\ti p,\ti q)$ is the cross-ratio, see \eqref{eq: cross-ratio2}.

\begin{rmk*}
As a Fock space field, the Gaussian free field $\Phi_D$ in a simply connected domain $D$ can be constructed from the Gaussian free field $\Phi_M$ on the Schottky double $M$ of $D$. 
More precisely, 
$$\Phi_D(z) = \Phi_M^+(z,z^*),$$
where $z\mapsto z^*$ is the canonical involution on $M.$ 
As a (multivalued) bi-variant field, $\Phi^+ \equiv \Phi_M^+$ is defined by
$$\Phi^+(p,q) = \{\Phi^+(\zeta)\,:\, \gamma \textrm{ is a path from } q \textrm{ to } p\} $$  
and 
$$\Phi^+(\gamma) = \int_\gamma J(\zeta)\,\dd\zeta, \quad J(\zeta) = \pa_\zeta \Phi(\zeta,z).$$
For example, in the upper half-plane $\H,$ we have  
$$\E\,\Phi_\H(\zeta) \Phi_\H(z) = \E\,\Phi_{\hat\C}(\zeta,\bar\zeta) \Phi_{\hat\C}(z,\bar z) =\log\frac{(\zeta -\bar z)(\bar\zeta-z)}{(\zeta-z)(\bar\zeta -\bar z)}=2\log\left|\frac{\zeta-\bar z}{\zeta-z}\right| = 2 G_\H(\zeta,z)$$
and in the unit disk $\D,$ we have  
$$\E\,\Phi_\D(\zeta) \Phi_\D(z) = \log\frac{(\zeta -z^*)(\zeta^*-z)}{(\zeta-z)(\zeta^*-z^*)}=2\log\left|\frac{1-\zeta\bar z}{\zeta-z}\right|  \Big(= \log\Big|\frac{\zeta-z^*}{(\zeta-z)z^*}\Big|\,\Big)= 2 G_\D(\zeta,z),$$
where $z^* = 1/\bar z$ is the symmetric point of $z$ with respect to the unit circle $\pa\D.$
A chordal/radial version of conformal field theory developed in \cite{KM13,KM12} can be constructed from the theory on a compact Riemann surface of genus zero.
More details will appear in a forthcoming paper. 
\end{rmk*}

\subsubsec{Formal Gaussian free field} 
We introduce $\Phi$ as a centered Gaussian formal field with formal correlation
$$\E\,\Phi(z)\Phi(z_0) = 2G_{z_0}(z).$$
Let
$$\bfs\tau= \sum_j \tau_j\cdot z_j,$$ 
where $\{z_j\}_{j=1}^N$ is a finite set of (distinct) points on $M$ and $\tau_j$'s are real
numbers, (``charges" at $z_j$'s),
$\tau_j =\tau_{z_j} =\bfs\tau(z_j).$
We can think of $\bfs\tau$ as a \emph{divisor}, \index{divisor} a function $\bfs\tau:M \to \R$ which takes the value $0$ at all points except for finitely many points. 
Sometimes it is convenient to consider $\bfs\tau$ as an atomic measure, 
$\bfs\tau= \sum_j \tau_j\cdot \delta_{z_j}.$ 
For example, 
$$\int \bfs\tau= \sum_j \tau_j.$$
Some of $\tau_j$'s can be zero, and in any case $\tau_z = 0$ if $z$ is not one of $z_j$'s.

In the genus zero case, we construct the formal 1-point field $\Phi$ alternately via the so-called ``rooting procedure."
We set 
$$\Phi_*(z):=\Phi(z,z_*)$$
as a centered Gaussian field and define 
$$\E[\Phi_*(\zeta)\Phi_*(z)]$$
as a normalized limit of $\log|\lambda(\zeta,\zeta_*,z,z_*)|^2,$ namely
$$\lim_{\zeta_*,z_*\to\infty} \Big(\log|\lambda(\zeta,\zeta_*,z,z_*)|^2+\log\Big|\frac{\zeta_*-z_*}{\zeta_* z_*}\Big|^2\Big) = \log\frac1{|\zeta-z|^2}.$$ 

Given a divisor $\bfs\tau= \sum_j \tau_j\cdot z_j,$ we define a formal bosonic field $\Phi[\bfs\tau]$ as a linear combination 
$$\Phi[\bfs\tau] := \sum_j\tau_j\Phi(z_j).$$
It is not possible to define the formal 1-point field $\Phi$ as a Fock space field, but 
under the \emph{neutrality condition} $(\NC_0):$ \index{neutrality condition!(NC$_0$)} 
$$ \sum_j\tau_j = 0,$$ 
$\Phi[\bfs\tau]$ is well-defined as a Fock space field.
For example, we have a representation $ \Phi(z,z_0) = \Phi(z) - \Phi(z_0).$

\begin{lem}
If a divisor $\bfs\tau$ satisfies $(\NC_0)$, then the formal bosonic field $\Phi[\bfs\tau]$ is a well-defined single-valued Fock space functional (in the complement of the support of $\bfs\tau$).
\end{lem} 

\begin{proof}
Let us choose any point $z_0\in M$. (It can be one of $z_j$'s.) 
Then
$$ \Phi[\bfs\tau]=\Phi(z_0)\int\bfs\tau + \sum_j\tau_j\Phi(z_j,z_0).$$
Under the neutrality condition, the first term in the right-hand side vanishes.
The representation is not unique, of course, but it is unique in the sense of Fock space fields. 
Indeed, if neither $\zeta$ nor $\zeta_0$ is in the support of $\bfs\tau,$ then 
$$\E\,\Phi[\bfs\tau]\Phi(\zeta,\zeta_0) = 2\sum_j \tau_j G_{\zeta,\zeta_0}(z_j)$$
does not depend on the choice of $z_0.$
\end{proof}

If both $\bfs\sigma = \sum \sigma_j\cdot z_j $ and $\bfs\tau = \sum \tau_k \cdot \zeta_k$ satisfy the neutrality condition $(\NC_0)$ and if they have disjoint supports, then
$$\E\,\Phi[\bfs\sigma]\Phi[\bfs\tau] = 2\sum_{j,k} \sigma_j\tau_k G_{\zeta_k,\zeta_0}(z_j)$$
does not depend on the choice of $z_0,\zeta_0.$

\subsection{(Conformal) Fock space fields} \label{ss: Conformal Fock space fields}
In this subsection we recall several notions introduced in \cite{KM13} to define conformal Fock space fields as linear combinations of basic Fock space fields. By construction, basic Fock space fields are equipped with certain conformal structures. Thus we further require any of coefficients to have an assignment of a smooth function to each local chart. 
 
 \subsubsec{Fock space correlation functionals}
By definition, a \emph{Fock space correlation functional} \index{Fock space correlation functional} is a linear combination over $\C$ of the constant $1$ and formal expressions 
$$\XX = X_1(\zeta_1,z_1)\odot\cdots\odot X_n(\zeta_n,z_n),$$
where $X_j$'s are derivatives of the Gaussian free field
$$X_j(\zeta_j,z_j)=\pa_{\zeta_j}^{\alpha_j}\bp_{\zeta_j}^{\alpha_{*j}}\pa_{z_j}^{\beta_j}\bp_{z_j}^{\beta_{*j}}\Phi(\zeta_j,z_j)$$
and $\odot$ denotes Wick's product.  
The points $\zeta_j,z_j$ in the expression of $\XX$ are called the \emph{node} \index{node} of $\XX$ and we write $S(\XX)$ or $S_\XX$ for the set of nodes of $\XX.$ By definition, $S(\sum c_j \XX_j) = \bigcup S(\XX_j).$
Some infinite linear combinations are allowed, e.g. Wick's exponential
$$\ee^{\odot \alpha\Phi(z,z_0)}=\sum_{n=0}^\infty \frac{\alpha^n}{n!}\Phi^{\odot n}(z,z_0).$$

For $\XX_j = X_{j1}(\zeta_{j1},z_{j1}) \odot \cdots \odot X_{jn_j}(\zeta_{jn_j},z_{jn_j}),$ we define their \emph{tensor products} \index{tensor product} by Wick's formula,
$$\XX_1\cdots \XX_m = \sum\prod_{\{v,v'\}} \E[X_{v}(\zeta_{v},z_{v})X_{v'}(\zeta_{v'},z_{v'})] \underset{v''}{\textstyle\bigodot} X_{v''}(\zeta_{v''},,z_{v''}),$$
where an unordered pair $\{v,v'\}$ of indices $jn_j$ is allowed for different $j$'s ($v = jn_j, v' = kn_k$ for some $j\ne k$) and the Wick's product is taken over unpaired $v''$'s. 
Here the sum is taken over all possible Feynman diagrams with vertices $v$ and edges $\{v,v'\}.$ 
The correlation $\E[X_{v}(\zeta_{v},z_{v})X_{v'}(\zeta_{v'},z_{v'})]$ is defined as the derivatives of $\E[\Phi(\zeta_{v},z_{v})\Phi(\zeta_{v'},z_{v'})].$

We define the correlation $\E\,\XX_1\cdots \XX_m $ of tensor products of $\XX_j$ in terms of the chaos decomposition of $\XX_1\cdots \XX_m:$
\begin{equation} \label{def: tensor product}
\E\,\XX_1\cdots \XX_m = \sum\prod_{\{v,v'\}} \E[X_{v}(\zeta_{v},z_{v})X_{v'}(\zeta_{v'},z_{v'})] \E\,\underset{v''}{\textstyle\bigodot} X_{v''}(\zeta_{v''},,z_{v''}),
\end{equation}
where 
$$\E\,\underset{v''}{\textstyle\bigodot} X_{v''}(\zeta_{v''},,z_{v''}) = \begin{cases} 1 &\textrm{if there is no unpaired vertex } v'', \\ 0 &\textrm{otherwise}.\end{cases}.$$
For example, $\E\, 1 = 1,$ $\E\, X_1(\zeta_1,z_1)\odot\cdots\odot X_n(\zeta_n,z_n) = 0,$ $\E\,\ee^{\odot \alpha\Phi(z,z_0)}  = 1$ and 
$$\E\,\Phi(\zeta_1,z_1) \cdots\Phi(\zeta_n,z_n) = \sum \prod_k 2\big(G_{\zeta_{i_k},z_{i_k}}(\zeta_{j_k}) - G_{\zeta_{i_k},z_{i_k}}(z_{j_k})\big),$$
where the sum is over all partitions of the set $\{1,\cdots,n\}$ into disjoint paris $\{i_k,j_k\}.$

\subsubsec{Fock space fields}
For derivatives $X_j$ of the Gaussian free field, the map 
$$(\zeta_1,z_1, \cdots,\zeta_n,z_n)\mapsto X_1(\zeta_1,z_1)\odot\cdots\odot X_n(\zeta_n,z_n)$$ 
is called a \emph{basic Fock space field}. \index{Fock space field!basic}
A \emph{(general) Fock space space field} is a linear combination of basic fields $X_j,$ \index{Fock space field!general}
$$X = \sum_j f_j X_j,$$  
where coefficients $f_j$ are non-random conformal fields, see their definition below. 
Recall that a local chart on a compact Riemann surface $M$ is a conformal map $\phi$ from an open subset $U$ of $M$ to the complex plane $\C$ such that the transition maps
$$h = \wt\phi\circ\phi^{-1}: \phi(U\cap\wt U) \to \wt\phi(U\cap\wt U)$$
are conformal transformations whenever $U\cap\wt U\ne\emptyset.$
A \emph{non-random conformal field} $f$ \index{non-random conformal field}  is defined as an assignment of a smooth function
$(f\,\|\,\phi):\phi U\to\C$ to each local chart $\phi:U\to\C.$
By definition, a conformal Fock space field is a linear combination of basic fields with non-random conformal fields as coefficients. 

We define the differential operators $\pa,\bp$ on Fock space fields such that their action on basic Fock space fields is consistent with the definition of derivatives of the Gaussian free field and their extension to general Fock space fields is defined by linearity and by Leibniz's rule with respect to the multiplication by non-random conformal fields. By definition, a field $X$ is \emph{holomorphic} in $M$ if all correlation functions $\E\,X(\zeta)\YY$ are holomorphic in $\zeta\in M\setminus S_\YY$ for any Fock space correlation functional $\YY.$  
\index{Fock space field!holomorphic}
 
\subsubsec{Differentials and forms}
For a pair of non-negative integers $(\lambda,\lambda_*),$ a non-random conformal field $f$ is called a \emph{differential} \index{differential} of conformal dimension $(\lambda,\lambda_*)$ if for any two overlapping local charts $\phi$ and $\wt\phi,$ we have 
$$ f = (h')^\lambda(\overline{h'})^{\lambda_*}\wt f\circ h,$$
where we write $f, \wt f$ for $(f\,\|\,\phi),(f\,\|\,\wt\phi),$ respectively and $h$ is the transition map between $\phi$ and $\wt\phi.$
We write $\DD^{(\lambda,\lambda_*)} = \DD^{(\lambda,\lambda_*)} (M)$ for the collection of non-random $(\lambda,\lambda_*)$-differentials. 
By definition, a non-random field $f$ is a \emph{Schwarzian form} (a \emph{pre-Schwarzian form} (a PS form), a \emph{pre-pre-Schwarzian form} (a PPS form)) of order $\mu$ if it has the transformation law
\index{Schwarzian form}
\index{pre-Schwarzian form} \index{PS form}
\index{pre-pre-Schwarzian form} \index{PPS form}
$$f=(h')^2\widetilde{f}\circ h + \mu S_h, \quad (f=h'\widetilde{f}\circ h + \mu N_h, \quad f= \widetilde{f}\circ h + \mu \log h'),$$
respectively, where 
$S_h = N_h' -{N_h^2}/2, (N_h =(\log h')')$
is Schwarzian derivative of $h.$ 
A notion of $(\lambda,\lambda_*)$-differentials, and (pre-)Schwarzian forms can be extended to conformal Fock space fields. 
For example, we say that a conformal Fock space field $X$ is a differential if the non-random field
$p\mapsto \E[X(p)\YY]$ is a differential in $p\in M$ for $\YY = \Phi(p_1,q_1)\odot \cdots \odot\Phi(p_n,q_n).$

We define the product $fg$ of a $(\lambda,\lambda_*)$-differential $f$ and a $(\mu,\mu_*)$-differential $g$ by
$$(fg\,\|\,\phi) := (f\,\|\,\phi)(g\,\|\,\phi).$$
Then $fg$ is $(\lambda+\mu,\lambda_*+\mu_*)$-differential.
There is a natural operation of complex conjugation on conformal fields:
$$(\bar f\,\|\,\phi) := \overline{(f\,\|\,\phi)}.$$
For a $(\lambda,\lambda_*)$-differential $f$, $\bar f$ is a $(\lambda_*,\lambda)$-differential.

We define derivatives of non-random conformal fields by differentiating in local charts:
$$(\pa f\,\|\,\phi) := \pa (f\,\|\,\phi),\qquad (\bp f\,\|\,\phi) := \bp (f\,\|\,\phi).$$
For example, 
$\pa \oplus \bp: \DD^{(0,0)} \to \DD^{(1,0)} \oplus \DD^{(0,1)}$
and 
$\pa \bp: \DD^{(0,0)} \to \DD^{(1,1)}.$

\begin{rmk*} One can identify the differential $a\in \DD^{(1,0)}$ with the differential form $a(z)\,\dd z$ (where $a(z)$ is the value of $a$ in $z$-coordinate), and $b \in \DD^{(0,1)}$ with the differential form $b(z)\,\dd \bar z.$ 
One can also identify $F\in \DD^{(1,1)}$ with 
$$F(z)\,\dd x\wedge \dd y = \frac i2 F(z)\,\dd z\wedge \dd \bar z.$$
This identification leads to the definition of integration of $(1,1)$-differentials;
$$\int_M F = \int_M F(z) \,\dd x\wedge \dd y.$$
\end{rmk*}

\subsection{Modular invariance of Gaussian free field} \label{ss: Modular invariance of GFF}

In this subsection we discuss modular invariance of Gaussian free field in the genus one case. 
Suppose that two tori $\T_\Lambda, \T_{\Lambda'}$ are conformally equivalent, where the lattices  $\Lambda,\Lambda' $ are generated by the modular parameters $\tau, \tau', (\Im\,\tau,\Im\,\tau'>0)$ respectively, i.e. $\Lambda=\Z + \tau\Z, \Lambda'= \Z + \tau'\Z.$ 
Then there is $\gamma \in \C$ such that $\Lambda = \gamma \Lambda',$ hence $\gamma,\gamma\tau'$ generate $\Lambda.$ 
Thus there are $a,b,c,d\in\Z$ such that $\gamma\tau' = a\tau + b$ and $\gamma = c\tau +d.$
Eliminating $\gamma,$ we find the relation between $\tau$ and $\tau'$ as 
\begin{equation}\label{eq: tau'}
\tau ' = \frac{a\tau+b}{c\tau+d}.
\end{equation}

A conformal transformation $f$ from $(\T_\Lambda,0)$ onto $(\T_{\Lambda'},0)$ can be lifted to the universal covering space and we get a linear map $\ti f$
\begin{equation}\label{eq: lift}
\ti f:z\mapsto z' := \,\frac z\gamma,\qquad (\gamma = c\tau +d).
\end{equation}
Let $G$ ($G'$) be the covering group of $\Lambda$ ($\Lambda'$) generated by the translations $z\mapsto z+1, z\mapsto z + \tau$ $(z\mapsto z + \tau'),$ respectively. 
The linear map $\ti f$  induces an isomorphism $F$ between the covering group $G$ of $\Lambda$ and $G'$ of $\Lambda':$ 
$$F: g \mapsto \ti f \circ g \circ \ti f^{-1}.$$
Since $F$ is an isomorphism, the matrix representing \eqref{eq: tau'} is invertible and we have the determinant $ad-bc = \pm1.$
However, both $\tau$ and $\tau'$ are in the upper half-plane and therefore $ad-bc =1.$

The modular invariance of the Gaussian free field is clear by its construction. Indeed, its constructions does not resort to coordinate dependence. 
For example, we have
\begin{equation}\label{eq: modular invariance GFF}
\E\,\Phi_{_{\T_{\Lambda'}}}(p',q')\Phi_{_{\T_{\Lambda'}}}(\ti p', \ti q') = \E\,\Phi_{_{\T_\Lambda}}(p,q)\Phi_{_{\T_\Lambda}}(\ti p, \ti q),
\end{equation}
where $\tau'$ is given by \eqref{eq: tau'} and $p' = p/(c\tau+d),$ etc. Cf. \eqref{eq: lift}.

\begin{eg*}
Using  \eqref{eq: modular invariance GFF} and the identity 
$$\frac{\Im(z/\tau) \, \Im(w/\tau)}{\Im\,1/\tau}  + \frac{\Im\,z\,\Im\,w}{\Im\,\tau} = \Im\,\frac{zw}\tau,$$
we derive 
\begin{align*}
\log&\left|\frac {\theta(\ti p'-q'|\tau')\theta(\ti q'-p'|\tau')} {\theta(\ti p'-p'|\tau')\theta(\ti q'- q'|\tau')} \right|
-\log\left|\frac {\theta(\ti p-q|\tau)\theta(\ti q-p|\tau)} {\theta(\ti p-p|\tau)\theta(\ti q- q|\tau)} \right| \\
&=-2\pi\, \Im\big(\frac1\tau (\ti p-\ti q)(p-q)\big),
\end{align*}
where $\tau' = -1/\tau$ and $p' = p/\tau,$ etc. 
Differentiating the above, and then integrating,  
\begin{align*}
\log&\,\frac {\theta(\ti p'-q'|\tau')\theta(\ti q'-p'|\tau')} {\theta(\ti p'-p'|\tau')\theta(\ti q'- q'|\tau')} 
-\log\frac {\theta(\ti p-q|\tau)\theta(\ti q-p|\tau)} {\theta(\ti p-p|\tau)\theta(\ti q- q|\tau)} \\
&=2\pi i\, \big(\frac1\tau (\ti p-\ti q)(p-q)\big).
\end{align*}
Exponentiating the above, we have the transformation formula for the cross-ratios, 
\begin{align*}
\frac {\theta(\ti p'-q'|\tau')\theta(\ti q'-p'|\tau')} {\theta(\ti p'-p'|\tau')\theta(\ti q'- q'|\tau')} 
&=\frac {\theta(\ti p-q|\tau)\theta(\ti q-p|\tau)} {\theta(\ti p-p|\tau)\theta(\ti q- q|\tau)} 
\,\ee^{2\pi i\, (\ti p-\ti q)(p-q)/\tau }\\
&=\frac {\theta(\ti p-q|\tau)\theta(\ti q-p|\tau)} {\theta(\ti p-p|\tau)\theta(\ti q- q|\tau)} 
\,\ee^{\pi i \, ((\ti p-q)^2 + (\ti q-p)^2 -(\ti p-p)^2 - (\ti q-q)^2 )/\tau}.
\end{align*}
It is well known (e.g. see \cite[p.~75]{Chandrasekharan}) that  if $\sqrt{\tau/i} =1$ for $\tau = i,$ then
\begin{equation*} \label{eq: theta invariance}
\theta(z'|\tau') =  -i\sqrt{\tau/i\,}\,\ee^{\pi i z^2/\tau } \theta(z|\tau).
\end{equation*}
\end{eg*}

\section{Conformal metrics and Gaussian fields} \label{s: Conformal metrics}
Given a conformal metric $\rho,$ we construct the bosonic field $\Phi_\rho$ associated with $\rho.$
Formally, one can think of the Fock space functional $\Phi_\rho(z)$ as $\Phi(\delta_z-\rho)$ if $\int \rho=1.$ 
Its 2-point function is computed in terms of the resolvent kernel. 
In the last subsection we represent the distributional field $\Phi_\rho$ as a random series in the Sobolev space. 
This representation is used to compute the field produced by the Dirichlet action. 

\subsection{Spectral theory and Laplacian} \label{ss: spectral theory}

A \emph{conformal metric} \index{conformal metric} $\rho$ on $M$ is a positive $(1,1)$-differential on $M.$ 
We can identify the differential $\rho$ with the Riemannian metric
$$\dd s^2 = \rho(z)\,|\dd z|^2.$$
Then $\rho\,\dd x\wedge \dd y$ is the volume form and $\rho\,\dd x\, \dd y$ is the corresponding volume.
The positive \emph{Laplace operator} \index{Laplace operator} $\Delta \equiv \Delta_\rho:C^\infty(M) \to C^\infty(M)$
is a differential operator 
$$f \mapsto -\frac{2\pa\bp f}{\rho}.$$
(Then $\varphi = \log \rho$ is a pre-pre-Schwarzian form (PPS form). 
See Subsection~\ref{ss: Conformal Fock space fields} or \ref{ss: Holomorphic PPS forms} for the definition of PPS form.)
The \emph{curvature} of the conformal metric $\rho$ on $M$ is defined as \index{curvature}
$$\kappa = \Delta_\rho\log\rho = -\frac{2\pa\bp\log \rho}{\rho}.$$

\begin{eg*} For the Riemann sphere $M = \widehat \C,$ the conformal metric 
$$\rho(z) = \frac4{(1+|z|^2)^2} \quad \textrm{ in } (\C, \id)$$
is the only metric such that its group of orientation-preserving isometries is $SO(3).$
In this case, $M$ has the positive constant curvature $\kappa\equiv 1.$ 
\end{eg*}

The positive Laplace operator $\Delta\equiv\Delta_\rho:C^\infty\to C^\infty$ extends to a self-adjoint operator in $L^2(\rho).$
This extension is a non-negative operator with compact resolvent.
We use the same notation $\Delta_\rho$ to denote this extension. 
Let 
$$0 = \lambda_0 < \lambda_1 \le \cdots $$
be the eigenvalues of $\Delta_\rho$, and let $\{\phi_n\}_{n\ge0}$ be orthonormal eigenfunctions,
$$\pa\bp\phi_n = -\frac12\lambda_n\phi_n\rho. $$
We can choose eigenfunctions real. From now on, we assume that $\int\rho = 1.$
Clearly, $\phi_0$ is constant $1$ and 
\begin{equation}\label{eq: phirho0} 
\int\phi_n\rho = 0 \qquad (n\ge1).
\end{equation}
We write $\phi_0^{\perp}$ for the orthogonal complement of $\langle\phi_0\rangle$ in $L^2(\rho).$

We now introduce the \emph{resolvent kernel} \index{resolvent kernel} $R\equiv R_\rho:$
$$R_\rho(z,w) := \sum_{n\ge1}\frac1{\lambda_n}\,\phi_n(z)\phi_n(w).$$
This is the kernel of the integral operator $R\equiv R_\rho:$
$$(R_\rho f)(z) := \int R_\rho(z,w) f(w)\rho(w).$$
\begin{rmk*}
Restricting to $\phi_0^{\perp},$ $R:\phi_0^{\perp}\to\phi_0^{\perp}$ is the inverse of $\Delta:\phi_0^{\perp}\to\phi_0^{\perp}.$
Both $\Delta$ and $R$ are zero on $\langle\phi_0\rangle.$
\end{rmk*}

\begin{prop} \label{R}The resolvent kernel $R$ has the following basic properties: 
\renewcommand{\theenumi}{\alph{enumi}}
{\setlength{\leftmargini}{1.8em}
\begin{enumerate}
\item \label{R1} $R(z,w) = R(w,z);$ 
\item \label{R2} $\displaystyle\int R(z,\cdot)\rho = 0;$
\item \label{R3} $\pa\bp R(z,\cdot) = -\frac12(\delta_z - \rho);$
\item \label{R4} $R(z,w) = -\frac1\pi\log|z-w| + O(1)$ in every chart as $z\to w.$
\end{enumerate}}
\end{prop}

\begin{proof} 
\eqref{R1} The symmetric property is clear from the definition of the resolvent kernel. \\
\eqref{R2} It follows from \eqref{eq: phirho0}.\\
\eqref{R3} For any smooth function $f,$ by the previous remark, we have 
$$R\Delta f = f^{\perp} = f - (f,\phi_0)\phi_0 = f - {\int f\rho}.$$
Here, $f^{\perp}$ is the orthogonal projection of $f$ onto $\phi_0^{\perp}.$
On the other hand, integration by parts gives
$$(R\Delta f)(z) = \int R(z,w) (\Delta f)(w) \rho(w) = -2 \int R(z,\cdot)\pa\bp f = -2\int f \pa\bp R(z,\cdot). $$
\eqref{R4} Let $K(z,w) = R(z,w) + \frac1\pi \log|z-w|.$ Then $\pa\bp K(z,\cdot) = \frac12 \rho.$ 
It is clear that $K(z,w) = O(1)$ in every chart as $z\to w.$
\end{proof}

Let us define $R_p(z) = R(z,p)$ and 
$$R_{p,q}(z) = R_p(z) - R_q(z).$$
Then by Proposition~\ref{R} \eqref{R1} and \eqref{R2}, we have 
\begin{equation} \label{eq: Rrho}
\int R_{p,q} \rho = 0.
\end{equation}
\begin{cor} \label{cor: RG} We have 
$$R_{p,q} = \frac1\pi \, G_{p,q} \qquad (\mathrm{mod}\,\, \R).$$
\end{cor}
\begin{proof} It follows from \eqref{eq: pabpG} and the previous proposition that 
$$-\pa\bp R_{p,q} = -\pa\bp R(p,\cdot)+\pa\bp R(q,\cdot) = \frac12\big(\delta_p-\delta_q\big) = -\frac1\pi\pa\bp G_{p,q}.$$
The logarithmic singularities of $R_{p,q} - \frac1\pi \, G_{p,q}$ at both $p$ and $q$ are removable. 
Thus $R_{p,q} - \frac1\pi \, G_{p,q}$ is a bounded harmonic function on $M,$ which is constant. 
\end{proof}
By the previous corollary, the function $R_{p,q}$ is independent of the metric up to an additive constant (depending on $p,q$). 
The constant is determined by the equation~\eqref{eq: Rrho}.

\begin{rmk*} By definition and Proposition~\ref{R} \eqref{R2}, we have 
$$R(z,p) = \int_{(q)} R_{p,q}(z) \,\rho(q).$$
This identity has the following physical interpretation: $R(z,p)$ has a potential at $z$ for the field with a (unit) source at $p$ and with sinks distributed according to $\rho.$
\end{rmk*}

\subsection{Sobolev space}
Recall that the elements of $\WW \equiv \WW(M)$ are functions $f$ modulo $\C, [f] = f + \C.$ 
The Sobolev space \index{Sobolev space} $\WW$ comes with the scalar product
$$(f,g)_\nabla :=-2\int f(\pa\bp \bar g) = (\Delta f, g)_{L^2(\rho)}.$$
It is more accurate to write $([f],[g])_\nabla$ in the left-hand side. 
For convenience, we sometimes refer to the elements in $\WW$ as functions and denote them by $f$ rather than $[f].$
Given a metric $\rho,$ we define the space $\WW(\rho)$ as the completion of the space of functions $f$ on $M$ 
satisfying the neutrality condition
$$\int f\rho=0$$
with respect to the scalar product 
$$(f,g)_\nabla = (\Delta f,g)_{L^2(\rho)}.$$

Clearly, the Hilbert spaces $\WW(\rho)$ and $\WW(M)$ are canonically equivalent using the map 
$$f\mapsto [f].$$
For $[f_0] \in \WW(M),$ there is $f \in \WW(\rho)$ such that $[f]=[f_0].$ Indeed, 
$$f = f_0 - \int f_0\rho.$$

Let $0 = \lambda_0 < \lambda_1 \le \cdots $ be the eigenvalues of $\Delta_\rho$ and $\{\phi_n\}_{n\ge0}$ be orthonormal (real) eigenfunctions,
$$\pa\bp\phi_n = -\frac12\lambda_n\phi_n\rho.$$
It is well known that 
$$\Big\{\frac{\phi_n}{\sqrt{\lambda_n}}\Big\}_{n\ge1}$$
forms an orthonormal basis for $\WW(\rho).$

\subsubsec{Dual spaces} 
Let us define the spaces $\WW'\equiv\WW'(M)$ and $\WW'(\rho)$ by 
$$\WW'(M) = \{\mu \in \DD^{(1,1)}(M): \|\mu\|_\Delta^2 < \infty, \int \mu = 0\},$$
and 
$$\WW'(\rho) = \{[\mu] : \mu \in \WW'(M)\}, \qquad [\mu] = \mu + \C\rho, $$
both with the norm
$$\|\mu\|_\Delta^2 = \iint R(z,w) \mu(z) \overline{\mu(w)}.$$
Then $\WW'(M)$ is the dual space of $\WW(M)$ and $\WW'(\rho)$ is the dual space of $\WW(\rho)$ with respect to the pairings
$$([f],\mu) \mapsto \int f \mu, \quad (f,[\mu]) \mapsto \int f \mu,$$ 
respectively.

\begin{prop}
We have isomorphisms
$-2\pa\bp: \WW \to \WW'$
and 
$\rho\Delta: \WW(\rho) \to \WW'(\rho).$
\end{prop}
\begin{proof}
The first part is obvious since $([\phi_m], \lambda_n\phi_n\rho) = \langle \phi_m,\phi_n \rangle_\nabla.$
For the second part, we observe that an orthonormal set $\{\sqrt{\lambda_n}[\phi_n\rho]\}_{n\ge1}$ in the space $\WW'(\rho)$
is the dual basis of 
$\{\phi_n/\sqrt{\lambda_n}\}_{n\ge1}.$
Proposition now follows since $\rho\Delta ( \phi_n/\sqrt{\lambda_n} ) = \sqrt{\lambda_n}[\phi_n\rho].$
\end{proof}
We remark that these two isomorphisms are actually the same up to canonical identifications.

\begin{cor} \label{cor: EW'}
The multiplication by $\sqrt{2\pi}$ is a unitary operator
$\sqrt{2\pi}:\EE\to\WW'.$
\end{cor}
\begin{proof} 
We need to show that 
$2\pi \|\mu\|_\Delta^2 = \|\mu\|_\EE^2,$
or 
$$2\pi\iint R(z,p) \,\mu(z)\overline{\mu(p)} = 2\iint G_{p,q}(z) \,\mu(z)\overline{\mu(p)}$$ 
for all $q, $ equivalently, 
$$\iint R(z,q) \,\mu(z)\overline{\mu(p)} = 0,$$
which is clear from the neutrality condition of $\mu:$ $\int\mu = 0.$ 
\end{proof}
As an alternative statement of the previous corollary, 
$-\sqrt{2/\pi}\,\pa\bp:\WW\to\EE$
is a unitary operator.

\subsection{Bosonic field and Gaussian fields}  \label{ss: Gaussian fields} 
We now define the \emph{bosonic field} $\Phi_\rho$ associated with the metric $\rho$ by \index{bosonic field associated with the metric}
$$\Phi_\rho(z) = \int \Phi(z,\zeta) \rho(\zeta).$$
Then the Fock space functionals $\Phi_\rho(z)$ can be viewed as ``generalized" elements $\Phi(\delta_z - \rho) $ of Fock space.
Clearly, we have 
$$\Phi(z,w) = \Phi_\rho(z)-\Phi_\rho(w).$$

We now compute the 2-point correlation function of $\Phi_\rho$ in terms of the resolvent kernel. 

\begin{prop} \label{prop: E Phi rho}
We have 
$$\E\,\Phi_\rho(z)\Phi_\rho(w) = 2\pi R(z,w). $$
\end{prop}
\begin{proof}
By the correlation \eqref{eq: correlator of Phi} of Gaussian free field and Corollary~\ref{cor: RG}, we have 
$$\E\,\Phi_\rho(z)\Phi_\rho(w) = 2\pi \iint \big(R_{z,\zeta}(w) - R_{z,\zeta}(\eta)\big)\,\rho(\zeta)\rho(\eta).$$
Proposition now follows from \eqref{eq: Rrho} and the last remark in Subsection~\ref{ss: spectral theory}.
\end{proof}

A \emph{Gaussian field} indexed by a Hilbert space $\HH$ is an isometry \index{Gaussian field indexed by a Hilbert space}
$$\HH~\to~ L^2(\Omega, \P)$$
such that the image consists of centered (i.e. mean zero) Gaussian variables; here $(\Omega,\P)$ is some probability space.

\begin{prop} The map 
$$\frac1{\sqrt{2\pi}}\,\Phi:\WW'\to L^2(\Omega, \P)$$
is a Gaussian field indexed by $\WW'.$
\end{prop}
\begin{proof}
By definition, the Gaussian free field $\Phi$ on $M$ is a Gaussian field indexed by the energy space $\EE(M).$
Proposition now follows from Corollary~\ref{cor: EW'}.
\end{proof}

\begin{prop} The map 
$$\frac1{\sqrt{2\pi}}\,\Phi_\rho:\WW'(\rho)\to L^2(\Omega, \P)$$
 is a Gaussian field indexed by $\WW'(\rho).$
\end{prop}
\begin{proof}
We need to show that 
$$\frac1{2\pi}\, \E\, \Phi_\rho(\mu_1) \Phi_\rho(\mu_2) = (\mu_1,\bar\mu_2)_\Delta.$$
By Proposition~\ref{prop: E Phi rho}, the expectation in the left-hand side is
$$\int\E\, \Phi_\rho(z) \Phi_\rho(w) \mu_1(z) \mu_2(w) = 2\pi \int R(z,w) \mu_1(z)\mu_2(w),$$
which completes the proof.
\end{proof}

\begin{cor} \label{cor: Gaussian fields}
We have
$$ \Phi_\rho(z) = \sqrt{2\pi} \sum_{n=1}^\infty \frac1{\sqrt{\lambda_n}} \chi_n(\omega) \phi_n(z),$$
where $\chi_n$'s are independent standard normal random variables.
\end{cor}

\begin{rmk*}
The previous corollary can be obtained by the following general construction.
Suppose $\HH$ and $\HH'$ are dual Hilbert spaces with respect to the pairing $\int h\mu$, e.g. $\WW$ and $\WW'.$
Let $\{h_n\}$ be an orthonormal basis for $\HH,$ and $\{\mu_n\}$ be its dual basis for $\HH'.$ 
Then
$$\Psi = \sum \chi_n(\omega) h_n $$
is a Gaussian field indexed by $\HH'.$ 
Indeed, we have 
$$\Psi(\mu) = \sum \Big(\int h_n \mu \Big)\chi_n(\omega),$$
so $\Psi$ maps $\mu_n$ to $\chi_n$ and $\Psi$ is an isometry $\HH'\to L^2(\Omega,\P).$
\end{rmk*}

\subsection{Dirichlet action and its partition function} \label{ss: Dirichlet action} 

In this subsection we explain that the bosonic field $\Phi_\rho$ can be obtained from the Dirichlet action. 
This action determines the distributions of random functions $\psi_n$ on some finite dimensional subspace of $\WW(\rho).$ 
We remark that a choice of the localization $\psi_n \perp 1$ resolves the divergence in the partition function.
We show that a sequence of random functions $\psi_n$ converges to $\Phi_\rho.$ 
In Section~\ref{s: Liouville action} we discuss the Liouville action and obtain the background charge modification of the Gaussian free field by modifying the method developed below. 

\subsubsec{Dirichlet action}
Let us define the \emph{Dirichlet action} \index{Dirichlet action} $\DD$ on $\WW(\rho)$ as 
$$\DD(\psi) = \frac1{4\pi} (\psi,\Delta \psi)_{L^2(\rho)} = \frac1{4\pi}(\psi,\psi)_\nabla,\quad \psi\in\WW(\rho).$$
Let $\lambda_k$ be the eigenvalues and $e_k$ the normalized eigenfunctions of the positive Laplace operator $\Delta_\rho,$ with $\lambda_0=0$ and $e_0=1.$
Recall that $\{e_n/\sqrt{\lambda_n}\}_{n\ge1}$ forms an orthonormal basis for $\WW(\rho).$
For a random function $\psi_n = \sum_{j=1}^n a_j e_j,$ the function 
$$\ee^{-\DD(\psi_n)} = \prod_{j=1}^n \ee^{-\frac{\lambda_ja_j^2}{4\pi}}$$
gives rise to the joint distribution measure $\mu_n$ for random variables $a_1,\cdots,a_n:$
$$\dd\mu_n = \frac1{Z_n}\prod_{j=1}^n \ee^{-\frac{\lambda_j x_j^2}{4\pi}}\dd x_j,\quad Z_n=\prod_{j=1}^n\int_{-\infty}^\infty \ee^{-\frac{\lambda_j x_j^2}{4\pi}}\,\dd x_j.$$
Then $a_j$'s are independent centered normal random variables with variances $2\pi/\lambda_j.$
It follows from Corollary~\ref{cor: Gaussian fields} that 
$$\psi_n= \sqrt{2\pi}\sum_{j=1}^n\frac1{\sqrt{\lambda_j}}\,\chi_je_j\to\Phi_\rho=\sqrt{2\pi}\sum_{j=1}^\infty \frac1{\sqrt{\lambda_j}}\,\chi_je_j,$$
where $\chi_j$'s are independent standard normal random variables.

\subsubsec{Partition function of Dirichlet action} \index{partition function! of Dirichlet action}
By means of a simple computation,
$$Z_n = \prod_{j=1}^n\sqrt{\frac{2\pi}{\lambda_j}}.$$ 
In general, $Z_n$ diverges as $n\to\infty.$
To define a partition function, this product should be properly regularized. 
It can be achieved by the so-called $\zeta$-function regularization.
We discuss this method in the genus one case below based on the materials in \cite[Appendix~A~1.1]{Sarnak}.

We consider a torus $\T_\Lambda, \Lambda = \Z + \tau Z, \Im\,\tau > 0$ with a flat metric $\rho = 1/\Im\,\tau$.
The eigenfunctions of the Laplacian $\Delta_\rho$ associated with this metric $\rho$ are
$$\varphi_w(z) = e^{2\pi i \langle z, w\rangle} $$ 
for some $w \in \Lambda^* := \{w\in\C\,|\,\langle z,w \rangle \in \Z \textrm{ for all } z \in \Lambda\}.$
It is easy to see that
$$ \Lambda^* = \frac{i}{\Im\,\tau}\Lambda.$$
(Indeed, due to the requirement $\langle 1, w\rangle \in \Z,$ we have  $ u = \Re\,w \in \Z.$ Write $u =-n$ for some $n\in \Z.$ 
Let $v = \Im\,w.$ On the other hand, the requirement $\langle \tau, w\rangle \in \Z$ leads to $-n\, \Re\,\tau + v\, \Im\,\tau = m$ for some $m\in\Z$ or $v = (m + n \Re\,\tau)/\Im\,\tau.$  Thus we have 
$$w = u + iv = -n +  i\frac{m + n \Re\,\tau}{\Im\,\tau} = i\,\frac{m + n\tau}{\Im\,\tau}.$$
Conversely, for $m,n,m',n'\in\Z,$ we have 
$$\frac1{\Im\,\tau}\langle m + n\tau,  i(m' + n' \tau)\rangle = \frac1{\Im\,\tau}\,\Im\big((m + n\tau)(m' + n' \bar\tau)\big) = m'n-mn' \in \Z.)$$

Let $w_{m,n} = i(m + n\tau)/\Im\,\tau$ and $\phi_{m,n}:=\varphi_{w_{m,n}}.$ 
Then $\phi_{m,n}$ form the eigenfunctions of $\Delta_\rho$ 
$$\pa\bp \phi_{m,n} = -\frac{\pi^2 |m+n\tau|^2}{(\Im\,\tau)^2}\phi_{m,n}  \Big(= -\frac12 \lambda_{m,n} \phi_{m,n} \rho\Big)$$
with eigenvalues
$$\lambda_{m,n} = \frac{2\pi^2|m+n\tau|^2}{\Im\,\tau}.$$
Thus we have 
$$\sqrt{\frac{2\pi}{\lambda_{m,n}}} = \frac{\sqrt{\Im\,\tau}}{\sqrt\pi\,|m+n\tau|}.$$

We introduce the normalized Eisenstein series $E^*(z,s)$ as
$$E^*(z,s) = \frac12 \pi^{-s}\Gamma(s){\sum_{m,n}}'\frac{y^s}{|m+nz|^{2s}},\qquad y = \Im\,z,$$
where $\sum_{m,n}'$ means $\sum_{(m,n)\in\Z^2\setminus\{(0,0)\}}.$ 
It is well known (e.g. see \cite[Appendix~A]{FGKP} or \cite[pp. 214~--~216]{CKM04}) that $E^*(z,s)$ has the Fourier expansion
\begin{align*}
E^*(z,s) &= \pi^{-s}\Gamma(s)\zeta(2s)y^s + \pi^{s-1}\Gamma(1-s) \zeta(2-2s) y^{1-s} \\
&+ 4\sqrt{y}\sum_{n=1}^\infty n^{s-1/2} \sigma_{1-2s}(n)K_{s-1/2}(2\pi n y)\cos(2\pi nx),
\end{align*}
where $x = \Re\,z,$ $\sigma_s(n) = \sum_{d|n} d^s$ and $K$ is the modified Bessel function of the second kind 
$$K_\nu(y) = \int_0^\infty  e^{-y\cosh t} \cosh(\nu t) \,\dd t.$$

Let 
$$\zeta_{\Delta_\rho}(s) = {\sum_{m,n}}' \frac{y^s}{|m+n\tau|^{2s}},\qquad y = \Im\,\tau.$$
Then it follows from the Fourier expansion of $E^*(z,s)$ that  
\begin{align*}
\zeta_{\Delta_\rho}(s)  &=2\zeta(2s)y^s + 2\pi^{2s-1}\frac{\Gamma(1-s)}{\Gamma(s)} \zeta(2-2s) y^{1-s} \\
&+ \frac{8\pi^s\sqrt{y}}{\Gamma(s)}\sum_{n=1}^\infty n^{s-1/2} \sigma_{1-2s}(n)K_{s-1/2}(2\pi n y)\cos(2\pi nx),
\end{align*}
where $x = \Re\,\tau.$
We now define the partition function $Z$ of the Dirichlet action by
\begin{equation}\label{eq: defZ4T}
Z := {\prod}'_\zeta \frac{\sqrt{\Im\,\tau}}{\sqrt\pi\,|m+n\tau|} \quad \textrm{ or } \quad Z:= \pi^{-\frac12 \zeta_{\Delta_\rho}(0) }\ee^{\frac12\zeta_{\Delta_\rho}'(0)}.
\end{equation}
See \cite{QHS93} for the definition and some basic properties of the $\zeta$-regularized products. 

It follows from $\zeta(0) =-1/2$ and $\Gamma(0) = \infty$ that $\zeta_{\Delta_\rho}(0) = -1.$
Using
$$\frac d{ds}\Big|_{s=0}\frac1{\Gamma(s)} = 1,$$
we compute
\begin{align*}
\zeta_{\Delta_\rho}'(0)  &= 4\zeta'(0) +2\zeta(0) \log y + 2\pi^{-1} \zeta(2) y\\
&+ 8\sqrt{y}\sum_{n=1}^\infty n^{-1/2} \sigma_{1}(n)K_{-1/2}(2\pi n y)\cos(2\pi nx).
\end{align*}
We use the well-known formulas
$$\zeta'(0)=-\frac12 \log(2\pi), \quad \zeta(0) = -\frac12, \quad \zeta(2) = \frac{\pi^2}6, \quad K_{-1/2}(y) = \sqrt{\frac\pi{2y}}\ee^{-y}$$
to obtain 
$$\zeta_{\Delta_\rho}'(0)  = -2 \log(2\pi)- \log y + \frac\pi3  y + 4\,\Re\sum_{n=1}^\infty \frac{\ee^{2\pi i n z}}n \sum_{d|n} d.$$

We now express the partition function $Z$ in terms of the Dedekind $\eta$-function
$$\eta(z) = \ee^{\pi i z/12}\prod_{n=1}^\infty (1-\ee^{2\pi i nz}).$$
It follows from Taylor series expansion of $\log(1-x)$ that
$$\log|\eta(z)| = -\frac\pi{12} y- \Re\,\sum_{n=1}^\infty \sum_{m=1}^\infty \frac{\ee^{2\pi i mn z}}m = -\frac\pi{12} y- \Re\sum_{l=1}^\infty\frac{\ee^{2\pi i l z}}l \sum_{n|l} n.$$
Thus we have 
$$ \zeta_{\Delta_\rho}'(0) = -2 \log(2\pi)- \log \Im\,\tau -4 \log|\eta(\tau)|.$$
By definition \eqref{eq: defZ4T}, the above formula for $\zeta_{\Delta_\rho}'(0),$ and   $\zeta_{\Delta_\rho}(0) = -1,$ we derive 
$$Z =  \frac1{2\sqrt\pi}\,\frac1{\sqrt{\Im\,\tau}|\eta(\tau)|^2}.$$
We emphasize that $\sqrt{\Im\,\tau}|\eta(\tau)|^2$ is modular invariant.

\section{Currents fields and Virasoro fields} 

One can think of the current fields $J = \pa \Phi, \bar J = \bp \Phi$ as Gaussian distributional fields. 
For example, $J(f) = -\Phi(\pa f)$ for $f\in C^\infty(M),$ so $J$ is a map from $C^\infty(M)$ to $L^2(\Omega,\P).$
In this section we define the current fields $J$ and $\bar J$ as single-variable Fock space fields such that $J(z) = \pa_z \Phi(z,z_0), \bar J(z) = \bp_z \Phi(z,z_0)$ within correlations. 
It is noteworthy that some correlations can be identified with fundamental object in function theory. 
For example, the 2-point function $-\E\,J(\zeta)J(z)$ is related to Weyl's form, the fundamental (normalized) bi-differential.
A $(1,1)$-differential $\E\,J(\zeta) \overline{J(z)}$ is represented by a matrix $(\Im\,\PM)^{-1},$ where $\PM$ is the period matrix. 

We introduce the Virasoro field $T$ in terms of the operator product expansion of the current field $J$ and itself. In Theorem~\ref{SET} we show that the Gaussian free field $\Phi$ has a stress energy tensor (with central charge $1$) which represents the Lie derivative operator within correlations of fields in a whole OPE family $\FF.$
(In Section~\ref{s: OPE EXP} we define the OPE products of two Fock space fields as the coefficients of their operator product expansion.) 
The proof of Theorem~\ref{SET} will be given for all background charge modifications (see Section~\ref{s: GFF modifications}) of $\FF.$
It is a fundamental fact that fields in $\FF$ have the same stress tensor (and thus the same Virasoro field) as $\Phi$ does.

\subsection{Current fields}
We define the \emph{current field} \index{current field} $J$ by 
$$J(z) = \pa_z \Phi(z,z_0).$$
Then $J$ is a well-defined single-variable Fock space field. 
The above definition does not depend on the choice of $z_0.$ 
Furthermore, it is single-valued and holomorphic (in the sense that the correlation $z\mapsto \E\, J(z)\XX$ is so on $M\sm S(\XX)$ for any string $\XX = \Phi(\zeta_1,z_1) \cdots \Phi(\zeta_n,z_n)$).
To see $J$ is holomorphic, it suffices to check that $\E\,J(p)\Phi(\zeta,z)$ is holomorphic on $M\sm\{\zeta,z\}$
Indeed, we have
\begin{align*} 
\E\,J(p)\Phi(\zeta,z) &= 2\pa_p (G_{p,q}(\zeta)-G_{p,q}(z))=2\pi\pa_p (R_p(\zeta)-R_p(z)) \\&= 2\pi\pa_p R_{\zeta,z}(p) = 2\pa_p G_{\zeta,z}(p). 
\end{align*}
It is obvious that $J$ is 1-differential. 
We also define $\bar J$ by 
$$\bar J(z) = \bp_z \Phi(z,z_0).$$
\begin{rmks*}
\renewcommand{\theenumi}{\alph{enumi}}
{\setlength{\leftmargini}{1.8em}
\begin{enumerate}
\item
In a sense, $J = \pa\Phi$ where $\Phi$ is a formal 1-point field with the formal correlation
$\E\,\Phi(p)\Phi(z) = 2G_p(z).$
\item For fixed $p$ and $q,$ $\pa G_{p,q}$ is a meromorphic 1-differential of third type with two simple poles $p,q.$
Its residue is $-1$ at $p$ and $1$ at $q.$ 
Furthermore, $\pa G_{p,q}$ has all imaginary periods.
\end{enumerate}
}
\end{rmks*}

In terms of harmonic function 
$$u_{\zeta_0,z_0}(\zeta,z) := G_{\zeta,\zeta_0}(z)-G_{\zeta,\zeta_0}(z_0) +\log|\zeta-z| + \log|\zeta_0-z_0|,$$
we have
\begin{align} \label{eq: EJJ}
\E\,J(\zeta) J(z) &= 2\pa_\zeta\pa_z G_{\zeta,\zeta_0}(z)
= 2\pa_\zeta\pa_z(- \log|\zeta-z|+u_{\zeta_0,z_0}(\zeta,z) ) \\
&= -\frac1{(\zeta-z)^2} - \frac16 S(\zeta,z), \nonumber
\end{align}
where $S(\zeta,z) := -12 \pa_\zeta\pa_z u_{\zeta_0,z_0}(\zeta,z)$ does not depend on $\zeta_0,z_0.$
If we fix $p$ and a chart at $p,$ then $\E\,J(z)J(p)$ is a meromorphic 1-differential of second type with a double pole at $p.$

\subsubsec{The genus zero case} 
On the Riemann sphere $\wh\C$ we have 
$$\E\,\Phi(\zeta,\zeta _0)\,\Phi(z,z_0) = 2\log\left|\frac{(\zeta _0-z)(\zeta-z_0)}{(\zeta-z)(\zeta _0-z_0)}\right|.$$
As a scalar, conformal invariance is manifest in the cross-ratio $\lambda(\zeta,\zeta _0,z,z_0)$ in the right-hand side.
Differentiating the above correlation function, we have 
$$\E\,J(\zeta)\Phi(z,z_0) = -\frac1{\zeta-z} +\frac1{\zeta-z_0}, \qquad \E\,J(\zeta)J(z) = -\frac1{(\zeta-z)^2} $$
and we can say the holomorphic sector $J$ and anti-holomorphic sector $\bar J$ is independent in the sense that 
\begin{equation} \label{eq: J bar J}
\E\,J(\zeta)\overline{J(z)} = 0.
\end{equation}
If $\zeta =z,$ the last formula is understood in the OPE sense.

\subsubsec{The genus one case}
We first recall the definition of \emph{Weierstrass $\wp$-function} \index{Weierstrass $\wp$-function} and its some basic properties.
Weierstrass $\wp$-function
$$\wp(z) \equiv \wp(z|\tau):= \frac1{z^2} + {\sum_{m,n}}' \Big(\frac1{(z+m+n\tau)^2}-\frac1{(m+n\tau)^2}\Big)$$
is an elliptic function with periods $1$ and $\tau.$
Here, ${\sum_{m,n}}'$ means $\sum_{(m,n)\in\Z^2\setminus\{(0,0)\}}.$ 
It has the Laurent series expansion around the only pole $0$ of $\wp:$
\begin{equation}\label{eq: wp expansion}
\wp(z) = \frac1{z^2} + \frac1{20} g_2 z^2 + O(z^4),
\end{equation}
where
$$g_2 = 60\,{\sum_{m,n}}' \frac1{(m+n\tau)^4}.$$
Weierstrass $\wp$-function is related to the Jacobi $\theta$-function as 
\begin{equation} \label{eq: wp theta}
\wp(z) = -\frac{\pa^2}{\pa z^2}\log \theta(z) +\frac13 \frac{\theta'''(0)}{\theta'(0)}
\end{equation}

In the identity chart of torus $\T_\Lambda,$ we have 
\begin{equation} \label{eq: EJPhi1}
\E\,J(z)\Phi(p,q) = -\frac{\theta'(z-p)}{\theta(z-p)} + \frac{\theta'(z-q)}{\theta(z-q)} + 2\pi i \,\frac{\Im(p-q)}{\Im\,\tau}
\end{equation}
and
$$\E\,J(\zeta)J(z) = \pa_z \Big(-\frac{\theta'(\zeta-z)}{\theta(\zeta-z)} + \frac\pi{\Im\,\tau}z\Big) =- \wp(\zeta-z) +\frac13 \frac{\theta'''(0)}{\theta'(0)}+\frac\pi{\Im\,\tau}.$$
Here we use \eqref{eq: wp theta}.
Taking $\bp_z$-derivative of $\E\,J(\zeta)\Phi(z,z_0),$  we have 
\begin{equation} \label{eq: EJbarJ torus}
\E\,J(\zeta)\overline{J(z)} = -\frac\pi{\Im\,\tau}.
\end{equation}
It follows from the Laurent series expansion \eqref{eq: wp expansion} of $\wp$ that 
\begin{equation} \label{eq: OPE JJ torus}
\E\,J(\zeta)J(z)= -\frac1{(\zeta-z)^2} +\frac13 \frac{\theta'''(0)}{\theta'(0)}+\frac\pi{\Im\,\tau} + o(1)
\end{equation}
as $\zeta\to z.$

\begin{eg*}
In Subsection~\ref{ss: Modular invariance of GFF} we discuss the modular invariance of the Gaussian free field. 
In this example, we use the modular invariance of the current to derive the transformation formula for Weierstrass $\wp$-function and $\eta_1.$
Differentiating \eqref{eq: modular invariance GFF}, we have
$$\frac1{(c\tau+d)^2}\,\E\, J_{_{\T_{\Lambda'}}}(\xi')J_{_{\T_{\Lambda'}}}(z') =  \E\, J_{_{\T_{\Lambda}}}(\xi)J_{_{\T_{\Lambda}}}(z).$$
In terms of two generators for the modular group $\mathrm{PSL}(2,\Z):=\mathrm{SL}(2,\Z)/\{\pm I\}$, 
$$\wp(z|\tau+1) +2\eta_1(\tau+1) = \wp(z|\tau) +2\eta_1(\tau)$$
and
\begin{equation}\label{eq: wp modular invariance}
\frac1{\tau^2}\Big(\wp\Big(\frac z\tau \Big|-\frac1\tau\Big) + 2\eta_1\Big(\!-\frac1\tau\Big) \Big)= \wp(z|\tau)+2\eta_1(\tau)
+\frac\pi{\tau^2 \Im\,(-1/\tau)} - \frac\pi{\Im\,\tau}.
\end{equation}
Comparing the coefficients of the expansions around the origin, we have 
$$\frac1{\tau^2}\wp\Big(\frac z\tau \Big|-\frac1\tau\Big)= \wp(z|\tau)$$
and
$$\frac1{\tau^2}\eta_1\Big(\!-\frac1\tau\Big) = \eta_1(\tau) - \frac{\pi i}{\tau}.$$
Here, we use
\begin{align*}
\frac\pi{2\tau^2 \,\Im\,(-1/\tau)} - \frac\pi{2\,\Im\,\tau} &=  \frac{\pi|\tau|^2}{2\tau^2 \,\Im\,\tau} - \frac\pi{2\,\Im\,\tau}=-\frac{\pi i}{\tau}.
\end{align*}
Of course, $\wp(z|\tau+1) = \wp(z|\tau)$ and $\eta_1(\tau+1) = -\eta_1(\tau).$
\end{eg*}

\subsubsec{The higher genus case}
We have 
\begin{align*}
\E\,J(\zeta) \Phi(p,q) &= \sum_j \Big(-\frac{\pa_j\Theta}{\Theta}\big(\AA(\zeta)-\AA(p)-e\big)+\frac{\pa_j\Theta}{\Theta}\big(\AA(\zeta)-\AA(q)-e\big) \Big)\,\omega_j(\zeta)\\
&+ 2\pi i\langle (\Im\,\PM)^{-1} \vec\omega(\zeta), \Im\,\AA(p)-\Im\,\AA(q) \rangle
\end{align*}
and 
\begin{align*}
\E\,J(\zeta) J(z) &= \sum_{j,k} \Big(\frac{\pa_{jk}\Theta}{\Theta}-\frac{\pa_j\Theta\pa_k\Theta}{\Theta^2}\Big)\big(\AA(\zeta)-\AA(z)-e\big)\,\omega_j(\zeta)\,\omega_k(z)\\
&+ \pi \langle (\Im\,\PM)^{-1} \vec\omega(\zeta), \overline{\vec\omega(z)} \rangle,
\end{align*}
where $\langle \cdot, \cdot \rangle$ is the usual inner product in $\C^g.$

By definition, the \emph{fundamental normalized bi-differential} \index{fundamental normalized bi-differential} $\Omega(\zeta,z)$ is a differential in both $\zeta$ and $z$ with the following properties
\begin{enumerate} 
\renewcommand{\theenumi}{\alph{enumi}}
{\setlength{\leftmargini}{1.8em}
\setlength\itemsep{.5em}
\item $\Omega(\zeta,z) = \Omega(z,\zeta);$
\item for fixed $z,$ $\Omega(\zeta,z)$ is meromorphic in $\zeta$ with a sole double pole at $\zeta = z;$ 
\item there exists a Schwarzian form $S$ of order 1 such that in a chart $\phi,$
$$\Omega(\phi^{-1}\zeta,\phi^{-1}z) = \dfrac1{(\zeta-z)^2} + \dfrac16S(\phi^{-1}z) + o(1)$$ as $\zeta\to z;$
\item $\displaystyle\oint_{\zeta\in a_j} \Omega(\zeta,z)= 0$ for each $a_j$-cycle.
}
\end{enumerate}
Clearly, $-\E\,J(\zeta) J(z)- \pi \langle (\Im\,\PM)^{-1} \vec\omega(\zeta), \overline{\vec\omega(z)} \rangle$ is the fundamental normalized bi-differential $\Omega(\zeta,z).$ 
See Proposition~\ref{prop: T0} in the next subsection.

Taking $\bp_z$-derivative of $\E\,J(\zeta)\Phi(z,z_0),$ we have 
$$\E\,J(\zeta) \overline{J(z)} = -\pi \langle (\Im\,\PM)^{-1} \vec\omega(\zeta), \vec\omega(z) \rangle.$$
We remark that $\E\,J(\zeta) \overline{J(z)}$ is non-trivial unless the genus $g$ is zero and it is a canonical version of 
$-\pi(\Im\,\PM)^{-1}$ in the sense that it does not require the specific basis $\{a_j,b_j\}$ of cycles and the basis $\{\omega_j\}$ of holomorphic 1-differentials associated with $\{a_j,b_j\}.$

\subsection{Virasoro fields} \label{ss: T}
We define the Virasoro field $T$ by the operator product expansion
$$J(\zeta) J(z) = -\frac1{(\zeta-z)^2} -2\,T(z) + o(1)$$
as $\zeta\to z.$
More precisely, 
$$\E\,J(\zeta) J(z)\XX = -\frac1{(\zeta-z)^2}\,\E\,\XX -2\,\E\,T(z)\XX + o(1)$$
holds for every Fock space correlation functional $\XX$ on $M$ such that $z\notin S_\XX.$
Here, $J(\zeta) J(z)$ and $T(z)$ mean of course $J(\phi^{-1}\zeta) J(\phi^{-1}z)$ and $T(\phi^{-1}z),$ respectively according to our convention.
In Section~\ref{s: OPE EXP} we study operator product expansion in more details.

\begin{prop} \label{prop: T0}
The Virasoro field $T$ is a holomorphic Schwarzian form of order $\frac1{12}.$
We have
$$T = -\frac12 J\odot J + \E\,T, $$
and
$$\,\E\,T(z) = \lim_{\zeta\to z} \Big(-\pa_\zeta\pa_z G_{z,z_0}(\zeta) -\frac12 \frac1{(\zeta-z)^2}\Big).$$
\end{prop}
\begin{proof}
It follows from \eqref{eq: EJJ} that 
$$T = -\frac12 J\odot J + \frac1{12}S,$$
where $S(z) = S(z,z).$
We claim that $S=12\,\E\,T$ is a Schwarzian form of order $1.$
Indeed, for two overlapping charts $\phi,\ti\phi,$ and their transition map $h,$ it follows from \eqref{eq: EJJ} that 
$$ - \frac16 (S\|\phi)(\zeta,z) = -\frac{h'(\zeta)h'(z)}{(h(\zeta)-h(z))^2} +\frac1{(\zeta-z)^2}- \frac16 h'(\zeta)h'(z)(S\|\ti\phi)(h(\zeta),h(z)).$$
Taking the limits on both sides as $\zeta$ to $z,$ we have 
$$(S\|\phi)(z) = h'(z)^2 (S\|\ti\phi)(h(z)) + S_h(z).$$
The other properties are obvious.
\end{proof}

\begin{rmk*} A non-random Schwarzian form of order $\frac1{12}$ exists always and is unique up to a holomorphic quadratic differential. 
\end{rmk*}

\begin{eg*} 
In the genus zero case, there are no nontrivial quadratic differentials, so $\E\,T = 0$ in the identity chart of $\C$ is the only such form.
\end{eg*}

\begin{eg*} 
In the genus one case, the dimension of the space of holomorphic quadratic differentials is 1 and $\E\,T$ is constant in the $\T_\Lambda$-uniformization. 
It follows from \eqref{eq: OPE JJ torus} that
$$\E\,T = -\frac16 \frac{\theta'''(0)}{\theta'(0)}-\frac12\frac\pi{\Im\,\tau}$$
in the identity chart of $\T_\Lambda.$
\end{eg*}

\begin{eg*} In the higher genus case, we have 
$$\E\,T = \frac1{12} S_{\omega^{(0)}}+ \frac18 \Big(\frac{\omega^{(2)}}{\omega^{(1)}}\Big)^2-\frac16 \frac{\omega^{(3)}}{\omega^{(1)}}-\frac\pi2 \langle (\Im\,\PM)^{-1} \vec\omega, \overline{\vec\omega} \rangle,$$
where $S_{\omega^{(0)}}$ is the Schwarzian derivative of the scalar $ \omega^{(0)},$ 
and $\omega^{(j)}$ is the $j$-differential ($j=0,1,2,3$) given by 
\begin{align} \label{eq: omega}
\omega^{(0)} = \int \omega^{(1)}, \qquad 
\omega^{(1)} &= \sum_j \pa_j\Theta(e)\,\omega_j \qquad 
\omega^{(2)} = \sum_{j,k} \pa_j\pa_k\Theta(e)\,\omega_j\omega_k, \qquad\\
\omega^{(3)} &= \sum_{j,k,l} \pa_j\pa_k\pa_l\Theta(e)\,\omega_j\omega_k\omega_l. \nonumber
\end{align}

The Schwarzian form 
$$S_{\omega^{(0)}}+ \frac32 \Big(\frac{\omega^{(2)}}{\omega^{(1)}}\Big)^2-2 \frac{\omega^{(3)}}{\omega^{(1)}} $$
is called the \emph{Bergman projective connection}.
We remark its important properties that it is \emph{holomorphic} and independent of the point $e$ in theta divisor, see \cite[p.19]{Fay}. 
From the point of the view of conformal field theory, these properties follow automatically from the construction of $T.$ 

\end{eg*}

\begin{prop} \label{OPE4T}
As $\zeta\to z,$ we have 
\renewcommand{\theenumi}{\alph{enumi}}
{\setlength{\leftmargini}{1.8em}
\begin{enumerate}
\ms \item \label{item: OPE(T,Phi)}
$T(\zeta)\Phi(z,z_0)= \dfrac{J(z)}{\zeta-z} + O(1),$
\ms \item \label{item: OPE(T,J)}
$T(\zeta)J(z)= \dfrac{J(z)}{(\zeta-z)^2}+ \dfrac{\partial J(z)}{\zeta-z}+ O(1),$
\ms \item \label{item: OPE(T,T)}
$T(\zeta)T(z)= \dfrac{1/2}{(\zeta-z)^4}+\dfrac{2T(z)}{(\zeta-z)^2}+ \dfrac{\partial T(z)}{\zeta-z}+ O(1).$
\end{enumerate}
}
\end{prop}

\begin{proof}

The proof is similar to that in the simply connected situation, see \cite[Proposition~3.5]{KM13}.
For reader's convenience, we only explain \eqref{item: OPE(T,T)}.
Use Wick's formula to obtain
$$T(\zeta)T(z) = \mathrm{I} + \mathrm{II} + T^{\odot2}(z) + o(1), \qquad(\zeta\to z),$$
where the terms $\mathrm{I}$ and $\mathrm{II}$ come from 1 and 2 contractions, respectively. 
We compute 
\begin{align*} 
\mathrm{II} &= \frac12 (\E[J(\zeta)J(z)])^2 \\
&= \frac12\Big(-\frac1{(\zeta-z)^2}-2\,\E\,T(z) -(\zeta-z)\pa_z\E\,T(z)+ O(|\zeta-z|^2)\Big)^2 \\
&= \frac{1/2}{(\zeta-z)^4} + \frac{2\,\E\,T(z)}{(\zeta-z)^2}  + \frac{\pa_z\E\,T(z)}{\zeta-z} + O(1).
\end{align*}
Next, we find $\mathrm{I} = \E[J(\zeta)J(z)] J(\zeta)\odot J(z)$ as
\begin{gather*}
\mathrm{I} =\Big(-\frac1{(\zeta-z)^2} + O(1)\Big)(J\odot J(z) + (\pa J)\odot J(z)(\zeta-z) + O(|\zeta-z|^2))\\
= -\frac{J\odot J(z)}{(\zeta-z)^2} - \frac{\pa J\odot J(z)}{\zeta-z} + O(1). 
\end{gather*}
Combining the above two, we obtain \eqref{item: OPE(T,T)}.
\end{proof}

\section{OPE exponentials} \label{s: OPE EXP}

To exponentiate the Gaussian free field $\Phi,$ one can approximate $\Phi$ by genuine random variables and then take the limit of the properly normalized exponentials of these random variables. 
However, a normalization procedure is not unique. 
One can identify this limit with the OPE exponential of $\Phi$ if one requires further that it has the same stress tensor as $\Phi$ does. We study basic properties of stress tensors and Ward's equations in Section~\ref{s: Ward}.

\subsection{Operator product expansion} \label{ss: OPE}
Suppose the field $X$ is holomorphic. 
By definition, it means that for each Fock space correlation functional $\YY$ on $M,$ the correlation function $\E[X(\zeta)\YY]$ is holomorphic with respect to $\zeta$ in $M\setminus S_\YY.$ 
Then the \emph{operator product expansion} \index{operator product expansion} (OPE) is defined as a Laurent series expansion within correlations:
$$X(\zeta) Y(z) = \sum C_n(z)(\zeta-z)^n, \quad \zeta\to z.$$ 
Thus for each Fock space correlation functional  $\ZZ,$ the function $\zeta\mapsto \E[X(\zeta)Y(z)\ZZ]$ is holomorphic in a punctured neighborhood of $z$ and has a Laurent series expansion. 
We use the symbol $\sim$ for the singular part of the operator product expansion,
$$X(\zeta) Y(z) \sim \sum_{n<0} C_n(z)(\zeta-z)^n.$$
For example, it follows from \eqref{eq: EJJ} that we have the following singular OPE for $J$ and itself:
$$J(\zeta)J(z) \sim -\frac1{(\zeta-z)^2}.$$

In this case, we define $*_n$-product by $X*_nY:=C_n$ and write $*$ for $*_0.$
We call $X*Y$ the OPE product, or the OPE multiplication of $X$ and $Y.$
Unlike Wick's multiplication, the OPE multiplication is neither commutative nor associative, see example below.
We define the OPE powers of $X$ by 
$$X^{*n} = X*(X^{*(n-1)})$$
so that the multiplications are evaluated from right to left.

\begin{eg*} 
In Subsection~\ref{ss: T} we define the Virasoro field $T$ by the OPE multiplication
$$T := -\frac12 J*J = -\frac12 J\odot J + \frac1{12}S,$$ 
where $S(z) = S(z,z),$ $S(\zeta,z) = -12\pa_\zeta\pa_zu_{\zeta_0,z_0}(\zeta,z),$ and $u_{\zeta_0,z_0}(\zeta,z) := \log|\zeta-z| + \log|\zeta_0-z_0|+\frac12\E\,\Phi(\zeta,\zeta_0)\Phi(z,z_0).$ 

We now compute $J*J^{*2}$ and $J^{*2}*J:$
$$J*J^{*2} = J^{\odot3} - \frac12 S J,\qquad J^{*2}*J = J^{\odot3} - \frac12 S J - \pa^2J.$$ 
To see the last formula, we use Wick's formula:
\begin{align*}
J^{\odot 2}(&\zeta)J(z) - J^{\odot 2}(\zeta)\odot J(z) = 2\E[J(\zeta)J(z)]J(\zeta) \\
&= 
-2\Big(\frac1{(\zeta-z)^2}+ \frac16 S(z)\Big)\Big(J(z) + \pa J(z)(\zeta-z) + \frac12\pa^2 J(z)(\zeta-z)^2+\cdots\Big)\\
&=-\frac{2J(z)}{(\zeta-z)^2} -\frac{2\pa J(z)}{\zeta-z} - \frac13 S J(z)-\pa^2 J(z) + o(1)
\end{align*}
as $\zeta\to z.$ 
Thus $J*J^{*2}\ne J^{*2}*J.$
\end{eg*}

The operator product expansion of the Gaussian free field and itself is an asymptotic expansion of the correlation functionals $\Phi(\zeta,\zeta_0)\Phi(z,z_0)$ as $(\zeta,\zeta_0)\to (z,z_0),\zeta\ne z, \zeta_0\ne z_0:$
$$\Phi(\zeta,\zeta_0)\Phi(z,z_0) = \log\frac1{|\zeta-z|^2} +\log\frac1{|\zeta_0-z_0|^2} + 2 c(z,z_0) + \Phi^{\odot 2}(z,z_0) + o(1).$$
Here, $c(z,z_0)$ is the limit of $u_{\zeta_0,z_0}(\zeta,z)$ as $(\zeta,\zeta_0)\to (z,z_0).$
In Subsection~\ref{ss: c} we study this non-random field $c(z,z_0)$ in more details.
The above statement means (by definition) that
\begin{align*}\E[\Phi(\zeta,\zeta_0)\Phi(z,z_0)\XX] &= \Big(\log\frac1{|\zeta-z|^2} +\log\frac1{|\zeta_0-z_0|^2} + 2 c(z,z_0)\Big)\E[\XX] \\
&+ \E[\Phi^{\odot 2}(z,z_0)\XX] + o(1)
\end{align*}
holds for every Fock space correlation functional $\XX$ on $M$ such that $z,z_0 \notin S_\XX.$
We have 
$$\Phi^{*2}(z,z_0)\equiv\Phi*\Phi(z,z_0) = \Phi^{\odot2}(z,z_0) + 2c(z,z_0).$$
\begin{rmk*}
Suppose that $X(\zeta)X(z)$ has a logarithmic type singular operator product expansion as $\Phi(\zeta,\zeta_0)\Phi(z,z_0)$ does.
Then the OPE multiplication is (power-)associative in the subalgebra generated by the singleton $X,$ e.g. $\Phi^{*3}*\Phi = \Phi^{*2}*\Phi^{*2}.$ 
\end{rmk*}

\subsection{OPE powers and OPE exponentials}

Suppose $X$ is a centered (multi-variant) Gaussian field. e.g. $X = \Phi[\bfs\tau]$ with a divisor $\bfs\tau$ satisfying the neutrality condition $(\NC_0).$
Let $$c_X := \frac12\,\E[X*X].$$
Then by Wick's calculus, as $\bfs\zeta\to \bfs z,$
\begin{align*}
X(\bfs\zeta)X(\bfs z) &= \E[X(\bfs\zeta)X(\bfs z)] + X(\bfs\zeta)\odot X(\bfs z) \\
&= \textrm{singular terms} + 2\,c_X(\bfs z) + X\odot X(\bfs z) + o(1).
\end{align*}

For a polynomial $p(z) = \sum_{k=0}^n a_k z^k$ and a non-random field $f,$ we set
$$p^\odot(f X) = \sum_{k=0}^n a_k f^k X^{\odot k}.$$
Let us denote by $H_n$ the Hermite polynomials.
Also set $\wt H_n(z):= i^{-n}H_n(iz).$ 
Then 
$$\wt H_0(z) = 1, \quad \wt H_1(z) = z, \quad \wt H_2(z) = z^2+1, \quad \wt H_3(z) = z^3+3z.$$

\begin{prop} 
We have
\begin{equation} \label{eq: OPE^n X}
X^{* n} = (2c_X)^{n/2}\wt H_n^\odot\left(\frac{X}{\sqrt{2c_X}}\right).
\end{equation}
\end{prop}

\begin{proof} By Wick's formula, as $\bfs\zeta\to \bfs z,$
$$X(\bfs\zeta)X^{\odot k}(\bfs z)=k \E[X(\bfs\zeta)X(\bfs z) ]X^{\odot (k-1)}(\bfs z)+X^{\odot (k+1)}(\bfs z)+o(1).$$
From the above relation, we find
\begin{equation} \label{eq: X*X^odot k}
X*X^{\odot k}=2c_XkX^{\odot (k-1)}+X^{\odot (k+1)}.
\end{equation}
Assuming that \eqref{eq: OPE^n X} holds for $X^{* n}$ and using \eqref{eq: X*X^odot k} for $0\le k \le n,$
we have
$$X*(X^{*n}) = (2c_X)^{(n+1)/2} \left( (\wt H_n')^\odot\left(\frac{X}{\sqrt{2c_X}}\right) +\frac{X}{\sqrt{2c_X}}\odot \wt H_n^\odot\left(\frac{X}{\sqrt{2c_X}}\right)\right).$$
Proposition now follows from the recurrence relation for the Hermite polynomials, or 
$\wt H_{n+1}(z)= z \wt H_{n}(z) + \wt H'_{n}(z).$
\end{proof}

We define the OPE exponential of $X$ by 
$$\ee^{*\alpha X} = \sum_{n=0}^\infty \alpha^n \frac{X^{*n}}{n!}.$$
Here, $X$ is not necessarily centered.

\begin{lem} \label{lem: ef eX}
If $f$ is a non-random field, then 
$$\ee^{*(X+f)} = \ee^f \ee^{*X}.$$
\end{lem}

\begin{proof}
We first note that 
$$(X+f)^{*n} = \sum_{k=0}^n \binom n k f^{n-k}X^{*k}.$$
It can be shown by induction argument together with the fact that $X*(fY) = f(X*Y) (\ne (fX)*Y$ in general) and $f*X = X*f = fX$ for a non-random field $f.$
Lemma now follows from the basic combinatorics, 
 $$\ee^{*(X+f)} = \sum_{n=0}^\infty \frac{(X+f)^{*n}}{n!} =\sum_{n=0}^\infty \sum_{k=0}^n \frac{f^{n-k}}{(n-k)!}\frac{X^{*k}}{k!}= \sum_{m=0}^\infty \frac{f^{m}}{m!}\sum_{k=0}^\infty\frac{X^{*k}}{k!}=\ee^f \ee^{*X},$$
which completes the proof. 
\end{proof}

Let $C_X = \ee^{c_X}.$ 
\begin{prop} \label{prop: OPE exp}
For any complex number $\alpha,$ we have 
\begin{equation} \label{eq: exp(*Phi)}
\ee^{*\alpha X}=C_X^{\alpha^2}\,\ee^{\odot\alpha X}.
\end{equation}
\end{prop}

\begin{proof}
Using \eqref{eq: OPE^n X} and the generating function 
$$\ee^{tx-{t^2}/2}=\sum_{n=0}^\infty \frac{t^n}{n!}H_n(x), \quad \Big(\textrm{or } \ee^{tx+{t^2}/2}=\sum_{n=0}^\infty \frac{t^n}{n!}\wt H_n(x)\Big)$$
for the Hermite polynomials, we get
$$\ee^{*\alpha X} = \sum_{n=0}^\infty \frac{\alpha^n}{n!} X^{* n} = \sum_{n=0}^\infty \frac{\alpha^n(2c_X)^{n/2}}{n!} \wt H_n^\odot\left(\frac{X}{\sqrt{2c_X}}\right)=\ee^{\odot\alpha X} \ee^{c_X\alpha^2},$$
which completes the proof. 
\end{proof}

\begin{eg*}
We set $u(\zeta,z) = G_\zeta(z) + \log|\zeta-z|$ formally 
and define the formal field $c(z)$ by 
$$(c\|\phi)(z) =\lim_{\zeta\to z} u(\zeta,z).$$
\end{eg*}
We define a formal 1-point field $\Phi^{*2}\equiv\Phi*\Phi$ by 
$$\Phi*\Phi(z) := \Phi\odot\Phi(z) + 2c(z).$$
In the next subsection we study a well-defined non-random field 
$$c(z,z_0) = \frac12\,\E\,\Phi*\Phi(z,z_0).$$

For two divisors $\bfs\sigma=\sum \sigma_j\cdot z_j,\bfs\tau=\sum\tau_k\cdot\zeta_k$ satisfying the neutrality condition $(\NC_0)$ such that $\supp\,\bfs\sigma\cap\supp\,\bfs\tau = \emptyset,$ the correlation $\E\,\Phi[\bfs\sigma]\Phi[\bfs\tau]$ is well-defined and we compute it in Subsection~\ref{ss: GFF}.
If $\supp\,\bfs\sigma\cap\supp\,\bfs\tau \ne \emptyset,$ then we need a normalization procedure via operator product expansion. 
At each point $z$ in $\supp\,\bfs\sigma\,\cap\, \supp\,\bfs\tau,$ we replace the divergent term $\Phi(z)\Phi(z)$ in the expansion of $\Phi[\bfs\sigma]\Phi[\bfs\tau]$ by the 1-point formal field $\Phi*\Phi(z).$
We denote by $\Phi[\bfs\sigma]*\Phi[\bfs\tau]$ this renormalized field:
\begin{equation}\label{eq: Phi*}
\Phi[\bfs\sigma]*\Phi[\bfs\tau] = \sum_{z_j\ne\zeta_k} \sigma_j\tau_k \,\Phi(z_j)\Phi(\zeta_k) +\sum_{z_j=\zeta_k} \sigma_j\tau_k \,\Phi*\Phi(z_j).
\end{equation}
It is a well-defined Fock space field as long as both $\bfs\sigma$ and $\bfs\tau$ satisfy the neutrality condition $(\NC_0).$ 

We define the formal 1-point field $\VV^{(\sigma)}$ by 
\begin{equation} \label{eq: formal 1V}
\VV^{(\sigma)}(z)\equiv \ee^{*i\sigma\Phi(z)} = C(z)^{-\sigma^2} \ee^{\odot i\sigma\Phi(z)}, \quad C(z)=\ee^{c(z)}.
\end{equation}
This can be extended to the formal multi-vertex field $\VV[\bfs\sigma]$ and it is a well-defined Fock space field if $\bfs\sigma$ satisfies the neutrality condition $(\NC_0),$ see Subsection~\ref{ss: V}.

\subsection{Logarithmic conformal radius} \label{ss: c}
As a correlational field, the Gaussian free field $\Phi_D$ in a simply connected domain $D\subsetneq\C$ is a single variable field.
In this case, $c_D:=\frac12 \E\,\Phi_D*\Phi_D$ is the \emph{logarithm of conformal radius} \index{logarithm of conformal radius} $C_D.$
In terms of a conformal map $w$ from $D$ onto the upper half-plane $\H,$
$$c_D(z) = \log C_D(z), \quad C_D(z) = \Big|\frac{w(z)-\overline{w(z)}}{w'(z)}\Big|.$$ 
On the other hand, $c_D$ can be expressed in terms of the Green's function as
$$c_D(z) = u_D(z,z), \quad u_D(\zeta,z) = G_D(\zeta,z) + \log|\zeta-z|.$$ 
Motivated from the simply connected case, we set 
$$u_{\zeta_0,z_0}(\zeta,z) = G_{\zeta,\zeta_0}(z)-G_{\zeta,\zeta_0}(z_0) +\log|\zeta-z| + \log|\zeta_0-z_0|$$
and define the non-random field $c(z,z_0)$ by
$$(c\|\phi,\phi_0)(z,z_0) =\!\!\!\!\lim_{(\zeta,\zeta_0)\to (z,z_0)}\!\!\!\!u_{\zeta_0,z_0}(\zeta,z).$$
Here, $G$ is bipolar Green's function and $G_{\zeta,\zeta_0}(z)$ means $G_{\phi^{-1}\zeta,\phi_0^{-1}\zeta_0}(\phi^{-1}z)$ according to our convention. 

It is easy to see that the field is a PPS form, i.e.
$$(c\|\phi,\phi_0)(z,z_0) = (c\|\ti\phi,\ti\phi_0)(h(z),h_0(z_0)) - \log|h'(z)| - \log|h_0'(z_0)|,$$
where $h, (h_0)$ is the transition map between the overlapping charts $\phi,\ti\phi, (\phi_0,\ti\phi_0),$ respectively.
In terms of the formal field $c(z)$ in the previous example, we have 
\begin{equation}\label{eq: ccc}
c(z,z_0) = c(z)+c(z_0) - 2G_{z_0}(z).
\end{equation}
Of course, $c(z,z_0)$ does not depend on the choice of $G_\zeta(z).$

\subsubsec{The genus zero case}
One can choose $G_p(z) = -\log|z-p|$ on the Riemann sphere $\wh\C.$
With this choice, we have $c(z) \equiv 0$ and $c(z,z_0) = \log |z-z_0|^2.$
Thus in terms of a uniformization $w:M\to\C,$ we have 
\begin{equation} \label{eq: c0}
c(z,z_0) = \log\frac{|w(z)-w(z_0)|^2}{|w'(z)w'(z_0)|}.
\end{equation}

\subsubsec{The genus one case}
In the identity chart of torus $\T_\Lambda,$ one can choose
$$G_p(z) = \log\frac1{|\theta(z-p)|}-2\pi\,\frac{\Im\,p\,\Im\,z}{\Im\,\tau}.$$ 
With this choice, we have 
$$c(z) = \log\frac1{|\theta'(0)|}-2\pi\,\frac{(\Im\,z)^2}{\Im\,\tau}.$$
In terms of a uniformization $w:M\to\T_\Lambda,$ 
\begin{equation} \label{eq: c1formal}
c(z) = \log\frac1{|\theta'(0)w'(z)|}-2\pi\,\frac{(\Im\,w(z))^2}{\Im\,\tau}.
\end{equation}
Thus we have 
\begin{equation} \label{eq: c1}
c(z,z_0) = \log\frac{|\theta(w(z)-w(z_0))|^2} {|\theta'(0)^2w'(z)w'(z_0)|} -2\pi\,\frac{(\Im\,(w(z)-w(z_0)))^2}{\Im\,\tau}.
\end{equation}

\subsubsec{The higher genus case}
One can choose 
$$G_p(z) = \log \frac1{|\Theta(\AA(z) - \AA(p)-e)|}- 2\pi\,(\Im\,\PM)^{-1}\Im\,\AA(p)\cdot\Im\,\AA(z)$$ 
with a generic $e$ in the theta divisor. 
With this choice, we have 
$$c(z) = \log\frac1{|\omega(z)|}-2\pi\,(\Im\,\PM)^{-1}\Im\,\AA(z)\cdot\Im\,\AA(z), \quad \omega = \sum \pa_j\Theta(e)\,\omega_j.$$
From this and \eqref{eq: ccc}, we obtain
\begin{equation} \label{eq: c2}
c(z,z_0) = \log\frac{|\theta(z-z_0)|^2} {|\omega(z)\omega(z_0)|} -2\pi\,(\Im\,\PM)^{-1}\Im\,(\AA(z)-\AA(z_0))\cdot\Im\,(\AA(z)-\AA(z_0)).
\end{equation}

Suppose that a divisor $\bfs\sigma = \sum \sigma_j \cdot z_j $ satisfies the neutrality condition $(\NC_0).$
Let $c[\bfs\sigma ] = c_{\Phi[\bfs\sigma]}.$
Then we have
\begin{equation} \label{eq: cPhiSigma}
c[\bfs\sigma] = \sum_j \sigma_j^2 c(z_j) + \sum_{j<k} 2\sigma_j\sigma_k G_{z_k}(z_j).
\end{equation}
To see this, let $\bfs\sigma = \sum \sigma_j\cdot z_j$ and $\bfs\sigma' = \sum \sigma_j\cdot z_j'.$ 
As $\bfs z'\to\bfs z,$
\begin{align*}
\frac12 \E\,\Phi[\bfs\sigma' ]\Phi[\bfs\sigma] &= \frac12 \sum_{j,k} \sigma_j\sigma_k \,\E\,\Phi(z_j')\Phi(z_k) = \sum_j \sigma_j^2 \log\frac1{|z_j'-z_j|}\\
&+ \sum_j\sigma_j^2 c(z_j) + \sum_{j<k} \sigma_j\sigma_k \,\E\,\Phi(z_j')\Phi(z_k) + o(1).
\end{align*}

\subsection{Non-chiral vertex fields} \label{ss: V}

For a divisor $\bfs\sigma$ satisfying the neutrality condition $(\NC_0),$ we define the \emph{Coulomb gas correlation function}  \index{Coulomb gas correlation function}  $\CC[\bfs\sigma]$ by
$$\CC[\bfs\sigma]:=\ee^{-c[\bfs\sigma]}.$$
It is a differential of conformal dimension $(\lambda_j,\lambda_j)$ at $z_j,$ where $\lambda_j=\sigma_j^2/2.$
 
We define the (non-chiral) \emph{multi-vertex field} \index{multi-vertex field} $\VV[\bfs\sigma]$ by
$$\VV[\bfs\sigma] = \CC[\bfs\sigma] \VV^\odot[\bfs\sigma], \qquad \VV^\odot[\bfs\sigma]:=\ee^{\odot i\Phi[\bfs\sigma]}.$$ 

\begin{prop}
Suppose that both $\bfs\sigma$ and $\bfs\tau$ satisfy the neutrality condition $(\NC_0)$ and that they have disjoint supports.
Then the tensor product of $\VV[\bfs\sigma]$ and $\VV[\bfs\tau]$ is $\VV[\bfs\sigma+\bfs\tau]:$
$$\VV[\bfs\sigma] \VV[\bfs\tau] = \VV[\bfs\sigma+\bfs\tau].$$
\end{prop}
\begin{proof}
Since $\VV^\odot[\bfs\sigma] \VV^\odot[\bfs\tau] = \VV^\odot[\bfs\sigma+\bfs\tau]\ee^{-\E\,\Phi[\bfs\sigma]\Phi[\bfs\tau]},$
we need to check that 
$$\CC[\bfs\sigma]\CC[\bfs\tau] \ee^{-\E\,\Phi[\bfs\sigma]\Phi[\bfs\tau]} = \CC[\bfs\sigma+\bfs\tau].$$
Equivalently, 
\begin{equation} \label{eq: c addition}
c[\bfs\sigma+\bfs\tau] = c[\bfs\sigma]+c[\bfs\tau] + \E\,\Phi[\bfs\sigma]\Phi[\bfs\tau].
\end{equation}
It is obvious by \eqref{eq: cPhiSigma}. 
\end{proof}

We now have OPE nature of $\VV[\bfs\sigma].$
\begin{prop} Suppose that a divisor $\bfs\sigma$ satisfies the neutrality condition $(\NC_0).$ Then
$$\VV[\bfs\sigma] = \ee^{*i\Phi[\bfs\sigma]} \equiv \sum_{n=0}^\infty \frac{i^n\Phi^{*n}[\bfs\sigma]}{n!}.$$
\end{prop}
\begin{proof}
By Proposition~\ref{prop: OPE exp}, we have 
$$ \ee^{*i\Phi[\bfs\sigma]} = \ee^{-c[\bfs\sigma]} \, \VV^\odot[\bfs\sigma] = \CC[\bfs\sigma]\, \VV^\odot[\bfs\sigma] = \VV[\bfs\sigma],$$
which completes the proof.
\end{proof}

Recall the definition~\eqref{eq: formal 1V} of the formal 1-point field $\VV^{(\sigma)},$ 
$$\VV^{(\sigma)}(z)= \ee^{-\sigma^2c(z)} \ee^{\odot i \sigma\Phi(z)}.$$

\begin{prop} \label{prop: V=VV}
Suppose that a divisor $\bfs\sigma$ satisfies the neutrality condition $(\NC_0).$ Then we have 
$$\VV[\bfs\sigma] = \VV^{(\sigma_1)}(z_1) \cdots \VV^{(\sigma_n)}(z_n).$$
\end{prop}
\begin{proof}
By Wick's formula, we have 
\begin{align*}
\prod_j \VV^{(\sigma_j)}(z_j) &= \prod_j \ee^{-\sigma_j^2c(z_j)} \ee^{\odot i\sigma_j \Phi(z_j)} \\
&= \prod_j \ee^{-\sigma_j^2c(z_j)}\prod _{j<k} \ee^{-\sigma_j\sigma_k \E\,\Phi(z_j)\Phi(z_k)} \ee^{\odot i\sum \sigma_j \Phi(z_j)}.
\end{align*}
It follows from \eqref{eq: cPhiSigma} that
$$\prod_j \VV^{(\sigma_j)}(z_j) = \ee^{-c[\bfs\sigma]} \ee^{\odot i\Phi[\bfs\sigma]} = \CC[\bfs\sigma] \VV^\odot[\bfs\sigma]= \VV[\bfs\sigma],$$
which completes the proof.
\end{proof}

\subsubsec{The genus zero case}
In the identity chart of $\C$ and the usual chart $z\mapsto -1/z$ at infinity, the evaluation of $\CC[\bfs\sigma]$ is 
$$\prod_{j<k}|z_j-z_k|^{2\sigma_j\sigma_k},$$
where the product is taken over finite $z_j$'s and $z_k$'s, and as usual $0^0:=1.$

\begin{eg*} We define the (non-chiral) bi-vertex field $\VV^{(\sigma)}(z,z_0)$ by 
$$\VV^{(\sigma)}(z,z_0) = \VV[\bfs\sigma],\quad \bfs\sigma = \sigma \cdot z - \sigma \cdot z_0.$$ 
Then we have 
$$\VV^{(\sigma)}(z,z_0) = \Big|\frac{w'(z)w'(z_0)}{(w(z)-w(z_0))^2}\Big|^{\sigma^2} \ee^{\odot i\sigma \Phi(z,z_0)},$$
where $w:M\to\wh\C$ is a conformal map.
\end{eg*}

In terms of a uniformization $w:M\to \wh\C,$ we have $c(z) = -\log|w'(z)|$ and 
$$\VV^{(\sigma)}(z)= |w'(z)|^{\sigma^2} \ee^{\odot i \sigma\Phi(z)}.$$
Alternately, we construct $\VV^{(\sigma)}$ via the rooting procedure.
Fix any point $z_*\in M$ and define
$$\VV^{(\sigma)}_*(z) = \VV^{(\sigma)}(z,z_*).$$
In terms of a uniformization $w:(M,z_*) \mapsto (\wh\C,\infty),$ we have
$$\VV^{(\sigma)}_*(z)= |w'(z)|^{\sigma^2} \ee^{\odot i \sigma\Phi(z)},$$
where we require that $w(\zeta) \sim -1/(\zeta-z_*)$ as $\zeta\to z_*$ in a fixed chart at $z_*.$
However, it hides the fact that the rooted field $\VV^{(\sigma)}_*(z) = \VV^{(\sigma)}(z,z_*)$ is a differential of conformal dimension $(\sigma^2/2,\sigma^2/2)$ at a rooted point $z_*.$
The tensor product $\VV^{(\sigma_1)}_*(z_1) \VV^{(\sigma_2)}_*(z_2)$ has a covariance at a marked point $z_*:$ it is a a differential of conformal dimension $(\sigma^2/2,\sigma^2/2)$ at $z_*,$ where $\sigma = \sigma_1+\sigma_2.$

In principle, under the neutrality condition $(\NC_0),$ the (non-chiral) multi-vertex fields $\VV[\bfs\sigma]\, (\bfs\sigma = \sum \sigma_j\cdot z_j)$ can be represented as the tensor products of formal 1-point vertex fields $\VV^{(\sigma_j)}(z_j)$ by means of formal application of Wick's rule to the tensor products, see Proposition~\ref{prop: V=VV}.

\begin{eg*} We have 
$$\VV^{(\sigma)}(z_1,z_2) = \VV^{(\sigma)}(z_1)\VV^{(-\sigma)}(z_2) = \Big|\frac{w'(z_1)w'(z_2)}{(w(z_1)-w(z_2))^2}\Big|^{\sigma^2} \ee^{\odot i\sigma\Phi(z_1,z_2)}.$$
\end{eg*}

\subsubsec{The genus one case} 
If a divisor $\bfs\sigma$ satisfies the neutrality condition $(\NC_0),$ then by \eqref{eq: c1formal} and \eqref{eq: cPhiSigma}
$$c[\bfs\sigma] = \sum_j \sigma_j^2\log\frac1{|\theta'(0)w'_j|} + 2\sum_{j<k} \sigma_j\sigma_k\log\frac1{|\theta(w_j-w_k)|} - \frac{2\pi}{\Im\,\tau} \big(\Im\sum_j \sigma_j w_j \big)^2,$$
where $w_j = w(z_j)$ and $w_j' = w'(z_j),$ etc. 
The Coulomb gas correlation function $\CC[\bfs\sigma]$ is
\begin{align}\label{eq: C4torus}
\CC[\bfs\sigma] &= |\theta'(0)|^{\sum_j \sigma_j^2}
\prod_j |w_j'|^{\sigma_j^2} \prod_{j<k} |\theta(w_j-w_k)|^{2\sigma_j\sigma_k} \\
&\exp \Big(\frac{2\pi}{\Im\,\tau}
\big(\Im\sum_j \sigma_j w_j \big)^2
\Big).\nonumber
\end{align}

\subsubsec{The higher genus case}
Recall that 
$$c(z) = \log\frac1{|\omega(z)|}-2\pi\,(\Im\,\PM)^{-1}\Im\,\AA(z)\cdot\Im\,\AA(z), \quad \omega = \sum \pa_j\Theta(e)\,\omega_j.$$
For a divisor $\bfs\sigma = \sum \sigma_j \cdot z_j$ satisfying the neutrality condition $(\NC_0),$ we have 
\begin{align*}
c[\bfs\sigma] &=
\sum_j \sigma_j^2\log\frac1{|\omega(z_j)|}
+2\sum_{j<k} \sigma_j\sigma_k\log\frac1{|\theta(z_j-z_k)|} \\
&- 2\pi(\Im\,\PM)^{-1} \Im\,\AA[\bfs\sigma]\cdot\Im\,\AA[\bfs\sigma],
\end{align*}
where $\theta(z) = \Theta(\AA(z)-e)$ and $\AA[\bfs\sigma] = \sum_j \sigma_j \AA(z_j).$
The Coulomb gas correlation function $\CC[\bfs\sigma]$ is
\begin{equation} \label{eq: C4higher}
\CC[\bfs\sigma] =
\prod_j |\omega(z_j)|^{\sigma_j^2} 
\prod_{j<k} |\theta(z_j-z_k)|^{2\sigma_j\sigma_k} \exp \big(2\pi
(\Im\,\PM)^{-1} \Im\,\AA[\bfs\sigma]\cdot\Im\,\AA[\bfs\sigma]\big). 
\end{equation}

\section{Background charge modifications of the Gaussian free field } \label{s: GFF modifications} 
In this section we discuss central/background charge modifications of the Gaussian free field on a compact Riemann surface $M.$
These modifications are constructed by adding certain non-random pre-pre-Schwarzian forms (PPS forms) $\varphi$ to the Gaussian free field. 
The reason to consider PPS forms becomes clear in Sections~\ref{s: O} and \ref{s: Ward}. 
The modifications of Gaussian free field are parametrized by real parameter $b\in\mathbb{R}$ and the modified Virasoro fields $T$ are Schwarzian form of order $c/12,$ where $c$ is the central charge $c = 1-12b^2.$ 
As exponentials of a PPS form $\varphi$ are differentials, so are (properly normalized) exponentials of the modifications of the Gaussian free field. 

The PPS form $\varphi$ we consider is harmonic such that $\bfs\beta = i/\pi \pa\bp \varphi$ is a finite linear combination of $\delta$-measures. By a certain extension of Gauss-Bonnet formula to a flat conformal metric with conical singularities, it turns out that the total mass of $\bfs\beta$ depends both on the Euler characteristic $\chi(M)$ and on the modification parameter $b.$ 
It is known as the neutrality condition of the background charge $\bfs\beta$ in the physics literature.

\subsection{Holomorphic, harmonic, and simple PPS forms} \label{ss: Holomorphic PPS forms}
\subsubsec{Holomorphic PPS forms} 
By definition, a holomorphic non-random field $\varphi$ is a $\PPS(ib)$-form (on a given Riemann surface), a \emph{pre-pre-Schwarzian form} \index{pre-pre-Schwarzian form} \index{PPS form} of order $ib$ if the transformation law is
$$\varphi = \ti \varphi \circ h + ib \log h',$$
where $\varphi = (\varphi\,\|\,\phi)$ in a chart $\phi,$ $\ti\varphi = (\varphi\,\|\,\ti\phi)$ in a chart $\ti\phi,$ and $h$ is the transition map between two overlapping charts $\phi,\ti\phi.$
In this case, $\pa\varphi$ is a $\PS(ib)$ form (a \emph{pre-Schwarzian form} \index{pre-Schwarzian form} \index{PS form} of order $ib$), i.e. it has the transformation law 
$$\pa\varphi = \pa\ti \varphi \circ h + ib \frac{h''}{h'}.$$
For example, $\varphi = ib \log\omega$ is a holomorphic $\PPS(ib)$-form if $\omega$ is a holomorphic 1-differential.

Note that a PPS form $\varphi$ is defined as an assignment of $(\varphi\,\|\,\phi)$ to each local chart $\phi:U\to\phi U.$
Globally, we allow isolated singularities of logarithmic type
$$\varphi \sim c_j \log(z-z_j), \quad z\to z_j $$
in the sense that $\pa\varphi$ is meromorphic with only simple poles.
The class of forms such that $\pa\varphi$ is meromorphic may be of interest.
In general, a holomorphic PPS form is multivalued and its branch depends on local coordinates.

\subsubsec{Harmonic PPS forms} 
A non-random field $\varphi$ is a harmonic form if there is a representation
$\varphi = \varphi^+ + \varphi^-$
with holomorphic forms $\varphi^+$ and $\overline{\varphi^-},$ which are not necessarily of the
same order.

\begin{def*} A non-random field $\varphi$ is called a $\PPS(ib,ib)$ form (or a  \emph{PPS form of $\log|\cdot|$-type}) \index{PPS form!of log-type} if the
transformation law is
$$\varphi = \ti \varphi \circ h + ib \log |h'|^2,$$
and $\varphi$ is called a $\PPS(ib,-ib)$ form (or a \emph{PPS form of $\arg$-type}) \index{PPS form!of arg-type} if 
$$\varphi = \ti \varphi \circ h -2b \arg h'.$$
More generally, $\varphi$ is called a $\PPS(\mu,\nu)$ form if $\varphi$ satisfies the following transformation law
$$\varphi = \ti \varphi \circ h +\mu \log h' + \nu \log \overline{h'}.$$
\end{def*}
In general, a PPS form $\varphi$ of arg-type is multivalued but $\pa\varphi$ can be single-valued.
The notation $\PPS(ib,ib)$ refers to conformal dimensions of the differential $\ee^\varphi.$
Lemma~\ref{lem: varphi pm} below follows from simple algebraic facts that
$\ee^{-2b\arg h'} = (h'/\overline{h'})^{ib}$ and $\ee^{ib\log |h'|^2} = ({h'}{\,\overline{h'}\,})^{ib}.$

\begin{lem} \label{lem: varphi pm}
Let $\varphi$ be a harmonic PPS form.
\renewcommand{\theenumi}{\alph{enumi}}
{\setlength{\leftmargini}{1.8em}
\begin{enumerate} 
\item If $\varphi$ is a PPS form of arg-type, then
$\varphi^+$ is a $\PPS(ib)$  form and $\overline{\varphi^-}$ is a $\PPS(ib)$ form. 
\item
If $\varphi$ is a PPS form of $\log|\cdot|$-type, then $\varphi^+$ is a $\PPS(ib)$ form and $\overline{\varphi^-}$ is a $\PPS(-ib)$ form. 
\end{enumerate}
}
\end{lem}

\begin{eg*} Let $\omega$ be a meromorphic 1-differential. Then
$\varphi = ib \log |\omega|^2$ and $\varphi = -2b\arg \omega$
are harmonic PPS-forms. We have
$-2b\arg \omega = ib\log\omega - ib \log\bar\omega,$
and
$ib\log |\omega|^2 = ib\log\omega + ib \log\bar\omega.$
\end{eg*}

\subsubsec{Simple PPS forms}
We mostly concern with (harmonic) PPS forms $\varphi$ such that $\pa\bp\varphi$ is a finite linear combination of $\delta$-measures, and we call such forms \emph{simple}. \index{simple PPS form} \index{PPS form!simple}

One way to think of simple forms is as relatives of the forms $\PPS(ib,ib)$ and $\PPS(ib,-ib)$, respectively
$$ib\log\rho,\qquad -ib\log\frac{\bp f}{\pa f},$$
where $\rho$ is a conformal metric (a positive $(1,1)$-differential) on $M$ and $\bp f/\pa f$ is the Beltrami coefficient of a scalar function $f.$

We focus on the single-valued case, i.e. on PPS forms of $\log|\cdot|$-type  (corresponding to flat conformal metrics with conical singularities), and we leave the equally important case of PPS forms of arg-type as an exercise.

Let us recall the \emph{Gauss-Bonnet formula}: for each conformal metric $\rho,$ we have
$$-\int_M \pa\bp \log \rho = \pi \chi(M).$$

\begin{thm} \label{Gauss-Bonnet} 
The Gauss-Bonnet formula extends to simple $\PPS(1,1)$ forms $\varphi:$
$$-\int_M \pa\bp \varphi = \pi \chi(M).$$
\end{thm}

We present the simplest example first. 
\begin{eg*} 
We consider the genus zero case and choose a conformal metric $\rho = |\omega_q|^2,$ where $\omega_q$ is an Abelian differential with a sole double pole at $q \in M.$ 
In the identity chart of $\C$ with $q=0,$ we can take $\rho(z) = |z|^{-4},$ so
$$ \pa\bp \log \rho = -\Delta\log|z| = -2\pi \delta_q$$
and 
$$\int_M \pa\bp \varphi = -2\pi = -\pi \chi(M).$$
\end{eg*}

\begin{proof}[Proof of Theorem~\ref{Gauss-Bonnet}]
Let us choose a conformal metric $\rho = |\omega|^2,$ where $\omega$ is a holomorphic 1-differential. 
Fix a simple $\PPS(1,1)$ form $\varphi$ and let $\varphi_* = \log \rho = \log |\omega|^2.$ 
We want to show that 
$$\int_M\pa\bp(\varphi-\varphi_*) = 0.$$
By Green's theorem, the integral is the sum of all residues of the meromorphic differential $\pa(\varphi-\varphi_*)$
(up to a multiplicative constant), which is zero. 
\end{proof}

\begin{def*}
Let $\varphi$ be a simple $\PPS(ib, ib)$ form (which determines a Gaussian free field modification). 
Denote
$$\bfs\beta = \frac i\pi \,\pa\bp \varphi.$$
We think of $\bfs\beta$ as a measure,
$$\bfs\beta = \sum \beta_k \delta_{q_k},$$ 
or (1,1)-differential, or a divisor, and call it the \emph{background charge} of $\varphi.$
\end{def*}
Clearly, the background charge determines the modification.

\begin{cor} \label{cor: NC}
We have the neutrality condition $(\NC_b)$, \index{neutrality condition!(NC$_b$)}
$$\int\bfs\beta = b \chi(M). $$
\end{cor}
\begin{proof}
It follows from Theorem~\ref{Gauss-Bonnet} that 
$$\int\bfs\beta = \frac i\pi \int \pa\bp\varphi = \frac i\pi ib (-\pi\chi(M)). $$
\end{proof}

\subsection{Basic forms and modifications of the Gaussian free field} \label{ss: Basic forms} \index{basic (PPS) form} \index{PPS form!basic}
For $q\in M$ let $\varphi_q$ denote a simple form with a sole singularity at $q,$
$$\varphi_q(z)\sim\log\frac1{|z-q|^2} $$
in any chart, i.e.
$$\pa\bp\varphi_q = -\pi\delta_q.$$

\begin{thm} \label{thm: PPS} Let $g\ne 1.$ Then for any $(M_g,q),$ 
there exists a unique (up to an additive constant) basic form $\varphi_q.$
This form is $\PPS(1/\chi,1/\chi).$
\end{thm}

\begin{proof}
Let's start with the last statement. If $\varphi_q$ is a $\PPS(s,s)$ form, then
$$\pi = -\int\pa\bp \varphi_q = s\pi\chi,$$
so $s = 1/\chi.$

Suppose we have two simple forms $\varphi_q,\wt\varphi_q$ with a sole singularity at $q$ such that $\pa\bp\varphi_q = \pa\bp\wt\varphi_q = -\pi\delta_q.$
Then $\varphi_q-\wt\varphi_q$ is a harmonic function (without singularities), hence is a constant. 

Now let's prove the existence. In the genus zero case, we can take
$$\varphi_q = \log|\omega_q|,$$
where $\omega_q$ is a meromorphic 1-differential with a double pole at $q$ and neither zeros nor poles elsewhere. 
For example, if $M = \wh\C$ and $q=\infty,$ then in the identity chart of  $\C,$ we have
$\varphi_q(z) = 0$ (up to an additive constant), so if $f:(M,q)\to(\wh\C,\infty),$ then
\begin{equation} \label{eq: basic form}
\varphi_q(z) = \log |f'(z)|,\qquad (z\ne q)
\end{equation}
because $\varphi_q$ is a $\PPS(1/2,1/2)$ form. 
In particular, if $M = \wh\C$ and $q\ne\infty,$ then in the identity chart of  $\C,$ we have
\begin{equation} 
\label{eq: varphi_q0}
\varphi_q(z) = -2\log|z-q|.
\end{equation}
To see this, apply \eqref{eq: basic form} to $f(z) = 1/(z-q).$ 
Then for $h:M\to\wh\C,$ and $h(q)\ne \infty,$ we have 
$$\varphi_q(z) = \log |h'(z)|-2\log|h(z)-h(q)|.$$

If $g=2$ and $q$ is a Weierstrass point, then there exists a holomorphic Abelian differential $\omega_q$ with a double zero at $q,$ so we can define
$$\varphi_q = -\log|\omega_q|.$$
To construct $\varphi_{\ti q}$ for a non-Weierstrass base point $\ti q,$ we use the general formula
$$\varphi_q - \varphi_{\ti q} = 2G_{q,\ti q}.$$

For higher genera, the construction is slightly less explicit. 
Take any holomorphic Abelian differential $\omega$ and let us assume that all zeros $q_1,\cdots,q_N$ of $\omega$ are simple. Here $N = 2g -2.$ Define
$$N\varphi_q = -\log|\omega|^2 + \sum_{k=1}^N 2G_{q,q_k},$$
which completes the proof.
\end{proof}

\begin{cor} \label{cor: PPS}
Let $g\ne 1.$ 
Given a background charge $\bfs\beta=\sum\beta_k\cdot q_k$ with the neutrality condition
\begin{equation} \tag{$\NC_b$}
\sum_k\beta_k = b\chi(M),
\end{equation}
there is a unique (up to an additive constant) simple $\PPS(ib,ib)$ form $\varphi_{\bfs\beta}$ with background charge $\bfs\beta,$ i.e. $ \bfs\beta = \frac i\pi \pa\bp\varphi_{\bfs\beta}.$
Furthermore, we have 
$$\varphi_{\bfs\beta} = i\sum_k\beta_k\varphi_{q_k} $$
modulo an additive constant. 
\end{cor}

\begin{proof}
Since $\pa\bp \varphi_{q_k} = -\pi \delta_{q_k},$ we have 
$$\frac i\pi \,\pa\bp\Big(i\sum_k\beta_k\varphi_{q_k}\Big) = \sum_k\beta_k\delta_{q_k},$$
or $\varphi = i\sum_k\beta_k\varphi_{q_k} $ is a simple $\PPS(ib,ib)$ form with background charge $\bfs\beta.$ 
If $\wt\varphi$ is another such simple $\PPS(ib,ib)$ form, then $\varphi-\wt\varphi$ is a harmonic function (without singularities) on $M,$ hence it is a constant. 
\end{proof}

\begin{rmk*}
A continuous version of the last statement would be
$$\varphi(z) = i\int \varphi_q(z)\,\dd \bfs\beta(q),\qquad \int \bfs\beta = b\chi(M).$$
\end{rmk*}

The last corollary shows how to construct the modified field from the background charge using the basic PPS forms: we define the $\bfs\beta$-modified field $\Phi_{\bfs\beta}$ by 
$$\Phi_{\bfs\beta}(z,z_0) = \Phi(z,z_0) + \varphi_{\bfs\beta}(z) - \varphi_{\bfs\beta}(z_0),$$
where $\Phi$ is the Gaussian free field. 

The space of all $\PPS(ib, ib)$ forms and the space of scalar functions on $M$ are \emph{non-canonically} equivalent. 
Indeed, we can fix some form $\varphi_*$ (e.g. with just one singular point, $\varphi_*= ib\chi\varphi_{q_*},$ or the form $\varphi_*= ib\log\rho_*$ with a constant curvature metric $\rho_*$), and call it the \emph{reference form} (or call $\rho_*$ the \emph{reference metric}). \index{reference form} \index{reference metric}
The reference form $\varphi_*= ib\chi\varphi_{q_*}$ is said to be \emph{basic}.  
Then all other $\PPS(ib, ib)$ forms differ from $\varphi_*$ by some scalar functions $\psi,$
$$\varphi = \varphi_* + \psi.$$

\begin{prop} \label{prop: PPS}
Let $g\ne 1.$
Given a basic reference form $\varphi_*= ib\chi\varphi_{q_*}$ and a background charge $\bfs\beta=\sum\beta_k\cdot q_k$ with the neutrality condition $(\NC_b),$ we have 
$$
\varphi_{\bfs\beta}(z) - \varphi_{\bfs\beta}(w) = \varphi_*(z) - \varphi_*(w) + 2i\sum_k \beta_k (G_{q_k,q_*}(z) - G_{q_k,q_*}(w)).
$$
Furthermore, $\varphi_{\bfs\beta}(z) - \varphi_{\bfs\beta}(w)$ does not depend on the choice of $q_*.$ 
\end{prop}

\begin{proof}
By the neutrality condition $b\chi = \sum_k\beta_k,$ we have $\varphi_*=i(\sum_k\beta_k) \varphi_{q_*}.$
It follows from the previous corollary that 
$$\varphi_{\bfs\beta}- \varphi_* = i\sum_k \beta_k(\varphi_{q_k} - \varphi_{q_*}) = 2i\sum_k \beta_k G_{q_k,q_*}.$$
The last claim follows from the uniqueness of $\varphi_{\bfs\beta}$ modulo an additive constant, see the previous corollary. 
\end{proof}

Let us consider the genus one case. 
Denote $\varphi_* = 2ib \log|\omega|,$ where $\omega$ is a holomorphic form.

\begin{prop} \label{prop: PPS1} 
Given a background charge $\bfs\beta = \sum \beta_k\cdot q_k$ with the neutrality condition $(\NC_b)$ and a point $q_\infty\in M$ we define the scalar function
$$\psi_{\bfs\beta,q_\infty}(z,z_0) := 2i\sum \beta_k (G_{q_k,q_\infty}(z) -G_{q_k,q_\infty}(z_0)).$$
Then the form
$$\varphi_{\bfs\beta}(z,z_0) = \varphi_*(z) - \varphi_*(z_0)  + \psi_{\bfs\beta,q_\infty}(z,z_0)$$
does not depend on the choice of the point $q_\infty.$
\end{prop}

\begin{proof}
Given points $q_\infty, q'_\infty \in M,$ we have 
$$\psi_{\bfs\beta,q_\infty}(z,z_0) - \psi_{\bfs\beta,q'_\infty}(z,z_0) = - 2i(\sum \beta_k) (G_{z,z_0}(q_\infty) - G_{z,z_0}(q'_\infty)) = 0$$
by the neutrality condition $\sum \beta_k= 0.$
Here we use the bilinear relation~\eqref{eq: bilinear relation} for bipolar Green's function. 
\end{proof}

\begin{rmk*}
Given a divisor $\bfs\tau = \sum \tau_j\cdot z_j$ satisfying the neutrality condition $(\NC_0)$, we define
$$\varphi_{\bfs\beta}[\bfs\tau]:= \sum_j \tau_j \varphi_{\bfs\beta}(z_j), \quad \Phi_{\bfs\beta}[\bfs\tau]:=\Phi[\bfs\tau]+ \varphi_{\bfs\beta}[\bfs\tau].$$
Let us consider the case $g\ne 1$ first. Given a point $q_*\in M,$ let $\bfs\beta_*:= b\chi \cdot q_*.$
Then the basic reference form $\varphi_*= ib\chi\varphi_{q_*}$ is $\varphi_{\bfs\beta_*}.$
By Proposition~\ref{prop: PPS}, if $\supp\,\bfs\tau \cap (\supp\,\bfs\beta \cup \supp\,\bfs\beta_* )= \emptyset,$ then we have 
\begin{equation} \label{eq: PPS}
\varphi_{\bfs\beta}[\bfs\tau] = \varphi_{\bfs\beta_*}[\bfs\tau] + i\,\E\,\Phi[\bfs\beta-\bfs\beta_*] \Phi[\bfs\tau].
\end{equation}
Next, we consider the case $g = 1.$ For $\bfs\beta_* = \bfs 0,$ we have $\varphi_{\bfs\beta_*} = \varphi_*.$ 
By Proposition~\ref{prop: PPS1}, \eqref{eq: PPS} holds if $\supp\,\bfs\tau \cap \supp\,\bfs\beta = \emptyset.$

Given two background charges $\bfs\beta_1$ and $\bfs\beta_2,$ it follows from \eqref{eq: PPS} that the modifications $\varphi_{\bfs\beta_1}$ and $\varphi_{\bfs\beta_2}$ are related as 
\begin{equation} \label{eq: PPSs}
\varphi_{\bfs\beta_2}[\bfs\tau] = \varphi_{\bfs\beta_1}[\bfs\tau] + i\,\E\,\Phi[\bfs\beta_2-\bfs\beta_1] \Phi[\bfs\tau]
\end{equation}
if $\supp\,\bfs\tau \cap (\supp\,\bfs\beta_1 \cup \supp\,\bfs\beta_2)= \emptyset.$
\end{rmk*}

\begin{rmk*}
Let $g\ge 2.$ Given a background charge $\bfs\beta,$ the representation 
$$\varphi_{\bfs\beta}= \varphi_* + 2i\sum_k\beta_k G_{q_k,q_*}, \quad \varphi_*=ib\big(\log|\omega|^2 - 2\sum_{j=1}^{2g-2} G_{q_*,p_j}\big)$$
does not depend on the choice of a holomorphic 1-differential $\omega$ and a reference point $q_*$ (modulo an additive constant), see Corollary~\ref{cor: PPS}.
Here $p_j$'s are zeros of $\omega.$ 
Therefore, we have
\begin{equation}\label{eq: varphi beta} 
\varphi_{\bfs\beta}(z)-\varphi_{\bfs\beta}(z_0)= 2ib \log\Big|\frac{\omega(z)}{\omega(z_0)}\Big| + i\,\E\,\Phi[\bfs\beta-\bfs \beta_0]\Phi(z,z_0),
\end{equation}
where $\bfs\beta_0 = -b\cdot(\omega)$ and $(\omega)$ is the divisor of $\omega,$ i.e. $\bfs \beta_0 = -\sum_{j=1}^{2g-2} b\cdot p_j.$

More generally, for any genus $g,$ 
given a meromorphic 1-differential $\omega,$ 
let $\bfs\beta_0 = -b \cdot (\omega).$
Then $\varphi_{\bfs\beta_0} = ib \log |\omega|^2.$
If $\supp\,\bfs\tau \cap (\supp\,\bfs\beta \cup \supp\,\bfs\beta_0) = \emptyset,$ then by \eqref{eq: PPSs}, we have 
\begin{equation}\label{eq: varphi beta tau} 
\varphi_{\bfs\beta}[\bfs\tau] = 2ib \sum_{j} \tau_j\log|\omega(z_j)| + i\,\E\,\Phi[\bfs\beta-\bfs\beta_0]\Phi[\bfs\tau].
\end{equation}
\end{rmk*}

\subsection{Modifications of the current field} 
Given a background charge $\bfs\beta=\sum_k\beta_k\cdot q_k,$ we define 
$$J_{\bfs\beta}(z) := \pa_z \Phi_{\bfs\beta}(z,z_0), \quad j_{\bfs\beta} := \E\, J_{\bfs\beta}.$$ 
Then $J_{\bfs\beta} = J + j_{\bfs\beta}$ and $J_{\bfs\beta}$ is a pre-Schwarzian form of order $ib.$

\subsubsec{The genus zero case}
In the identity chart of $\C,$ we have 
$$j_{\bfs\beta}(z)= - i\sum_k \frac{\beta_k}{z -q_k}.$$ 
It follows from \eqref{eq: varphi_q0} and Proposition~\ref{prop: PPS}.

\subsubsec{The genus one case}
In the identity chart of torus $\T_\Lambda,$ we have
\begin{equation} \label{eq: jbeta4torus}
j_{\bfs\beta}(z) =- i \sum_k \beta_k \frac{\theta'(z-q_k)}{\theta(z-q_k)} -\frac{2\pi}{\Im\,\tau}\,\sum_k \beta_k\, \Im\,q_k.
\end{equation}
It follows from Proposition~\ref{prop: PPS1}.

\subsubsec{The higher genus case}
Fix a homomorphic 1-differential $\omega$ with zeros at $p_j$ $(j=1,\cdots,N=2g-2).$
It follows from \eqref{eq: varphi beta} that 
\begin{equation}\label{eq: EJbeta} 
j_{\bfs\beta}(z) = ib\, \frac{\pa\omega(z)}{\omega(z)} + i\,\E\, J(z)\Phi[\bfs\beta-\bfs\beta_0].
\end{equation}
The non-random field $j_{\bfs\beta}$ has a removable singularity at each zero $p_j$ of $\omega$ and a simple pole at $q_k$ with the residue $-i\beta_k.$

\subsection{Modifications of the Virasoro field}
Given a background charge $\bfs\beta = \sum \beta_k\cdot q_k,$ we define 
$$T_{\bfs\beta} := -\frac12 J_{\bfs\beta}*J_{\bfs\beta} + ib \pa J _{\bfs\beta}.$$
Then $T_{\bfs\beta}$ is a Schwarzian form of order $c/{12}$ ($c = 1-12b^2$ is the central charge), see Theorem~\ref{SET}.
From the relation $J_{\bfs\beta} = J + j_{\bfs\beta},$ we find
$$T_{\bfs\beta} = T + ib\pa J - j_{\bfs\beta}J - \frac12 j_{\bfs\beta}^2 + ib\pa j_{\bfs\beta}.$$
In particular,
\begin{equation}\label{eq: ETbeta} 
\E\,T_{\bfs\beta} = \E\, T -\frac12 j_{\bfs\beta}^2 + ib\pa j_{\bfs\beta}.
\end{equation}

\subsubsec{The genus zero case} 
In the $\wh\C$-uniformization, we have 
\begin{equation} \label{eq: ETbeta0}
\E\,T_{\bfs\beta}(z) = \sum _k \frac{\lambda_{q_k}}{(z-q_k)^2} + \sum_{j < k}\frac{\beta_j\beta_k}{(z-q_j)(z-q_k)},
\end{equation}
where $\lambda_{q_k} = \frac12\beta_k^2 - b\beta_k.$

\subsubsec{The genus one case} 
In the $\T_{\Lambda}$-uniformization, we have 
\begin{align*}
\E\,T_{\bfs\beta}(z) & =-\frac16 \frac{\theta'''(0)}{\theta'(0)}-\frac12\frac\pi{\Im\,\tau}+ \sum_k \lambda_{q_k} \Big(\frac{\theta'(z-q_k)}{\theta(z-q_k)}\Big)^2
+\sum_k b\beta_k\frac{\theta''(z-q_k)}{\theta(z-q_k)}\\
&+\sum_{j < k}\beta_j\beta_k\frac{\theta'(z-q_j)\theta'(z-q_k)}{\theta(z-q_j)\theta(z-q_k)} \\
&-\frac{2\pi i}{\Im\,\tau}\big(\sum_k\beta_k\,\Im\,q_k\big)\Big(\sum_k\beta_k\frac{\theta'(z-q_k)}{\theta(z-q_k)}\Big)
-\frac{2\pi^2}{(\Im\,\tau)^2}\big(\sum_k\beta_k\,\Im\,q_k\big)^2.
\end{align*}

\subsubsec{The higher genus case}
It is convenient to choose $\omega = \omega^{(1)},$ see \eqref{eq: omega} for the definition of the $j$-differential $\omega^{(j)}.$ 
We have 
\begin{align*}
\E\,T_{\bfs\beta}(z) &=\frac{c}{12} S_{\omega^{(0)}}(z)+ \frac18 \Big(\frac{\omega^{(2)}(z)}{\omega^{(1)}(z)}\Big)^2-\frac16 \frac{\omega^{(3)}(z)}{\omega^{(1)}(z)}-\frac\pi2 \langle (\Im\,\tau)^{-1} \vec\omega(z), \overline{\vec\omega(z)} \rangle\\
&+b\,\E\,\Big(\frac{\pa\omega^{(1)}(z)}{\omega^{(1)}(z)}J(z)-\pa J(z) \Big)\Phi[\bfs\beta-\bfs\beta_0]+\frac12\big(\E\,J(z) \Phi[\bfs\beta-\bfs\beta_0]\big)^2.
\end{align*}
By construction, it has poles only at $q_k$'s.

\section{Liouville action} \label{s: Liouville action}
In Subsection~\ref{ss: Dirichlet action} we argue that the Gaussian free field $\Phi$ can be produced by the Dirichlet action. 
In this section we modify this argument to show that the background charge modification $\Phi_{\bfs\beta}$ can be obtained from the Liouville action $S_{\bfs\beta}$ associated with a background charge $\bfs\beta$ on the PPS forms in the case that the modification parameter $b$ is imaginary, hence $c>1.$

\subsection{Liouville action on forms and localization} \index{Liouville action}
Given an imaginary $b,$ fix a $\PPS(ib,ib)$ form $\phi.$ 
For example, $\phi = ib \log \rho$ is a $\PPS(ib,ib)$ form if $\rho$ is a conformal Riemann metric. 
Another typical example is the simple PPS form $\phi = \varphi_{\bfs\beta}$ associated with a background charge $\bfs\beta,$ i.e. $\bfs\beta = i/\pi \,\pa\bp \varphi_{\bfs\beta}.$

Let us fix a reference $\PPS(ib,ib)$ form $\varphi_*.$ 
Our result does not depend on the choice of a reference form. 
We only consider the case that a reference form is basic: $\varphi_* = ib\chi(M) \varphi_{q_*}, \pa\bp\varphi_{q_*} = -\pi \delta_{q_*}$ for a fixed marked point $q_*.$
 
We define the ``action" associated with $\phi$ on the space of some forms by  \index{Liouville action}
$$S(\varphi) \equiv S_\phi(\varphi) = \DD(\psi) + \frac1\pi \int \psi \pa\bp \phi, \quad \varphi = \varphi_* + \psi.$$
Later, we specify the space via localization and consider the action on its finite dimensional subspace. 
If we consider $\phi = ib \log \rho$, then the second term becomes
$$\frac{ib}{\pi}\int \psi \pa\bp\log\rho = \frac b{2\pi i}\int \psi \kappa\rho,$$
where $\kappa$ is the curvature of the conformal metric, see Subsection~\ref{ss: spectral theory}.
When we consider $\phi = \varphi_{\bfs\beta}$ associated with a background charge $\bfs\beta,$ we write $S_{\bfs\beta}$ for $S_\phi.$ 
In this case, the action becomes
$$S_{\bfs\beta}(\varphi) = \DD(\psi) -i\int\psi\bfs\beta, \quad \DD(\psi) = \frac1{4\pi}(\psi,\psi)_\nabla.$$ 

To resolve the divergence in the partition function, we need to introduce
some localization, e.g.
\begin{equation} \label{eq: localization} 
\int \psi\mu_* = 0,
\end{equation}
where $\mu_*$ is a reference measure, or a (1,1)-form. 
For example, the choice $\mu_* = \bfs\beta$ reduces the action to the Dirichlet integral. 
We assume that
\begin{equation} \label{eq: mass one}
\int \mu_* = 1.
\end{equation}

\subsection{Fields produced by the Liouville action}
Let $\rho$ be an arbitrary conformal Riemann metric on $M.$
(This $\rho$ must be sufficiently smooth.) 
We assume that 
$$\int\rho = 1.$$
(One of the goals is to show that the result does not depend on the choice of $\rho.$)
Denote
$$H_n = \textrm{span}\{e_0,\cdots,e_n\} \cap \{\psi:\int\psi\mu_*=0\},$$
where $e_k$'s are the normalized eigenfunctions of the positive Laplace operator $\Delta_\rho,$ with $e_0=1.$ 

The Liouville action $S(\varphi),$ $(\varphi=\varphi_*+\psi)$ determines the probability measures
$$\frac1{Z} \ee^{-S(\varphi_n)}[\dd\varphi_n], \qquad \varphi_n = \varphi_*+\psi_n, \quad\psi_n=\sum_{k=0}^na_ke_k\in H_n, \quad [\dd\varphi_n]=\dd a_1\cdots \dd a_n,$$
where $Z$ is a normalizing constant or the partition function. 
Therefore the Liouville action determines random fields $\psi_n$ on $M$ as explained in Subsection~\ref{ss: Dirichlet action}. 
We compute the limit field $\Phi_S,$
$$\varphi_n \to \Phi_S$$
as distributional fields (or within correlations).

\begin{lem} Given $\lambda > 0, c\in\R,$
if $\mu$ is the probability measure on $\R$ corresponding to the partition function
$$Z = \int_\R \ee^{-\lambda x^2 - cx}\,\dd x,$$
then $\mu$ has the same moments $\displaystyle\int_\R x^k\,\dd \mu(x)$ as the Gaussian random variable 
$$X \sim N(-\frac c{2\lambda},\frac1{2\lambda}).$$
\end{lem}

\begin{proof}
For real $m,$ $\dfrac1{\sqrt{2\pi}\sigma}\ee^{-\frac{(x-m)^2}{2\sigma^2}}$ is the density of $N(m,\sigma^2)$, so we equate
$$\lambda x^2 + cx = \frac1{2\sigma^2}(x^2-2mx).$$
Comparing the coefficients, we have $m = -c/{(2\lambda)}$ and $\sigma^2 = 1/{(2\lambda)}.$
\end{proof}

\begin{thm} \label{Liouville action}
The limit field $\Phi_S$ is 
$$\Phi_S(z) = \Phi_\rho(z) + \varphi_*(z) - \int \Phi_\rho\mu_* - 2\iint R_{z,q}(\zeta)\nu_*(\zeta)\mu_*(q), $$
where
\begin{equation} \label{eq: nu*}
\nu_* = \pa\bp \phi - \mu_* \int \pa\bp\phi.
\end{equation}
\end{thm}

\begin{proof}
For $\psi_n = \sum_{k=0}^n a_k e_k \in H_n,$ we have 
$$S(\varphi_n) = \frac1{4\pi}\sum_{k=1}^n \lambda_ka_k^2 + \frac1{\pi}\sum_{k=0}^n c_ka_k, $$
where $\varphi_n = \varphi_*+\psi_n$ and 
$$c_k = \int e_k \,\pa\bp \phi.$$
By localization \eqref{eq: localization} and the assumption \eqref{eq: mass one}, we have 
$$ a_0 = -\sum_{k=1}^n a_k \int e_k \mu_*,$$
so
\begin{equation} \label{eq: action}
S(\varphi_n) = \frac1{4\pi}\sum_{k=1}^n \lambda_ka_k^2 + \frac1{\pi}\sum_{k=1}^n d_ka_k,
\end{equation}
where
\begin{equation}\label{eq: dk}
d_k = c_k-c_0 \int e_k \mu_* = \int e_k\nu_*,\qquad (k\ge1).
\end{equation}
Since the variables $\{a_k\}_{k=1}^n$ in the probability measure 
\begin{equation} \label{eq: mu Z}
\dd\mu_n = \frac{\ee^{-S(\varphi_n)}}{Z_n} \,\dd a_1\cdots \dd a_n, \quad Z_n = \int_{\R^n } \ee^{-S(\varphi_n)} \,\dd a_1\cdots \dd a_n
\end{equation}
determined by the action~\eqref{eq: action} are separated, we can consider the coefficients $\{a_k\}_{k=1}^n$ of $\psi_n$ as independent random variables such that
$$a_k \overset{d}{=} \frac{\sqrt{2\pi}}{\sqrt{\lambda_k}}\chi_k - \frac{2d_k}{\lambda_k},\quad \chi_k\sim N(0,1),$$
see the previous lemma. 
It follows that
\begin{align*}
\Phi_S &= \varphi_*+ a_0 + \sum_{k=1}^\infty a_k e_k = \varphi_*+ \sum_{k=1}^\infty a_k\Big(e_k-\int e_k\mu_*\Big)\\
&=\varphi_*+ \sqrt{2\pi}\sum_{k=1}^\infty \frac{\chi_k}{\sqrt{\lambda_k}} \Big(e_k-\int e_k\mu_*\Big) - 2\sum_{k=1}^\infty \frac{d_k}{\lambda_k}\Big(e_k-\int e_k\mu_*\Big): = \varphi_*+ \mathrm{I}+\mathrm{II}.
\end{align*}
By Corollary~\ref{cor: Gaussian fields}, we have 
$$\mathrm{I} = \Phi_\rho - \int \Phi_\rho\mu_*.$$
To compute the second term, we first note that
$$\sum\frac{d_k}{\lambda_k} e_k(z) = \sum \int \frac{e_k(z)e_k(\zeta)\nu_*(\zeta)}{\lambda_k} = \int R(z,\zeta)\nu_*(\zeta),$$
and therefore
\begin{align*}
\mathrm{II}&= -2\int R(z,\zeta)\nu_*(\zeta) +2\iint R(q,\zeta)\nu_*(\zeta) \mu_*(q) \\
&= -2\iint \big(R(z,\zeta)-R(q,\zeta)\big)\nu_*(\zeta) \mu_*(q),
\end{align*}
which completes the proof.
\end{proof}

\begin{rmk*}
\renewcommand{\theenumi}{\alph{enumi}}
{\setlength{\leftmargini}{1.8em}
\begin{enumerate}
\item We can rewrite I and II as follows:
$$ \textrm{I} = \int (\Phi_\rho(z)-\Phi_\rho(\zeta))\mu_*(\zeta) = \int\Phi(z,\zeta)\mu_*(\zeta) =:\Phi_{\mu_*}(z)$$
and
$$\textrm{II} = -\frac2\pi\iint G_{z,q}(\zeta)\nu_*(\zeta)\mu_*(q)$$
because $\int \nu_* = 0$. 
(Recall that $G_{z,q}$ is defined up to an additive constant which depends on $z$ and $q.$)
Thus
$$\Phi_S(z) = \varphi_*(z) + \Phi_{\mu_*}(z) -\frac2\pi \iint G_{z,q}(\zeta)\nu_*(\zeta)\mu_*(q) $$
does not depend on the choice of $\rho.$
\item The second term II or $\psi:=\E\,\Phi_S$ requires some smoothness of $\mu_*,$ 
e.g. if $\mu_* = \delta_{q*},$ then 
$$\textrm{II} = -\frac2\pi\int G_{z,q_*}(\zeta)\nu_*(\zeta)$$ 
and typically $\nu_*$ has an atom at $q_*$ hence the divergence. 
This divergence disappears when we consider the 2-point fields:
$$\Phi_S(z,z_0) = \Phi(z,z_0) +\varphi_*(z)- \varphi_*(z_0)+ \psi(z) - \psi(z_0),$$ 
where 
$$\psi(z) - \psi(z_0) = -\frac2\pi \int G_{z,z_0}\nu_*.$$ 
\end{enumerate}
}
\end{rmk*}

We can now choose $\mu_* = \delta_{q_*}$ and obtain the following corollary.
\begin{cor}
If $\mu_* = \delta_{q_*},$ then
$$\Phi_S(z,z_0) = \Phi(z,z_0)+ \varphi_*(z)- \varphi_*(z_0) -\frac2\pi \int (G_{\zeta,q_*}(z)- G_{\zeta,q_*}(z_0))\pa\bp\phi(\zeta).$$
(The case of general normalization measure can be treated as an integral over $q_*.$)
\end{cor}
\begin{proof}
By \eqref{eq: nu*}, we have 
$$\psi(z) - \psi(z_0) = -\frac2\pi \int G_{z,z_0}\nu_*  =  -\frac2\pi\int G_{z,z_0}(\zeta)\Big(\pa\bp \phi (\zeta)- \delta_{q_*}(\zeta) \int \pa\bp\phi \Big).$$
It follows from the bilinear relation~\eqref{eq: bilinear relation} for bipolar Green's function that 
\begin{align*}
\psi(z) - \psi(z_0)  
&= -\frac2\pi \int (G_{z,z_0}(\zeta)- G_{z,z_0}(q_*))\pa\bp\phi(\zeta)\\
&= -\frac2\pi \int (G_{\zeta,q_*}(z)- G_{\zeta,q_*}(z_0))\pa\bp\phi(\zeta),
\end{align*}
which completes the proof.
\end{proof}

Let us consider the case that $\phi$ is the simple $\PPS(ib,ib)$ form $\varphi_{\bfs\beta}$ associated with a background charge $\bfs\beta,$ i.e. $\pa\bp\phi = \pi/i\bfs\beta.$

\begin{cor}
Given a background charge $\bfs\beta,$ we identify the modified field $\Phi_{\bfs\beta}(z,z_0)$ with the the field $\Phi_S(z,z_0)$ given by the action $S = S_{\bfs\beta}.$ 
The bi-variant field $\Phi_S(z,z_0)$ does not depend on $\rho,$ the choice of a reference form $\varphi_*,$ and a marked point $q_*.$ 
\end{cor} 

\subsection{Partition functions of Liouville actions on a torus} \index{partition function!of Liouville action}
In this subsection we consider the genus one case only. 
For the action $S_{\bfs\beta}$ associated with a background charge $\bfs\beta,$ 
it follows from \eqref{eq: mu Z} that 
$$Z_n^{\bfs\beta} = \int_{\R^n} e^{-\frac1{4\pi} \sum_{k=1}^n \lambda_ka_k^2 - \frac1\pi \sum_{k=1}^n d_k a_k}\,\dd a_1\cdots \dd a_n,$$
where $d_k$ is given by \eqref{eq: dk} and $\nu_*$ in \eqref{eq: dk}  is given by \eqref{eq: nu*}.
It follows from the neutrality condition $\int \bfs\beta = 0$ that
$$\nu_* =  -i \pi \bfs\beta.$$
Thus we have 
$$d_n = -i\pi \int e_n \bfs\beta = -i\pi \sum_k \beta_k e_n(q_k).$$
Using the equation
$$\int_\R e^{-\lambda x^2 - cx}\,dx = \sqrt{\frac\pi\lambda} e^{c^2/4\lambda},$$
we obtain
$$\log\frac{Z_n^{\bfs\beta}}{Z_n}= \frac1\pi \sum_{k=1}^n \frac{d_k^2}{\lambda_k}.$$ 
However, the right-hand side diverges as $n\to\infty.$ 
Indeed, by the neutrality condition, we have 
\begin{align*}
\frac1\pi\sum_{n=1}^\infty \frac{d_n^2}{\lambda_n} &= -\pi\sum_{k,l}\beta_k\beta_l \sum_{n=1}^\infty \frac{e_n(q_k)e_n(q_l)}{\lambda_n} \\
&= -\pi\sum_{k,l}\beta_k\beta_l R_\rho(q_k,q_l) = - \sum_{k,l}\beta_k\beta_l G_{q_k}(q_l).
\end{align*}
Subtracting a divergent term, we define $Z_{\bfs\beta}$ by 
$$\log\frac{Z_{\bfs\beta}}Z  = -\sum_k \beta_k^2 c(q_k) - 2\sum_{k<l} \beta_k\beta_l G_{q_k}(q_l).$$
In terms of $c[\bfs\beta]= c_{\Phi[\bfs\beta]}$ (see \eqref{eq: cPhiSigma}), we have 
$$Z_{\bfs\beta}=Z  \,\ee^{-c[\bfs\beta]}.$$
Later, we introduce a puncture operator $\PP_{\bfs\beta}$ and identify it with $\ee^{-c[\bfs\beta]}$ in the genus one case.  
In this case, we have
\begin{equation} \label{eq: Zbeta=ZPbeta}
Z_{\bfs\beta}=Z  \,\PP_ {\bfs\beta}.
\end{equation}

\section{OPE exponentials of modified Gaussian free fields} \label{s: O}
In this section we extend the concept of multi-vertex fields to the case $b \ne 0.$
Given a divisor $\bfs\tau$ with the neutrality condition $(\NC_0)$ we define the modified multi-vertex fields $\OO_{\bfs\beta}[\bfs\tau]$ as the OPE exponentials of $i\Phi_{\bfs\beta}[\bfs\tau].$
We represent their correlations $\E\,\OO_{\bfs\beta}[\bfs\tau]$ in terms of the Coulomb gas correlation functions $\CC_{(b)}$ with the modification parameter $b$ and the puncture operators $\PP_{\bfs\beta}.$

\subsection{Definition of OPE exponentials}
Fix a meromorphic 1-differential $\omega$ and let $\bfs\beta_0 = -b\cdot(\omega).$ 
Given a background charge $\bfs\beta,$ if a divisor $\bfs\tau$ with the neutrality condition $(\NC_0)$ satisfies $\supp\,\bfs\tau \cap (\supp\,\bfs\beta \cup \supp\,\bfs\beta_0) = \emptyset,$ then we have the representation~\eqref{eq: varphi beta tau} for $\varphi_{\bfs\beta}[\bfs\tau].$ 
It can be extended to the general case:
\begin{equation}\label{eq: varphi beta*tau} 
\varphi_{\bfs\beta}[\bfs\tau] = 2ib \sum_{j} \tau_j\log|\omega(z_j)| + i\,\E\,\Phi[\bfs\beta-\bfs\beta_0]*\Phi[\bfs\tau].
\end{equation}
See \eqref{eq: Phi*} for the definition of $\Phi[\bfs\sigma]*\Phi[\bfs\tau].$

We use the puncture operators $\PP_{\bfs \beta}$ to modify the multi-vertex fields so that the OPE calculus of the modified multi-vertex fields has a very simple and natural form.
We call the modified multi-vertex fields the OPE exponentials; the reason for these terminologies becomes clear in the definition of $\OO_{\bfs\beta}[\bfs\tau]$ below and that of multi-vertex fields $\VV[\bfs\sigma]$ in Subsection~\ref{ss: OC}.

We now define the \emph{OPE exponentials} \index{OPE exponentials} (the \emph{modified multi-vertex fields}) \index{multi-vertex field!modified} $\OO_{\bfs\beta}[\bfs\tau]$ by 
$$\OO_{\bfs\beta}[\bfs\tau] = \ee^{*i\Phi_{\bfs\beta}[\bfs\tau]}, \quad \Phi_{\bfs\beta}[\bfs\tau] = \Phi[\bfs\tau] + \varphi_{\bfs\beta}[\bfs\tau].$$
Then it follows from Lemma~\ref{lem: ef eX} and Proposition~\ref{prop: OPE exp} that 
$$\OO_{\bfs\beta}[\bfs\tau] = \ee^{i \varphi_{\bfs\beta}[\bfs\tau]}\CC[\bfs\tau]\VV^{\odot}[\bfs\tau].$$
Recall that $\CC[\bfs\tau] = \ee^{-c[\bfs\tau]}$ and $\VV^{\odot}[\bfs\tau] = \ee^{\odot i \Phi[\bfs\tau]}.$

\subsection{Coulomb gas correlation functions} \label{ss: C}

Let us fix the (background charge) parameter $b\in\R.$ 
This parameter $b$ is related to the central charge $c$ in the following way:
$$c = 1 - 12b^2.$$

Suppose a divisor $\bfs\sigma$ on $M$ satisfies the neutrality condition $(\NC_b).$
Given a reference point $q_*,$ let $\bfs\beta_* = b\chi\cdot q_*.$  
We now define the \emph{Coulomb gas correlation functions} \index{Coulomb gas correlation function} $\CC_{(b)}[\bfs\sigma]$ by
\begin{equation} \label{def: Cb}
\CC_{(b)}[\bfs\sigma] := \E\,\OO_{\bfs\beta_*}[\bfs\sigma-\bfs\beta_*].
\end{equation}
They are well-defined differentials of conformal dimension $(\lambda_j,\lambda_j)$ at $z_j,$ where $\lambda_j = \sigma_j^2/2-b\sigma_j.$ 
Formally, one can extend \eqref{eq: cPhiSigma}~--~\eqref{eq: c addition} and \eqref{eq: varphi beta*tau} to arbitrary divisors (which do not necessarily satisfy the neutrality condition $(\NC_0)$). 
At least formally, we have 
\begin{align*}
\log\CC_{(b)}[\bfs\sigma] 
&= - c[\bfs\sigma-\bfs\beta_*] + i\varphi_{\bfs\beta_*}[\bfs\sigma-\bfs\beta_*]\\
&= \log C_* - c[\bfs\sigma] + \E\,\Phi[\bfs\beta_0]\Phi[\bfs\sigma] - 2b\sum\sigma_j\log|\omega(z_j)|,
\end{align*}
where $\log C_* = c[\bfs\beta_*] -\E\,\Phi[\bfs\beta_0]\Phi[\bfs\beta_*] + 2b^2\chi \log|\omega(q_*)|$ is a constant and $\bfs\beta_0 = -b\cdot(\omega).$ 
We often ignore the constant $C_*.$
Sometimes it is convenient to write \eqref{def: Cb} in terms of formal expressions as
\begin{equation} \label{eq: Cb}
\CC_{(b)}[\bfs\sigma] = C_*\,\CC[\bfs\sigma] \prod_j |\omega(z_j)|^{-2b\sigma_j}\ee^{\E\,\Phi[\bfs\beta_0]\Phi[\bfs\sigma]}.
\end{equation}

\subsubsec{The genus zero case} 
If we choose $\omega = \omega_q$ as a meromorphic 1-differential with a double pole at $q$ and neither zeros nor poles elsewhere and if $(M,q) = (\wh\C,\infty),$ then $\log|\omega(z)| = 0$ (up to an additive constant) in the identity chart of $\C.$ 
Thus the evaluation of $\CC_{(b)}[\bfs\sigma]$ coincides with that of $\CC_{(0)}[\bfs\sigma]$ (up to a multiplicative constant) in the identity chart of $\C.$ 
Fix $q_*$ in $M.$ 
In terms of a conformal map $w: (M,q_*)\to(\wh\C,\infty),$ (up to a multiplicative constant)
$$\CC_{(b)}[\bfs\sigma] = \prod_j |w'_j|^{2\lambda_j}\prod_{j<k}|w_j-w_k|^{2\sigma_j\sigma_k},$$
where $w_j = w(z_j)$ and $w'_j = w'(z_j).$ 

\subsubsec{The genus one case} 
If we choose $\omega$ as a holomorphic 1-differential, then $\bfs\beta_0 (= -b(\omega)) = \bfs0$ and $\omega$ is constant in the identity chart of torus $\T_\Lambda.$
Thus the evaluation of $\CC_{(b)}[\bfs\sigma]$ coincides with that of $\CC_{(0)}[\bfs\sigma]$ (up to a multiplicative constant) in the identity chart of torus $\T_\Lambda.$
In terms of a conformal map $w:M\to\T_\Lambda,$ (up to a multiplicative constant)
$$\CC_{(b)}[\bfs\sigma] = |\theta'(0)|^{\sum_j \sigma_j^2} \prod_j |w'_j|^{2\lambda_j}\prod_{j<k} |\theta(w_j-w_k)|^{2\sigma_j\sigma_k} 
\exp \Big(\frac{2\pi}{\Im\,\tau}
\big(\Im\sum_j \sigma_j w_j \big)^2\Big).$$
 
\subsubsec{The higher genus case}
It is convenient to choose $\omega:=\sum_j \pa_j\Theta(e)\omega_j.$
Let $p_l$'s be the zeros of $\omega.$ 
Then it follows from \eqref{eq: C4higher} and \eqref{eq: Cb} that (up to a multiplicative constant)
\begin{align*}
\CC_{(b)}[\bfs\sigma] = C_* & \prod_j |\omega(z_j)|^{2\lambda_j}
\prod_{j<k} |\theta(z_j-z_k)|^{2\sigma_j\sigma_k}
\prod_{j,l} |\theta(z_j-p_l)|^{2b\sigma_j} \\
& \exp \big(2\pi(\Im\,\PM)^{-1}\Im\,\AA[\bfs\sigma]\cdot\Im\,\AA[\bfs\sigma-2\bfs \beta_0]\big),
\end{align*}
where $\bfs\beta_0 = -b\cdot(\omega).$

\subsection{OPE exponentials and Coulomb gas correlation functions} \label{ss: OC}
We now prove Theorem~\ref{main: EO}.
\begin{proof}[Proof of Theorem~\ref{main: EO}] 
By definition, 
$$\frac{\CC_{(b)}[\bfs\tau+\bfs\beta]}{\CC_{(b)}[\bfs\beta]} = \frac{\E\, \OO_{\bfs\beta_*}[\bfs\tau+\bfs\beta - \bfs\beta_*]}{\E\, \OO_{\bfs\beta_*}[\bfs\beta - \bfs\beta_*]}.$$
Equation~\eqref{eq: EO} now follows from the cocycle identity: if both $\bfs\tau$ and $\bfs\tau'$ satisfy the neutrality condition $(\NC_0),$ then
$$\E\, \OO_{\bfs\beta}[\bfs\tau + \bfs\tau'] = \E\, \OO_{\bfs\beta}[\bfs\tau] \, \E\, \OO_{\bfs\beta+\bfs\tau}[\bfs\tau'].$$ 
Equivalently, 
$$i\varphi_{\bfs\beta}[\bfs\tau + \bfs\tau'] -c[\bfs\tau + \bfs\tau'] = i\varphi_{\bfs\beta}[\bfs\tau] -c[\bfs\tau] + i\varphi_{\bfs\beta+\bfs\tau}[\bfs\tau'] -c[\bfs\tau'].$$
To see this, we need the following identities,
\begin{equation}\label{eq: c addition*} 
c[\bfs\tau + \bfs\tau'] -c[\bfs\tau]-c[\bfs\tau'] = \E\,\Phi[\bfs\tau]*\Phi[\bfs\tau']
\end{equation}
and 
$$\varphi_{\bfs\beta}[\bfs\tau + \bfs\tau'] - \varphi_{\bfs\beta}[\bfs\tau] - \varphi_{\bfs\beta+\bfs\tau}[\bfs\tau'] 
=  -i \E\,\Phi[\bfs\tau]*\Phi[\bfs\tau'].$$
Applying a normalization procedure via operator product expansion to \eqref{eq: c addition}, we have the first identity.
The second one follows from \eqref{eq: varphi beta*tau}. 
Indeed, we have 
\begin{align*}
i\varphi_{\bfs\beta}&[\bfs\tau + \bfs\tau'] - i\varphi_{\bfs\beta}[\bfs\tau] - i\varphi_{\bfs\beta+\bfs\tau}[\bfs\tau'] 
\\&= -\E\, \Phi[\bfs\beta-\bfs\beta_0]*\Phi[\bfs\tau+\bfs\tau'] +\,\E\, \Phi[\bfs\beta-\bfs\beta_0]*\Phi[\bfs\tau] +\,\E\, \Phi[\bfs\beta-\bfs\beta_0+\bfs\tau]*\Phi[\bfs\tau']\\
&= \E\,\Phi[\bfs\tau]*\Phi[\bfs\tau'].
\end{align*}
Suppose that there exist two differentials $\CC_{(b)}[\bfs\sigma]$ and $\wt\CC_{(b)}[\bfs\sigma]$ satisfying \eqref{eq: EO}.
Fix a reference background charge $\bfs\beta_*.$
Given a divisor $\bfs\sigma$ satisfying the neutrality condition $(\NC_b),$ let $\bfs\tau = \bfs\sigma-\bfs\beta_*.$
Then \eqref{eq: EO} with $\bfs\beta = \bfs\beta_*$ implies 
$$\wt\CC_{(b)}[\bfs\sigma] = C_* \,\CC_{(b)}[\bfs\sigma],$$
where $C_* = \wt\CC_{(b)}[\bfs\beta_*]/\CC_{(b)}[\bfs\beta_*]$ is a constant. 
\end{proof}

\begin{egs*}
(a) \textit{The genus zero case.} 
If $\supp\,\bfs\tau\cap\supp\,\bfs\beta = \emptyset,$ then in terms of a conformal map $w: (M,q_*)\to(\wh\C,\infty),$ 
$$\E\,\OO_{\bfs\beta}[\bfs\tau] = \prod_j |w'(z_j)|^{2\lambda_j}\prod_{j,k} |w(z_j)-w(q_k)|^{2\tau_j\beta_k} \prod_{j<k}|w(z_j)-w(z_k)|^{2\tau_j\tau_k}.$$ 
If $\supp\,\bfs\tau\cap\supp\,\bfs\beta \ne \emptyset,$ then the formula for $\E\,\OO_{\bfs\beta}[\bfs\tau]$ can be obtained by the previous theorem or by applying the following rooting rule to the above formula at each $z_{j_0}=q_{k_0}\in\supp\,\bfs\tau\cap\supp\,\bfs\beta:$
\begin{itemize}
\item the point $z_{j_0}$ in the terms $w'(z_{j_0}),$ $w(z_{j_0})-w(z_k)$ and $w(z_{j_0})-w(q_k)\, (k\ne k_0)$ is replaced by the puncture $q_{k_0};$
\item the term $w(z_{j_0})-w(q_{k_0})$ is replaced by $w'(q_{k_0}),$
\end{itemize}

\medskip\noindent 
(b) \textit{The genus one case.} 
If $\supp\,\bfs\tau\cap\supp\,\bfs\beta = \emptyset,$ then in terms of a conformal map $w:M\to\T_\Lambda,$
\begin{align*}
\E\,\OO_{\bfs\beta}[\bfs\tau] =&|\theta'(0)|^{\sum_j \tau_j^2} \prod_j |w'(z_j)|^{2\lambda_j}\\
&\prod_{j,k} |\theta(w(z_j)-w(q_k))|^{2\tau_j\beta_k} \prod_{j<k} |\theta(w(z_j)-w(z_k))|^{2\tau_j\tau_k} \\
&\exp \frac{2\pi}{\Im\,\tau}\Big(\big(\sum_j \tau_j\, \Im\,w(z_j) \big)^2 + 2\sum_{j,k}\tau_j\beta_k\,\Im\, w(z_j)\,\Im\,w(q_k) \Big).
\end{align*}
In the case $\supp\,\bfs\tau\cap\supp\,\bfs\beta \ne \emptyset,$ the rooting rule in (a) can be applied to (b). 
In addition, the point $z_{j_0}$ in the terms $\Im\, w(z_{j_0})$ is replaced by the puncture $q_{k_0}.$

\medskip\noindent 
(c) \textit{The higher genus case.}
If $\supp\,\bfs\tau\cap\supp\,\bfs\beta = \emptyset,$ then
\begin{align*}
\E\,\OO_{\bfs\beta}[\bfs\tau] =&\prod_j |\omega(z_j)|^{2\lambda_j}
\prod_{j<k} |\theta(z_j-z_k)|^{2\tau_j\tau_k} 
\prod_{j,l} |\theta(z_j-p_l)|^{2b\tau_j} 
\prod_{j,k} |\theta(z_j-q_k)|^{2\tau_j\beta_k} \\
&
\exp \big(2\pi(\Im\,\PM)^{-1}\Im\,\AA[\bfs\tau]\cdot\Im\,\AA[\bfs\tau+2(\bfs\beta-\bfs \beta_0)] \big).
\end{align*} 
If $\supp\,\bfs\tau\cap\supp\,\bfs\beta \ne \emptyset,$ then the following rooting rule can be applied:
\begin{itemize}
\item the point $z_{j_0}$ in the terms $\omega(z_{j_0}),$ $\theta(z_{j_0} - p_l),$ $\AA(\bfs\tau),$ $\theta(z_{j_0}-z_k),$ and $\theta(z_{j_0}-q_k)$ $(k\ne k_0)$ is replaced by the puncture $q_{k_0};$
\item the term $\theta(z_{j_0}-q_{k_0})$ is replaced by $\omega(q_{k_0}).$
\end{itemize}
\end{egs*}

We define the \emph{multi-vertex fields} \index{multi-vertex field} $\VV[\bfs\sigma]\equiv \VV_{\bfs\beta}[\bfs\sigma]$ by 
$$\VV[\bfs\sigma] = \CC_{(b)}[\bfs\sigma]~\VV^{\odot}[\bfs\sigma-\bfs\beta], \quad \VV^{\odot}[\bfs\tau] = \ee^{\odot i \Phi[\bfs\tau]}.$$
Let us introduce the \emph{puncture operators} \index{puncture operator} $\PP_{\bfs\beta}$ as  
$$\PP_{\bfs\beta}:=\VV[\bfs\beta] = \CC_{(b)}[\bfs\beta].$$
In terms of the multi-vertex fields and the puncture operator, we have 
$$\OO_{\bfs\beta}[\bfs\tau] = \PP_{\bfs\beta}^{-1}\VV[\bfs\tau + \bfs\beta].$$ 
We call $\bfs\tau$ the divisor of Wick's exponents and $\tau_{q_k}+\beta_k$ the \emph{effective charges} at $q_k.$ \index{effective charge}
Thus the sum of exponents is zero but the sum of (effective) charges is $b\chi(M).$

\begin{prop}\label{OPE nature}
If both $\bfs\tau_1$ and $\bfs\tau_2$ satisfy the neutrality condition $(\NC_0),$ then 
$$\OO_{\bfs\beta}[\bfs\tau_1]*\OO_{\bfs\beta}[\bfs\tau_2] = \OO_{\bfs\beta}[\bfs\tau_1+\bfs\tau_2].$$
\end{prop}
\begin{proof}
Let us consider the simplest case $b= 0$ for two formal 1-point fields.
We have the operator expansion of two such fields as 
$$\VV^{(\tau_1)}(\zeta)\VV^{(\tau_2)}(z) ={|\zeta-z|^{2\tau_1\tau_2}}\big(C_{00}(z) + O(|\zeta-z|)\big).$$ 
We claim that the first coefficient $C_{00}(z) = \VV^{(\tau_1+\tau_2)}(z).$ 
First note that
$$\VV^{(\tau_1)}(\zeta)\VV^{(\tau_2)}(z)  = C(\zeta)^{-\tau_1^2}C(z)^{-\tau_2^2} \ee^{-\tau_1\tau_2\E\,\Phi(\zeta)\Phi(z)} \ee^{\odot i\tau_1\Phi(\zeta)}\odot \ee^{\odot i\tau_2\Phi(z)},$$
where $C \equiv C_{\Phi} = \ee^{c_\Phi}.$
We expand the above as 
$${|\zeta-z|^{2\tau_1\tau_2}}\big(C(z)^{-\tau_1^2-\tau_2^2} \ee^{-\tau_1\tau_2\E\,\Phi*\Phi(z)} \ee^{\odot i(\tau_1+\tau_2)\Phi(z)} + O(|\zeta-z|)\big).$$
Thus we have
$$C_{00}(z) = C(z)^{-(\tau_1+\tau_2)^2} \ee^{\odot i(\tau_1+\tau_2)\Phi(z)} =\VV^{(\tau_1+\tau_2)}(z).$$
In the general case, using a similar argument, the only thing to check is 
$$\E\,\OO_{\bfs\beta}[\bfs\tau_1]\E\,\OO_{\bfs\beta}[\bfs\tau_2]  \ee^{-\E\,\Phi[\bfs\tau_1]*\Phi[\bfs\tau_2]}= \E\,\OO_{\bfs\beta}[\bfs\tau_1+\bfs\tau_2].$$
Taking the logarithm on both sides of the above, it is equivalent to the identity  
$$i\varphi_{\bfs\beta}[\bfs\tau_1]  -c[\bfs\tau_1] + i\varphi_{\bfs\beta}[\bfs\tau_2]  -c[\bfs\tau_2] -\E\,\Phi[\bfs\tau_1]*\Phi[\bfs\tau_2] = i\varphi_{\bfs\beta}[\bfs\tau_1+\bfs\tau_2]  -c[\bfs\tau_1+\bfs\tau_2].$$
It follows from \eqref {eq: c addition*} and 
$$\varphi_{\bfs\beta}[\bfs\tau_1]  + \varphi_{\bfs\beta}[\bfs\tau_2]  = \varphi_{\bfs\beta}[\bfs\tau_1+\bfs\tau_2].$$
The above identity is obvious by \eqref{eq: varphi beta*tau}. 
\end{proof}

\subsection{Insertion fields}
In this subsection we explain how background charge modifications of the Gaussian free field can be obtained from the insertion procedure.
This procedure is widely used in conformal field theory. 
For example, the insertion of the chiral bi-vertex field with specific charges is used as a boundary condition changing operator in a simply connected domain with two marked boundary points to explain the relation between chordal SLE theory and conformal field theory, see \cite{KM13}.

Fix $q_*$ in $M.$ In the genus one case, we also fix a holomorphic differential $\omega.$
Given a basic reference form $\varphi_*:$ 
$$\varphi_* = 
\begin{cases} 
ib\chi\, \varphi_{q_*}, &\textrm{ if } g\ne 1; \\
ib\log|\omega|^2,   &\textrm{ if } g=1,
\end{cases}$$
we define a central charge modification $\Phi_{(b)}$ of Gaussian free field by 
$$\Phi_{(b)}(z,z_0) = \Phi(z,z_0) + \varphi_*(z)-\varphi_*(z_0).$$
Note that $\Phi_{(b)} = \Phi_{\bfs\beta_*}, \bfs\beta_*= b\chi\cdot q_*$ for $g\ne 1$ and $\Phi_{(b)} = \Phi_{\bfs0}$ for $g=1.$
Given a background charge $\bfs\beta,$ the correspondence $\XX\mapsto\wh\XX$ on Fock space functionals/fields is defined by the formula 
\begin{equation}\label{eq: hat formula}
\wh\Phi(z,z_0)= \Phi(z,z_0) + 2i\sum_k \beta_k (G_{q_k,q_*}(z) - G_{q_k,q_*}(z_0))
\end{equation}
and the rules 
\begin{equation}\label{eq: hat rules}
1 \mapsto 1,\quad \mu\XX + \nu\YY \mapsto \mu\wh\XX + \nu\wh\YY,\quad \pa \XX \mapsto \pa\wh\XX, \quad \bp \XX \mapsto \bp\wh\XX, \quad \XX\odot\YY \mapsto \wh\XX\odot\wh\YY,
\end{equation}
where $\mu$ and $\nu$ are non-random functions. 
Then 
$$\wh\Phi_{(b)}(z,z_0)\equiv\Phi_{\bfs\beta}(z,z_0) = \Phi_{(b)}(z,z_0) + 2i\sum_k \beta_k (G_{q_k,q_*}(z) - G_{q_k,q_*}(z_0)).$$
Let us denote 
\begin{equation}\label{def: wh E}
\wh\E[\XX] := \E\, \VV^{\odot}[\bfs\beta-\bfs\beta_*]\XX, \quad \VV^{\odot}[\bfs\tau]=\ee^{\odot i\Phi[\bfs\tau]}.
\end{equation}
\begin{prop}
Suppose $\XX$ maps to $\wh\XX$ under the correspondence given by the formula \eqref{eq: hat formula} and the rules \eqref{eq: hat rules}. 
Then we have 
\begin{equation} \label{eq: hat hat}
\wh\E \,\XX = \E\, \wh\XX.
\end{equation}
\end{prop}
\begin{proof}
For $\XX = X_1(\zeta_1,z_1)\odot \cdots \odot X_n(\zeta_n,z_n), X_j = \pa_{\zeta_j}^{m_j} \bp_{\zeta_j}^{\wt m_j}\pa_{z_j} ^{n_j} \bp_{z_j}^{\wt n_j}\Phi(\zeta_j,z_j),$ it follows from the formula \eqref{eq: hat formula} and the rules \eqref{eq: hat rules} that
$$\wh\XX = \wh X_1(\zeta_1,z_1)\odot \cdots \odot \wh X_n(\zeta_n,z_n),$$ 
where $\wh X_j(\zeta_j,z_j) = X_j(\zeta_j,z_j)+ 2i\sum_k \beta_k\, \pa_{\zeta_j}^{m_j} \bp_{\zeta_j}^{\wt m_j}\pa_{z_j} ^{n_j} \bp_{z_j}^{\wt n_j}(G_{q_k,q_*}(\zeta_j) - G_{q_k,q_*}(z_j)).$
By Wick's formula and \eqref{def: wh E}, we obtain
$$\wh\E\,X_j(\zeta_j,z_j) = i\,\E\, X_j(\zeta_j,z_j) \Phi[\bfs\beta-\bfs\beta_*]= \E\, \wh X_j(\zeta_j,z_j).$$
By the definition \eqref{def: tensor product} of tensor product of Fock space functionals, 
\begin{align*}
\wh\E\,\XX &= \sum_{j=1}^\infty \frac{i^j}{j!}\,\E\,\Phi^{\odot j}[\bfs\beta-\bfs\beta_*] X_1(\zeta_1,z_1)\odot \cdots \odot X_n(\zeta_n,z_n)\\
&= \prod_{j=1}^n i\E\,\Phi[\bfs\beta-\bfs\beta_*] X_j(\zeta_j,z_j) = \prod_{j=1}^n \E\,\wh X_j(\zeta_j,z_j) = \E\,\wh\XX,
\end{align*}
which completes the proof.
\end{proof}

Theorem~\ref{main: Insertion} generalizes the previous proposition.
\begin{proof}[Proof of Theorem~\ref{main: Insertion}]
Recall the relation \eqref{eq: PPSs} between the modifications $\varphi_{\bfs\beta_1}$ and $\varphi_{\bfs\beta_2}$ in the case that $\supp\,\bfs\tau \cap (\supp\,\bfs\beta_1\cup\supp\,\bfs\beta_2) = \emptyset.$
This relation can be extended to the general case as
$$\varphi_{\bfs\beta_2}[\bfs\tau] = \varphi_{\bfs\beta_1}[\bfs\tau] + i\,\E\,\Phi[\bfs\beta_2-\bfs\beta_1]*\Phi[\bfs\tau]. $$
We rewrite it as 
$$\E\,\Phi_{\bfs\beta_2}[\bfs\tau] = \,\E\,\VV^{\odot}[\bfs\beta_2-\bfs\beta_1]*\Phi_{\bfs\beta_1}[\bfs\tau].$$
We define the correspondence $\XX \mapsto \wh\XX$ by the formula $\Phi_{\bfs\beta_1}[\bfs\tau]\mapsto\Phi_{\bfs\beta_2}[\bfs\tau]$ and the rules \eqref{eq: hat rules}.
We denote by $\wh\E[\XX]:=\E\,\VV^{\odot}[\bfs\beta_2-\bfs\beta_1]*\XX.$
Then applying the same argument as in the previous proposition, we arrive to $\wh\E \,\XX = \E\, \wh\XX.$
\end{proof}

Using Theorem~\ref{main: Insertion}, we present an alternate proof of Theorem~\ref{main: EO}.

\begin{proof}[Proof of Theorem~\ref{main: EO}]
Given a marked point $q_* \notin \supp\,\bfs\tau \cup \supp\,\bfs\beta,$ let $\bfs\beta_* = b\chi\cdot q_*.$
It follows from Theorem~\ref{main: Insertion} that
$$\E\,\OO_{\bfs\beta}[\bfs\tau] = \E\,\OO_{\bfs\beta_*}[\bfs\tau]*\VV^{\odot}[\bfs\beta-\bfs\beta_*].$$
By definitions of Wick's exponentials and the OPE exponentials,
$$\E\,\OO_{\bfs\beta_*}[\bfs\tau]*\VV^{\odot}[\bfs\beta-\bfs\beta_*] = \frac{\E\,\OO_{\bfs\beta_*}[\bfs\tau]*\OO_{\bfs\beta_*}[\bfs\beta-\bfs\beta_*]}{\E\,\OO_{\bfs\beta_*}[\bfs\beta-\bfs\beta_*]}.$$
The OPE nature of OPE exponentials (Proposition~\ref{OPE nature}) implies that 
$$\frac{\E\,\OO_{\bfs\beta_*}[\bfs\tau]*\OO_{\bfs\beta_*}[\bfs\beta-\bfs\beta_*]}{\E\,\OO_{\bfs\beta_*}[\bfs\beta-\bfs\beta_*]}=\frac{\E\,\OO_{\bfs\beta_*}[\bfs\tau+\bfs\beta-\bfs\beta_*]}{\E\,\OO_{\bfs\beta_*}[\bfs\beta-\bfs\beta_*]}.$$
By definition of the Coulomb gas correlation function $\CC_{(b)},$ we have
$$\frac{\E\,\OO_{\bfs\beta_*}[\bfs\tau+\bfs\beta-\bfs\beta_*]}{\E\,\OO_{\bfs\beta_*}[\bfs\beta-\bfs\beta_*]} = \frac{\CC_{(b)}[\bfs\tau+\bfs\beta]}{\CC_{(b)}[\bfs\beta]},$$
which completes the proof.
\end{proof}

For modifications of arg-type, we define 
$$\Phi_{(b)}(z,z_0) = \Phi(z,z_0) + \wt\varphi_*(z)-\wt\varphi_*(z_0),$$
where $\wt\varphi_*$ is a harmonic conjugate of $\varphi_*.$
Given a background charge $\bfs\beta,$ the correspondence $\XX\mapsto\wh\XX$ is defined by the rules \eqref{eq: hat rules} and the formula
$$\wh\Phi_{(b)}(z,z_0) = \Phi_{(b)}(z,z_0)- 2\sum \beta_j (\wt G_{q_j,q_*}(z) - \wt G_{q_j,q_*}(z_0)),$$
where $\wt G$ is the harmonic conjugate of the bipolar Green's function. 
Let us denote 
$$\wh\E[\XX]= \E\, [\ee^{\odot -\sum_j\beta_j\wt\Phi(q_j,q_*)}\XX],$$
where $\wt\Phi(z,z_0)$ is the harmonic conjugate of  $\Phi(z,z_0).$ 
More precisely, it is defined by the equation
$$2\Phi^+(z,z_0) =\Phi(z,z_0) + i \wt\Phi(z,z_0),$$
where
$$\Phi^+(z,z_0) = \Big\{\int_\gamma J(\zeta)\,\dd \zeta\,:\,\gamma \textrm{ is a curve from } z_0 \textrm{ to }z\Big\}.$$
Then we have \eqref{eq: hat hat}.

\subsection{KZ equations for OPE exponentials}
In Subsection~\ref{ss: Ward identities} we use certain types of the Knizhnik-Zamolodchikov equations (\emph{KZ equations}) \index{KZ equations} for the puncture operator $\PP_{\bfs\beta}$ to prove Theorem~\ref{main: Ward identity} (the Ward identities). 
In this subsection we derive certain KZ equations for the OPE exponentials $\OO_{\bfs\beta}[\bfs\tau].$
Let us consider the case $\supp\,\bfs\tau\,\cap\,\supp\,\bfs\beta=\emptyset$ only. 
The general case can be treated by applying the rooting procedure or normalization procedure.  

Recall that $\E\,\OO_{\bfs\beta}[\bfs\tau] = \ee^{i\varphi_{\bfs\beta}[\bfs\tau]} \CC[\bfs\tau]$ and $\CC[\bfs\tau]=\ee^{-c[\bfs\tau]}, (c[\bfs\tau]:=c_{\Phi[\bfs\tau]}).$
Thus we have 
\begin{equation} \label{eq: KZ0}
\frac{\pa_j\E\,\OO_{\bfs\beta}[\bfs\tau]}{\E\,\OO_{\bfs\beta}[\bfs\tau]} = i\pa_j \varphi_{\bfs\beta}[\bfs\tau] - \pa_jc[\bfs\tau]=i\tau_j j_{\bfs\beta}(z_j)-\pa_jc[\bfs\tau].
\end{equation}
Sometimes it is convenient to express $c[\bfs\tau]$ as a linear combination of the formal correlations: 
$$c[\bfs\tau] = \frac12\sum_j{\tau_j^2}\, \E\,\Phi*\Phi(z_j) + \sum_{k\ne j} \tau_j\tau_k\,\E\,\Phi(z_j)\Phi(z_k).$$
By differentiating $c[\bfs\tau]$ with respect $z_j,$ we have 
\begin{equation} \label{eq: KZ0a}
\pa_j c[\bfs\tau] = \tau_j^2\, \E\,J*\Phi(z_j) + \sum_{k\ne j} \tau_j\tau_k\,\E\,J(z_j)\Phi(z_k).
\end{equation}
A less formal representation 
$$\pa_j c[\bfs\tau] = \tau_j^2\, \E\,J*\Phi(z_j,z_0) + \sum_{k\ne j} \tau_j\tau_k\,\E\,J(z_j)\Phi(z_k,z_0)$$
does not depend on the choice of $z_0.$ 
Here, the OPE product is taken with respect to the $z_j$-variable. 
Let us consider the meromorphic differential $\omega$ as in Subsection~\ref{ss: C}. 
Then \eqref{eq: EJbeta} holds in any genus case,
\begin{equation} \label{eq: jbeta}
j_{\bfs\beta}(z) = ib\frac{\pa\omega(z)}{\omega(z)} + i\,\E\,J(z)\Phi[\bfs\beta-\bfs\beta_0],
\end{equation}
where $\bfs \beta_0 = -b\cdot(\omega).$ 
Combining \eqref{eq: KZ0}~--~\eqref{eq: jbeta}, we have
\begin{align} \label{eq: KZ}
\frac{\pa_j\E\,\OO_{\bfs\beta}[\bfs\tau]}{\E\,\OO_{\bfs\beta}[\bfs\tau]} &= -b\tau_j \frac{\pa\omega(z_j)}{\omega(z_j)} -\tau_j \E\,J(z_j)\Phi[\bfs\beta-\bfs\beta_0]\\
&-\tau_j^2\, \E\,J*\Phi(z_j) - \sum_{k\ne j} \tau_j\tau_k\,\E\,J(z_j)\Phi(z_k). \nonumber
\end{align}

\subsubsec{The genus zero case}
If $M = \wh\C, q_* = \infty,$ then in the identity chart of $\C,$ $\varphi_{q_*} = \log|\omega_{q_*}| = 0$ (up to an additive constant), and $\bfs\beta_0 = 2b\cdot q_\infty,$ $c(z) \equiv 0.$ 
By \eqref{eq: KZ}, we have 
\begin{equation} \label{eq: KZg0}
\pa_j\E\,\OO_{\bfs\beta}[\bfs\tau] = \Big(\sum_{k\ne j} \tau_j\tau_k\, \frac1{z_j-z_k}+\sum_k \tau_j \beta_k \frac1{z_j-q_k}\Big)\, \E\,\OO_{\bfs\beta}[\bfs\tau].
\end{equation}

\subsubsec{The genus one case}
In the identity chart of torus $\T_\Lambda,$ $\varphi_*$ is constant.
In addition, $\bfs\beta_0 = \bfs 0,$ and $c(z) = -\log|\theta'(0)|-2\pi(\Im\,z)^2/\Im\,\tau.$ 
By \eqref{eq: KZ}, we have  
\begin{align} \label{eq: KZ4torus}
\pa_j\E\,\OO_{\bfs\beta}[\bfs\tau] &= \Big(\sum_{k\ne j} \tau_j\tau_k\, \frac{\theta'(z_j-z_k)}{\theta(z_j-z_k)}-2\pi i \sum_{k} \tau_j\tau_k\frac{ \Im\,z_k}{\Im\,\tau} \Big)\, \E\,\OO_{\bfs\beta}[\bfs\tau]\\
&+ \Big(\sum_k \tau_j \beta_k \frac{\theta'(z_j-q_k)}{\theta(z_j-q_k)}- 2\pi i\sum_k \tau_j\beta_k \frac{\Im\,q_k}{\Im\,\tau}\Big)\, \E\,\OO_{\bfs\beta}[\bfs\tau]. \nonumber
\end{align}

\subsubsec{The higher genus case}
It is convenient to choose $\omega = \sum_j\pa_j\Theta(e)\omega_j.$
One can choose
$$c(z) = \log\frac1{|\omega(z)|} - 2\pi \langle (\Im\,\PM)^{-1}\Im\,\AA(z),\Im\,\AA(z)\rangle.$$
By \eqref{eq: KZ}, we have  
\begin{align*}
\frac{\pa_j\E\,\OO_{\bfs\beta}[\bfs\tau]}{\E\,\OO_{\bfs\beta}[\bfs\tau]} =\lambda_j \frac{\pa\omega(z_j)}{\omega(z_j)}&+ \sum_{k\ne j} \tau_j\tau_k\langle\frac{\nabla\Theta}{\Theta}(\AA(z_j)-\AA(z_k)-e),\overline{\vec\omega(z_j)}\rangle\\
&+\sum_k\tau_j \beta_k\langle\frac{\nabla\Theta}{\Theta}(\AA(z_j)-\AA(q_k)-e),\overline{\vec\omega(z_j)}\rangle\\
&+\sum_l \,\,\,\,b\tau_j \langle\frac{\nabla\Theta}{\Theta}(\AA(z_j)-\AA(p_l)-e),\overline{\vec\omega(z_j)}\rangle\\
&- 2\pi i \tau_j \langle (\Im\,\PM)^{-1}\vec\omega(z_j),\Im\,\AA[\bfs\tau+\bfs\beta-\bfs\beta_0]\rangle,
\end{align*}
where $\lambda_j = \frac 12 \tau_j^2-b\tau_j$ and $\bfs\beta_0 = -b\cdot (\omega).$

\section{Ward's equations} \label{s: Ward}
 
Given a background charge $\bfs\beta,$ we define the Virasoro pair for the OPE family $\FF_{\bfs\beta}(M^*)$ ($M^*:=M\setminus\bfs q, \bfs q:=\supp\,\bfs\beta$) of $\Phi_{\bfs\beta}$ as a pair of Schwarzian forms such that the residue operators associated with these forms represent the Lie derivative operators in applications to the fields in the OPE family $\FF_{\bfs\beta}(M^*).$ 
We extend this representation (Ward's identity) to the OPE family $\FF_{\bfs\beta}(M)$ using the puncture operator $\PP_{\bfs\beta}:=\CC_{(b)}[\bfs\beta],$ see Theorem~\ref{main: Ward identity}.

\subsection{Stress tensors and Virasoro fields} 
Let $v$ be a non-random local smooth vector field and $\psi_t$ be the local flow associated with $v.$ 
Let us recall the definition of the \emph{Lie derivative} \index{Lie derivative} $\LL_v X$ of a Fock space field $X:$
$$(\LL_v X \| \phi) (z) = \frac d{dt}\Big|_{t=0} (X\|\phi\circ \psi_{-t})(z).$$
Set 
$$\LL_v^\pm = \frac{\LL_v \mp i \LL_{iv}}2$$
so that $\LL_v = \LL_v^+ + \LL_v^-.$ 
Whereas $\LL_v$ depends $\R$-linearly on $v,$ its holomorphic part $\LL_v^+$ depends $\C$-linearly on $v.$

\begin{prop}[Cf. Proposition 4.1 in \cite{KM13}]
If $X$ is a $(\lambda,\lambda_*)$-differential, then
\begin{equation} \label{eq: Lv+diff}
\LL_v^+X =  (v\pa + \lambda v')X, \qquad \LL_v^-X =  (\bar v\bp + \lambda_* \overline{v'})X;
\end{equation}
if $\varphi$ is a $\PPS(ib,ib),$ then 
$$\LL_v^+\varphi = v\pa\varphi + ib v',\qquad \LL_v^-\varphi = \bar v\bp\varphi + ib \overline{v'};$$
if $\varphi$ is a $\PPS(ib,-ib),$ then 
$$\LL_v^+\varphi = v\pa\varphi + ib v',\qquad \LL_v^-\varphi = \bar v\bp\varphi - ib \overline{v'};$$
if $X$ is a pre-Schwarzian form of order $\mu,$ then
$$\LL_vX  = \LL_v^+X =  (v\pa + v')X + \mu\,v'';$$
if $X$ is a Schwarzian form of order $\nu,$ then
\begin{equation} \label{eq: Lv+S}
\LL_vX  = \LL_v^+X =  (v\pa + 2v')X + \nu\,v'''.
\end{equation}
\end{prop}
By definition, a Fock space field $X$ on $M$ has a \emph{stress tensor} \index{stress tensor} $(A,\bar A_*)$ 
if $(A,A_*)$ is a pair of holomorphic quadratic differentials and if
$$\LL_v^+ X = A_v X, \qquad \LL_v^+ \bar X = (A_*)_v \bar X$$
hold in $D_{\hol}(v)$ for all non-random local vector fields $v.$
Here, $D_\hol(v)$ is a maximal open set where $v$ is holomorphic and $A_v$ is the residue operator on Fock space fields $(A_v X)(z) = A_v(z)X(z)$ defined by  
$$A_v(z): X(z) \mapsto \frac1{2\pi i} \oint_{(z)} vA\,X(z).$$
We denote by $\FF(A,\bar A_*)$ the collection of Fock space fields which have a common stress tensor $(A,\bar A_*).$ 
By definition, a pair $(T,\bar T_*)$ of Fock space fields is the Virasoro pair for the family $\FF(A,\bar A_*)$ if 
$T, \bar T_* \in \FF(A,\bar A_*)$ and if $T-A, T_*-A_*$ are non-random holomorphic (or meromorphic with poles where background charges are placed) Schwarzian forms.
We remark that the family $\FF(A,\bar A_*)$ is closed under OPE products $*_n,$ see \cite[Proposition~5.8]{KM13}.

\begin{thm} \label{SET} 
Let $\varphi_{\bfs\beta}$ be a simple PPS form with the background charge $\bfs\beta(=i\pa\bp\varphi_{\bfs\beta}/\pi).$
\renewcommand{\theenumi}{\alph{enumi}}
{\setlength{\leftmargini}{1.8em}
\begin{enumerate}
\item If $\varphi_{\bfs\beta}$ is a PPS form of $\log|\cdot|$-type, then $\Phi_{\bfs\beta}$ has a stress energy tensor $(A_{\bfs\beta},\bar A^*_{\bfs\beta}):$ 
$$A_{\bfs\beta} =A+ib\pa J -j_{\bfs\beta}J,
\quad A^*_{\bfs\beta} = A-ib\pa J -j_{\bfs\beta}J,$$
where $A = -\frac12 J \odot J, j _{\bfs\beta}= \pa\varphi_{\bfs\beta}.$
The Virasoro pair $(T_{\bfs\beta},\bar T^*_{\bfs\beta})$ for the family $\FF(A_{\bfs\beta},\bar A^*_{\bfs\beta})$ is given by 
$$T_{\bfs\beta} = -\frac12 (\pa\Phi_{\bfs\beta}*\pa\Phi_{\bfs\beta}) + ib\pa^2\Phi_{\bfs\beta}, 
\quad
T^*_{\bfs\beta} = -\frac12 (\pa\bar\Phi_{\bfs\beta}*\pa\bar\Phi_{\bfs\beta}) - ib\pa^2\bar\Phi_{\bfs\beta}.$$ 
They are (holomorphic) Schwarzian forms of order $c/12 = 1/12 - b^2.$
\item If $\varphi_{\bfs\beta}$ is a PPS form of arg-type, then $\Phi_{\bfs\beta}$ has a stress energy tensor $(A_{\bfs\beta},\bar A_{\bfs\beta}):$ 
$$ A_{\bfs\beta} =A+ib\pa J -j_{\bfs\beta}J.$$
The Virasoro pair $(T_{\bfs\beta},\bar T_{\bfs\beta})$ for the family $\FF(A_{\bfs\beta},\bar A_{\bfs\beta})$ is given by 
$$T_{\bfs\beta} = -\frac12 (\pa\Phi_{\bfs\beta}*\pa\Phi_{\bfs\beta}) + ib\pa^2\Phi_{\bfs\beta}.$$
\end{enumerate}
}
\end{thm}

\begin{proof}
First, we claim that $(ib\pa -j_{\bfs\beta})J$ and thus $A_{\bfs\beta}$ are holomorphic quadratic differentials. 
This claim is immediate from the following transformation laws,
\begin{align*}
ib\pa J &= ibh'' \ti J\circ h + ib(h')^2 \pa \ti J\circ h, \\ 
j_{\bfs\beta}J &= ib \Big(\frac{h''}{h'}\Big)h' \ti J \circ h +(h')^2 (\ti \jmath_{\bfs\beta} \ti J) \circ h.
\end{align*}
From the expressions of $A_{\bfs\beta}$ and $T_{\bfs\beta},$ we have
\begin{equation} \label{eq: Tbeta}
T_{\bfs\beta} = T +(A_{\bfs\beta}-A)- \frac12 j_{\bfs\beta}^2 + ib\pa j_{\bfs\beta}.
\end{equation}
To show that $T_{\bfs\beta}$ is a Schwarzian form of order $c/12,$ we need to check that $-\frac12 j_{\bfs\beta}^2 + ib\pa j_{\bfs\beta}$ is a Schwarzian form of order $-b^2.$
It follows from the transformation laws for $-\frac12 j_{\bfs\beta}^2$ and $ib\pa j_{\bfs\beta}:$ 
\begin{align*}
-\frac12 j_{\bfs\beta}^2 &= -ib h'' \ti \jmath_{\bfs\beta}\circ h -\frac12 (h')^2(\ti \jmath_{\bfs\beta}\circ h)^2 + \frac12 b^2 \Big(\frac{h''}{h'}\Big)^2,\\
ib\pa j_{\bfs\beta} &= \phantom{-} ib h'' \ti \jmath_{\bfs\beta}\circ h + ib (h')^2 (\pa \ti \jmath_{\bfs\beta}\,) \circ h - b^2 \Big(\frac{h''}{h'}\Big)'.
\end{align*}

We need to check Ward's OPE:
$$A_{\bfs\beta}(\zeta) \Phi_{\bfs\beta}(z,z_0) \sim \frac{J_{\bfs\beta}(z)}{\zeta-z} + ib\,\frac1{(\zeta-z)^2}.$$
It follows immediately from the singular OPEs in the case $b=0,$
$$A(\zeta) \Phi(z,z_0) \sim \frac{J(z)}{\zeta-z}, \qquad J(\zeta)\Phi(z,z_0)\sim -\frac1{\zeta-z}.$$
By differentiating the last OPE with respect to $\zeta,$ we find
$$\pa J(\zeta)\Phi(z,z_0)\sim \frac1{(\zeta-z)^2}.$$
Multiplying the singular OPE of $J(\zeta)\Phi(z,z_0)$ by $-j_{\bfs\beta}(\zeta)$ gives rise to 
$$-j_{\bfs\beta}(\zeta)J(\zeta)\Phi(z,z_0)\sim \frac{j_{\bfs\beta}(z)}{\zeta-z}.$$
We define $A_{\bfs\beta}^-$ by  
$$ A_{\bfs\beta}^- = \bar A_{\bfs\beta}^*, \quad A_{\bfs\beta}^* = A-ib\pa J -j_{\bfs\beta}J.$$
Then we have Ward's OPE:
$$A_{\bfs\beta}^-(\zeta) \Phi_{\bfs\beta}(z,z_0) \sim \overline{\frac{J_{\bfs\beta}(z)}{\zeta-z} - ib\,\frac1{(\zeta-z)^2}} = \frac{\bar J_{\bfs\beta}(z)}{\bar \zeta-\bar z} + ib\,\frac1{(\bar\zeta-\bar z)^2},$$
which completes the first part of theorem. 

We leave the second part to the reader as an exercise. 
\end{proof}

\begin{rmk*}
More generally, one can consider the following modification of mixed type
$$\varphi_{\bfs\beta} = \beta_+\varphi_{\bfs\beta}^+ + \beta_-\varphi_{\bfs\beta}^-, \quad (\beta_++\beta_-=1),$$ 
where $\varphi_{\bfs\beta}^+$ is a $\PPS(ib,ib)$ form and $\varphi_{\bfs\beta}^-$ is a $\PPS(ib,-ib)$ form. 
By definition, $\varphi_{\bfs\beta}$ is a $\PPS(ib,ib(\beta_+-\beta_-))$ form.
Then we have a stress tensor $(A_{\bfs\beta},\bar A_{\bfs\beta}^*)$ for the bosonic field $\Phi_{\bfs\beta}:=\Phi+\varphi_{\bfs\beta}:$
\begin{align*}
A_{\bfs\beta} &=A+ib\pa J -j_{\bfs\beta}J,\\
A_{\bfs\beta}^* &=  A- ib(\beta_+-\beta_-)\pa J - j_{\bfs\beta} J.
\end{align*}
Indeed, we have Ward's OPE:
$$\bar A_{\bfs\beta}^*(\zeta) \Phi_{\bfs\beta}(z,z_0) \sim \overline{\frac{J_{\bfs\beta}(z)}{\zeta-z} - ib(\beta_+-\beta_-)\,\frac1{(\zeta-z)^2}} = \frac{\bar J_{\bfs\beta}(z)}{\bar \zeta-\bar z} + ib(\beta_+-\beta_-)\,\frac1{(\bar\zeta-\bar z)^2}.$$
In this case, the Virasoro pair $(T_{\bfs\beta},\bar T^*_{\bfs\beta})$ for the family $\FF(A_{\bfs\beta},\bar A^*_{\bfs\beta})$ is given by 
$$T_{\bfs\beta} = -\frac12 (\pa\Phi_{\bfs\beta}*\pa\Phi_{\bfs\beta}) + ib\pa^2\Phi_{\bfs\beta}, 
\quad
T^*_{\bfs\beta} = -\frac12 (\pa\bar\Phi_{\bfs\beta}*\pa\bar\Phi_{\bfs\beta}) - ib(\beta_+-\beta_-) \pa^2\bar\Phi_{\bfs\beta}.$$ 
\end{rmk*}

The following singular operator product expansions (Ward's OPEs) can be derived by using the method in \cite[Corollaries~5.4~--~5.5]{KM13}.

\begin{prop}
As $\zeta\to z,$ we have 
\renewcommand{\theenumi}{\alph{enumi}}
{\setlength{\leftmargini}{1.8em}
\begin{enumerate}
\ms \item 
$T_{\bfs\beta}(\zeta)\Phi_{\bfs\beta}(z,z_0)\sim \dfrac{ib}{(\zeta-z)^2} + \dfrac{J_{\bfs\beta}(z)}{\zeta-z},$
\ms \item 
$T_{\bfs\beta}(\zeta)J_{\bfs\beta}(z)\sim \dfrac{2ib}{(\zeta-z)^3} + \dfrac{J_{\bfs\beta}(z)}{(\zeta-z)^2} + \dfrac{\partial J_{\bfs\beta}(z)}{\zeta-z},$
\ms \item
$T_{\bfs\beta}(\zeta)T_{\bfs\beta}(z)\sim \dfrac{c/2}{(\zeta-z)^4}+\dfrac{2T_{\bfs\beta}(z)}{(\zeta-z)^2}+ \dfrac{\partial T_{\bfs\beta}(z)}{\zeta-z},$
\ms \item $T_{\bfs\beta}(\zeta) \OO_{\bfs\beta}[\bfs\tau] \sim ~\dfrac {\lambda\OO_{\bfs\beta}[\bfs\tau]}{(\zeta-z)^2}+\dfrac{\partial \OO_{\bfs\beta}[\bfs\tau]}{\zeta-z}, \quad(\bfs\tau = \tau\cdot z + \cdots, \,\,\lambda = \dfrac{\tau^2}2-b\tau, \, \displaystyle\int\bfs\tau = 0).$
\end{enumerate}
}
\end{prop}

\subsection{Ward's equations without modifications} \label{ss: Ward's equations 0}
In this subsection we derive Ward's equations in the case $b=0.$ 
Suppose a meromorphic vector field $v$ has poles at $\xi_k$'s.
For a string $\XX$ of fields with $S_\XX = \{z_1,\cdots,z_n\},$ 
the map 
$$\xi \mapsto \E\,v(\xi)A(\xi)\XX$$
is a homomorphic 1-differential on $M\sm\bigcup_{j,k}\{z_j,\xi_k\}$ and hence its residues at $z_j,\xi_k$ are well-defined.
We define the Ward functional $W_v^+$ by 
\begin{equation} \label{eq: W def0}
\E\,W_v^+\XX = - \sum_k \Res_{\,\xi_k} v\,\E\,A\XX
\end{equation}
for any tensor product $\XX$ of Fock space fields such that $S_\XX$ does not intersect the set of poles of $v.$ 

\begin{prop} \label{prop: Ward}
For any tensor ptoduct $\XX$ of fields in the OPE family $\FF$ of the Gaussian free field, we have
$$\E\,W_v^+\XX = \E\, \LL_v^+\XX.$$ 
\end{prop}

\begin{proof}
Suppose $C_j \,(C'_k)$ is a positively oriented simple closed curve that bounds a sufficiently small open disk $D_j$ ($D'_k$) such that $z_j\in D_j \,(\xi_k\in D'_k)$ and $C_j \,(C'_k)$ has winding number 1 about $z_j \,(\xi_k),$ respectively. 
Applying Stokes' theorem to 
$$d\big(v(z)\E\,A(z)\XX\,\dd z\big) = -\bp\big(v(z)\E\,A(z)\XX\big)\,\dd z\wedge \dd \bar z = 2i \bp\big(v(z)\E\,A(z)\XX\big) \,\dd x \wedge \dd y$$
over $M \sm (\bigcup_j D_j \cup \bigcup_k D'_k),$ we have 
$$\E\,W_v^+\XX = \frac1{2\pi i} \sum_j \oint_{(z_j)} \E\,vA\XX.$$ 
Now proposition follows from a residue form of Ward's identity: 
$$\frac1{2\pi i} \sum_j \oint_{(z_j)} \E\,vA\XX =\sum_j \E\, \LL_v^+(z_j)\XX = \E\,\LL_v^+\XX.$$
\end{proof}

\begin{cor} \label{cor: Ward}
If $X_j$'s are $(\lambda_j,\lambda_{*j})$-differentials in the OPE family $\FF,$ then 
$$- \sum_k \E\,\Res_{\,\xi_k}(vA)\XX = \sum_j (v(z_j)\pa_j + \lambda_j v'(z_j))\E\,\XX,$$
where $\XX = X_1(z_1)\cdots X_n(z_n).$ 
\end{cor}
\begin{proof}
It is clear from Leibniz's rule for the Lie derivative and the formula~\eqref{eq: Lv+diff} for the Lie derivative of a differential.  
\end{proof}

\subsubsec{The genus zero case}
Given $\xi\in\C,$ let us consider the vector field $v_\xi$ given by 
$$k_\xi(z) = \frac1{\xi-z}$$
in the identity chart of $\C.$
Its continuous extension to $\wh\C$ has a triple zero at infinity. 
Applying Proposition~\ref{prop: Ward} to the vector field $k_\xi,$ we obtain the following proposition. 
\begin{prop}
If the fields $X_j$ belong to the OPE family $\FF,$ then in the $\wh\C$-uniformization, 
\begin{equation} \label{eq: Ward4C}
\E\,A(\xi)\XX = \E\, \LL_{k_\xi}^+\XX,
\end{equation}
where $\XX = X_1(z_1)\cdots X_n(z_n).$ 
\end{prop} 

\begin{eg*} Let us consider $\XX=\OO[\bfs\tau],$ where a divisor $\bfs\tau$ satisfies the neutrality condition $(\NC_0).$ 
By Wick's calculus, we find the left-hand side in \eqref{eq: Ward4C} as 
$$\E\,A(\xi)\XX = \frac12 \big(\E\,J(\xi)\Phi[\bfs\tau]\big)^2\E\,\XX = \Big(\sum_j \frac12\frac{\tau_j^2}{(\xi-z_j)^2} +\sum_{j<k} \frac{\tau_j\tau_k}{(\xi-z_j)(\xi-z_k)} \Big)\E\,\XX.$$
It follows from the KZ equations that the right-hand side in \eqref{eq: Ward4C} is 
$$\LL^+_{k_\xi}\E\,\XX= \Big(\sum_j \frac{\lambda_j}{(\xi-z_j)^2} + \sum_j\sum_{k\ne j} \frac{\tau_j\tau_k}{(\xi-z_j)(z_j-z_k)}\Big)\E\,\XX,$$
where $\lambda_j = \frac12 \tau_j^2.$
We now verify \eqref{eq: Ward4C} for $\XX=\OO[\bfs\tau]:$
\begin{align*}\sum_{j<k} \frac{\tau_j\tau_k}{(\xi-z_j)(\xi-z_k)} &=\sum_{j<k} \tau_j\tau_k\Big(\frac1{\xi-z_j}-\frac1{\xi-z_k}\Big)\frac1{z_j-z_k}\\
&=\sum_{j<k} \frac{\tau_j\tau_k}{(\xi-z_j)(z_j-z_k)}+\sum_{k>j} \frac{\tau_j\tau_k}{(\xi-z_k)(z_k-z_j)}\\
&=\sum_{j<k} \frac{\tau_j\tau_k}{(\xi-z_j)(z_j-z_k)}+\sum_{j>k} \frac{\tau_j\tau_k}{(\xi-z_j)(z_j-z_k)}\\
&=\sum_{j\ne k} \frac{\tau_j\tau_k}{(\xi-z_j)(z_j-z_k)}.
\end{align*}
\end{eg*}

\subsubsec{The genus one case}
Given $\xi,\xi_0\in\T_\Lambda,$ let us consider the vector field $v_{\xi,\xi_0}$ given by 
$$ v_{\xi,\xi_0}(z) = \frac{\theta'(\xi-z)}{\theta(\xi-z)}-\frac{\theta'(\xi_0-z)}{\theta(\xi_0-z)}$$
in the identity chart of $\T_\Lambda.$ 
Applying Proposition~\ref{prop: Ward} to the vector field $v_{\xi,\xi_0},$ we obtain the following proposition. 
\begin{prop} \label{Ward4g1a}
If the fields $X_j$ belong to the OPE family $\FF,$ then in the $\T_\Lambda$-uniformization,
\begin{equation} \label{eq: WardTorus1}
\E\,(A(\xi)-A(\xi_0))\XX = \E\, \LL_{v_{\xi,\xi_0}}^+\XX,
\end{equation}
where $\XX = X_1(z_1)\cdots X_n(z_n).$ 
\end{prop}

Given $\xi\in\T_\Lambda,$ let us consider the vector field $v_{\xi}$ on the torus $\T_\Lambda$ given by 
$$v_\xi(z) = -\wp(\xi-z).$$ 
Applying Proposition~\ref{prop: Ward} to the vector field $v_\xi,$ we obtain the following proposition. 
\begin{prop}
If the fields $X_j$ belong to the OPE family $\FF,$ then in the $\T_\Lambda$-uniformization,
\begin{equation} \label{eq: WardTorus2}
\E\,\pa A(\xi)\XX = \E\, \LL_{v_\xi}^+\XX,
\end{equation}
where $\XX = X_1(z_1)\cdots X_n(z_n).$ 
\end{prop}

\begin{eg*}
Ward's equation~\eqref{eq: WardTorus2} with $\XX = J(z)J(z_0)$ gives rise to the addition theorem for Weierstrass $\wp$-function: \index{addition theorem!for Weierstrass $\wp$-function}
$$\begin{vmatrix}
1 & 1 & 1 \\
\wp(\xi-z) & \wp(z-z_0) & \phantom{-}\wp(\xi-z_0) \\
\wp'(\xi-z) & \wp'(z-z_0) & -\wp'(\xi-z_0) \\
\end{vmatrix} = 0.$$
Recall that in the $\T_\Lambda$-uniformization, $\E\,J(z)J(z_0) = -\wp(z-z_0)-2\,\E\,T,$ where $\E\,T$ is constant. 
By Wick's calculus, we find the left-hand side in \eqref{eq: WardTorus2} as 
$$\pa_\xi \E\,A(\xi)\XX = -\pa_\xi\big(\wp(\xi-z)\wp(\xi-z_0) +2\,\E\,T(\wp(\xi-z)+\wp(\xi-z_0))\big). $$
As a Lie derivative of $\E\,\XX,$ the right-hand side in \eqref{eq: WardTorus2} is found as 
\begin{align*}
\LL^+_{v_\xi}\E\,\XX 
&= \wp(\xi-z)\pa_z\wp(z-z_0) - \wp'(\xi-z)(\wp(z-z_0)+2\,\E\,T)\\
&+ \wp(\xi-z_0)\pa_{z_0}\wp(z-z_0) - \wp'(\xi-z_0)(\wp(z-z_0)+2\,\E\,T).
\end{align*}
\end{eg*}

\begin{eg*} 
It is well known that the Weierstrass $\wp$-function satisfies the following differential equation
\begin{equation} \label{eq: wpwp'}
\wp\wp' = \frac1{12} \wp'''.
\end{equation}
Ward's equation~\eqref{eq: WardTorus2} with $\XX = T(z)$ gives rise to the above formula.
Applying Wick's calculus to the left-hand side of \eqref{eq: WardTorus2}, we have 
$$\pa_\xi\E\,A(\xi)T(z) = \frac12\pa_\xi (\E\,J(\xi)J(z))^2 = (\wp(\xi-z)+2\,\E\,T)\wp'(\xi-z).$$ 
On the other hand, as a Lie derivative of Schwarzian form $\E\,T(z)$ of order $1/12,$ the right-hand side in \eqref{eq: WardTorus2} is found as 
$$\LL_{v_\xi} \E\,T(z) = 2 \, \wp'(\xi-z)\,\E\,T + \frac1{12}\wp'''(\xi-z).$$ 
\end{eg*}

Before presenting the next example with $\XX = \OO[\tau\cdot z_1- \tau\cdot z_2],$ let us recall some basic properties of \emph{Weierstrass $\zeta$-function and $\sigma$-function}. \index{Weierstrass $\zeta$-function} \index{Weierstrass $\sigma$-function}
(See \cite{Chandrasekharan} for more details.)
Weierstrass $\zeta$-function
$$\zeta(z):=\frac1z + \sum_{\lambda\in\Lambda\setminus\{0\}} \Big(\frac1{z-\lambda} + \frac1\lambda + \frac z{\lambda^2}\Big)$$
is a meromorphic odd function which has simple poles at $\lambda\in\Lambda$ with residue 1.
Although it is not an elliptic function, every elliptic function can be expressed by Weierstrass $\zeta$-function and its derivatives, e.g. $\wp = -\zeta'.$
Weierstrass $\zeta$-function has the following quasi-periodicities
\begin{equation} \label{eq: zeta periodicities}
\zeta(z + m + n\tau) =  \zeta(z) + 2m\eta_1 + 2n\eta_2, \qquad \big(\eta_1 = \zeta(1/2),\quad \eta_2 = \zeta(\tau/2)\big).
\end{equation}
It is well known that
\begin{equation} \label{eq: eta1}
\eta_1 = -\frac16 \frac{\theta'''(0)}{\theta'(0)}.
\end{equation}
Later, we derive \eqref{eq: eta1} from Ward's equation, see the example below Theorem~\ref{lem: EOKM0}.

Weierstrass $\sigma$-function
$$\sigma(z):=z \prod_{\lambda\in\Lambda\setminus\{0\}} \Big(1-\frac z\lambda\Big)\ee^{\frac z\lambda + \frac{z^2}{2\lambda^2}}$$
is an entire odd function which has simple zeros at $\lambda\in\Lambda.$ 
It is related to Weierstrass $\zeta$-function as $\zeta = \sigma'/\sigma$ and has the following quasi-periodicities
\begin{equation} \label{eq: sigma periodicities}
\sigma(z + 1) =  -\sigma(z)\ee^{2\eta_1(z + \frac12)},\qquad \sigma(z+\tau)  = -\sigma(z)\ee^{2\eta_2(z + \frac12\tau)}.
\end{equation}
Weierstrass $\sigma$-function and Jacobi $\theta$-function are related as
\begin{equation} \label{eq: sigma theta}
\sigma(z) = \frac{\theta(z)}{\theta'(0)} \ee^{\eta_1z^2}.
\end{equation}
Thus we find the relation between Jacobi $\theta$-function and Weierstrass $\zeta$-function as 
\begin{equation} \label{eq: eta theta}
\frac{\theta'(z)}{\theta(z)} = \zeta(z) - 2\eta_1 z.
\end{equation}

\begin{eg*} 
Ward's equation~\eqref{eq: WardTorus2} with $\XX = \OO[\tau\cdot z_1- \tau\cdot z_2]$ gives rise to the addition theorem for Weierstrass $\zeta$-function: \index{addition theorem! for Weierstrass $\zeta$-function}
\begin{equation} \label{eq: addition4zeta}
\zeta(z+w) = \zeta(z) + \zeta(w) + \frac12\,\frac{\zeta''(z)-\zeta''(w)}{\zeta'(z)-\zeta'(w)}.
\end{equation}
For $\XX = \OO[\bfs\tau],$ application of Wick's calculus to the left-hand side of \eqref{eq: WardTorus2} leads to 
\begin{align*}
\pa_\xi \E\,A(\xi)\XX &= \frac12\pa_\xi\big(\E\,J(\xi)\Phi[\bfs\tau]\big)^2 \E\,\XX\\
&=\Big(-\sum\tau_j\frac{\theta'(\xi-z_j)}{\theta(\xi-z_j)} + \frac{2\pi i}{\Im\,\tau} \sum \tau_j\Im\,z_j\Big)\Big(\sum\tau_j\wp(\xi-z_j)\Big)\E\,\XX.\
\end{align*}
On the other hand, as the $\C$-linear part of Lie-derivative of the differential $\E\,\XX$, we find the right-hand side in \eqref{eq: WardTorus2} as 
\begin{align*}
\LL^+_{v_\xi} \E\,\XX &= \sum \frac12\tau_j^2\wp'(\xi-z_j)\E\,\XX\\
&-\sum_j \tau_j\wp(\xi-z_j)\Big(\sum_{k\ne j}\tau_k\frac{\theta'(z_j-z_k)}{\theta(z_j-z_k)}-2\pi i \sum_k\tau_k\, \frac{\Im\,z_k}{\Im\,\tau}\Big)\E\,\XX.
\end{align*}
Thus we have 
\begin{equation} \label{eq: Ward4EO}
\sum_j \frac12\tau_j^2\wp'(\xi-z_j)\ =  -\sum_j\tau_j\wp(\xi-z_j)\Big(\sum\tau_j\frac{\theta'(\xi-z_j)}{\theta(\xi-z_j)}-\sum_{k\ne j}\tau_k\frac{\theta'(z_j-z_k)}{\theta(z_j-z_k)}\Big).
\end{equation}
In particular, a choice of $\bfs\tau = \tau\cdot z_1 - \tau\cdot z_2$ gives 
\begin{align*}
\frac12\Big(&\wp'(\xi-z_1) + \wp'(\xi-z_2)\Big) \\
&= \Big(\wp(\xi-z_1) - \wp(\xi-z_2)\Big) \Big(-\frac{\theta'(\xi-z_1)}{\theta(\xi-z_1)}+\frac{\theta'(\xi-z_2)}{\theta(\xi-z_2)} -\frac{\theta'(z_1-z_2)}{\theta(z_1-z_2)}\Big).
\end{align*}
Setting $z = \xi-z_1$ and $w = -\xi+z_2,$ we have the addition theorem for $\theta'/\theta:$
$$\frac{\theta'z+w)}{\theta(z+w)} = \frac{\theta'(z)}{\theta(z)} + \frac{\theta'(w)}{\theta(w)} + \frac12\,\frac{\wp'(z)-\wp'(w)}{\wp(z)-\wp(w)}.$$
The addition theorem~\eqref{eq: addition4zeta} for Weierstrass $\zeta$-function now follows from \eqref{eq: wp theta} and \eqref{eq: eta theta}.
Let us apply the addition theorem for $\theta'/\theta$ to \eqref{eq: Ward4EO}: 
\begin{align*}
\sum_j &\frac12\tau_j^2\wp'(\xi-z_j) \\
&= -\sum_j\tau_j\wp(\xi-z_j)\Big(\tau_j\frac{\theta'(\xi-z_j)}{\theta(\xi-z_j)}+\sum_{k\ne j}\tau_k\big(\frac{\theta'(\xi-z_k)}{\theta(\xi-z_k)}-\frac{\theta'(z_j-z_k)}{\theta(z_j-z_k)}\big)\Big).
\end{align*}
The right-hand side simplifies to
$$-\sum_j\tau_j\wp(\xi-z_j)\Big(\tau_j\frac{\theta'(\xi-z_j)}{\theta(\xi-z_j)}+\sum_{k\ne j}\tau_k\big(\frac{\theta'(\xi-z_j)}{\theta(\xi-z_j)}+\frac12  \frac{\wp'(\xi-z_j) + \wp'(\xi-z_k)}{\wp(\xi-z_j) - \wp(\xi-z_k)}\big)\Big).$$
By the neutrality condition $\sum_j \tau_j = 0,$ we have 
\begin{equation} \label{eq: Ward4Owp}
\sum_j \tau_j^2 \wp'(\xi-z_j) = - \sum_j\sum_{k\ne j}\tau_j\tau_k \frac{\wp'(\xi-z_j) + \wp'(\xi-z_k)}{\wp(\xi-z_j) - \wp(\xi-z_k)}\wp(\xi-z_j).
\end{equation}
\end{eg*}

\subsection{Ward identities} \label{ss: Ward identities}
Given a background charge $\bfs\beta,$ let $\bfs q :=\supp\,\bfs\beta.$ 
Recall the definition of the puncture operator $\PP_{\bfs \beta},$ \index{puncture operator}
$$\PP_{\bfs \beta}:=\CC_{(b)}[\bfs\beta] = \E\,\OO_{\bfs\beta_*}[\bfs\beta-\bfs\beta_*],$$
where $\bfs\beta_* = b\chi\cdot q_*$ and $q_*$ is a reference point. 
Given a non-random meromorphic vector field $v$ with poles at $\xi_k$'s, we define the Ward functional $W_v^+$ by 
$$W_v^+ : = -\frac1{2\pi i}\sum_k \oint_{(\xi_k)} vT_{\bfs \beta}$$
in correlations with any tensor product of Fock space fields. 
To prove Theorem~\ref{main: Ward identity}, we need the following Lemma. 
It states that $W_v^+ = \PP_{\bfs \beta}^{-1}\LL_v^+\PP_{\bfs \beta}$ in application to the trivial string $\XX\equiv 1$ within correlation.
\begin{lem}
For each $q_k\in \bfs q,$ we have 
\begin{equation} \label{eq: res vET}
\frac1{2\pi i}\oint_{(q_k)} v\,\E\,T_{\bfs\beta} = \lambda(\beta_k)v'(q_k) + v(q_k) \PP_{\bfs \beta}^{-1}\pa_{q_k}\PP_{\bfs \beta}.
\end{equation}
\end{lem}
\begin{proof}
Recall \eqref{eq: ETbeta}, $\E\,T_{\bfs\beta} = \E\, T -\frac12 j_{\bfs\beta} ^2 + ib\pa j_{\bfs\beta},$
where $j_{\bfs\beta} = \E\,J_{\bfs\beta}.$ 
As we mentioned in the last example of Subsection~\ref{ss: T}, 
$\E\, T$ or the Bergman projective connection is holomorphic. 
Thus $\E\, T$ has no contribution to \eqref{eq: res vET}.
The 1-point function $j_{\bfs\beta}$ (see \eqref{eq: jbeta} for its formula) has the following expansion: as $z\to q_k,$ 
$$j_{\bfs\beta}(z) = i\Big( -\frac{\beta_k}{z-q_k} + b \frac{\pa\omega(q_k)}{\omega(q_k)} + \E\,J(q_k)*\Phi[\bfs\beta-\bfs\beta_0] + o(1) \Big).$$
By the above expansion and \eqref{eq: ETbeta}, we find 
$$\frac1{2\pi i}\oint_{(q_k)} v\,\E\,T_{\bfs\beta} = \lambda(\beta_k)v'(q_k) - v(q_k) \Big(b\beta_k \frac{\pa\omega(q_k)}{\omega(q_k)}+\beta_k\E\,J(q_k)*\Phi[\bfs\beta-\bfs\beta_0]\Big).$$
It follows from the KZ equations~\eqref{eq: KZ} that 
\begin{align*}
 \PP_{\bfs \beta}^{-1}\pa_{q_k} \PP_{\bfs \beta}& = \frac{\pa_{q_k}\E\,\OO_{\bfs\beta_*}[\bfs\beta-\bfs\beta_*]}{\E\,\OO_{\bfs\beta_*}[\bfs\beta-\bfs\beta_*]} \\
&= - b\beta_k \frac{\pa\omega(q_k)}{\omega(q_k)}-\beta_k\E\,J(q_k)\Phi[\bfs\beta_*-\bfs\beta_0]-\beta_k\E\,J(q_k)*\Phi[\bfs\beta-\bfs\beta_*]\\
&= -b\beta_k \frac{\pa\omega(q_k)}{\omega(q_k)}-\beta_k\E\,J(q_k)*\Phi[\bfs\beta-\bfs\beta_0],
\end{align*}
which completes the proof. 
\end{proof}

\begin{def*}
We say $X_{q}\in \FF_{\bfs\beta}(q)$ if $X_q = X_1*(\cdots*(X_n*1_q)\cdots)$ for some fields $X_1,\cdots X_n\in \FF_{\bfs\beta}(M^*).$
The extended OPE family $\FF_{\bfs\beta}(M)$ is the collection of Fock space strings
$$\{\YY = \XX X_{q_1} \cdots X_{q_n}\,|\,\XX\in\FF_{\bfs\beta}(M^*), X_{q_k}\in \FF_{\bfs\beta}(q_k)\}.$$
\end{def*}

\begin{eg*}
Let us compute $(J_{\bfs\beta})_{q_l}, q_l \in \supp\,\bfs\beta$ in the genus zero case.  
In terms of $w:(M,q_*)\to(\wh\C,\infty),$
$$j_{\bfs\beta}(z) = ib\frac{w''(z)}{w'(z)} -  i\sum_k \frac{\beta_k \, w'(z)}{w(z)-w(q_k)}.$$
We have 
$$(J_{\bfs\beta})_{q_l} = J(q_l) + i \big(b-\frac{\beta_l}2\big)\frac{w''(q_l)}{w'(q_l)} - i\sum_{k\ne l} \frac{\beta_k \, w'(q_l)}{w(q_l)-w(q_k)}.$$
Here, we use 
$$\Big(\frac{w'}{w-w(q)}\Big)*1_{q} = \lim_{z\to q} \Big(\frac{w'(z)}{w(z)-w(q)}-\frac1{z-q}\Big)=\frac12 \frac{w''(q)}{w'(q)}.$$
\end{eg*}

We now prove Theorem~\ref{main: Ward identity}.
\begin{proof}[Proof of Theorem~\ref{main: Ward identity}]
Analysis similar to that in the proof of Proposition~\ref{prop: Ward} shows that 
$$W_v^+ = \sum_{z_j\in S_\XX\sm{\bfs q}} \frac1{2\pi i}\oint_{(z_j)} v\,T_{\bfs\beta} + \sum_{q_k\in\bfs q} \frac1{2\pi i}\oint_{(q_k)} v\,T_{\bfs\beta} $$
in correlations with any tensor product $\XX$ of Fock space field. 
Here $S_\XX$ is the set of all nodes of $\XX.$
For any tensor product $\XX$ in the extend OPE family $\FF_{\bfs\beta}(M),$ we have 
$$\E\sum_{z_j\in S_\XX\sm{\bfs q}} \frac1{2\pi i}\oint_{(z_j)} v\,T_{\bfs\beta} \XX = \E\sum_{z\in S_\XX\sm{\bfs q}} \LL_v^+(z) \XX =\E\, \PP_{\bfs \beta}^{-1} \sum_{z\in S_\XX\sm{\bfs q}} \LL_v^+(z) \PP_{\bfs \beta}\XX.$$
It remains to check that for each $q_k\in\bfs q,$
$$\frac1{2\pi i}\oint_{(q_k)} v\,T_{\bfs\beta} = \PP_{\bfs \beta}^{-1}\LL_v^+(q_k)\PP_{\bfs \beta}$$
in correlations with fields in $\FF_{\bfs\beta}(M).$
If $(vA_{\bfs\beta})*_{-1}X = \LL_v^+X$ and if $(vA_{\bfs\beta})*_{-1}Y_{q_k} = \LL_v^+(q_k)Y_{q_k},$ then 
$$(vA_{\bfs\beta})*_{-1}(X*Y_{q_k}) = \LL_v^+(q_k)(X*Y_{q_k}).$$
Thus in correlations with fields in $\FF_{\bfs\beta}(M),$ it follows that  
$$\frac1{2\pi i}\oint_{(q_k)} vA_{\bfs\beta} = \LL_v^+(q_k).$$ 
Combining this fact with the previous lemma, we find 
$$\frac1{2\pi i}\oint_{(q_k)} v\,T_{\bfs\beta} = \LL_v^+(q_k) + \frac1{2\pi i}\oint_{(q_k)} v\,\E\,T_{\bfs\beta} = \LL_v^+(q_k) + \lambda(\beta_k)v'(q_k) + v(q_k) \PP_{\bfs \beta}^{-1}\pa_{q_k}\PP_{\bfs \beta}$$
in correlations with fields in $\FF_{\bfs\beta}(M).$
Since $\PP_{\bfs \beta}$ is a non-random differential of conformal dimension $(\lambda(\beta_k),\lambda(\beta_k))$ at $q_k,$ it follows from Leibniz's rule that 
$$\PP_{\bfs \beta}^{-1}\LL_v^+(q_k)\PP_{\bfs \beta} = \LL_v^+(q_k) + \lambda(\beta_k)v'(q_k) + v(q_k) \PP_{\bfs \beta}^{-1}\pa_{q_k}\PP_{\bfs \beta}$$
in correlations with fields in $\FF_{\bfs\beta}(M).$
\end{proof}

\subsubsec{Genus zero case} We derive the following version of Ward's equations on the Riemann sphere.

\begin{prop} \label{Ward equation0}
In the $\wh\C$-uniformization, for any tensor product $\XX_{\bfs\beta}$ of fields in the extended OPE family $\FF_{\bfs\beta},$
\begin{equation} \label{eq: Ward equation0}
\E\,T_{\bfs\beta}(\xi) \XX_{\bfs\beta} =  \E\,T_{\bfs\beta}(\xi)\,\E\,\XX_{\bfs\beta} + \E\, \LL_{k_\xi}^+\XX_{\bfs\beta}, \qquad \Big(k_\xi(z) = \frac1{\xi-z}\Big)
\end{equation}
and 
\begin{equation} \label{eq: Ward equation0'}
\E\,\pa T_{\bfs\beta}(\xi) \XX_{\bfs\beta} =  \E\,\pa T_{\bfs\beta}(\xi)\,\E\,\XX_{\bfs\beta} + \E\, \LL_{\ti k_\xi}^+\XX_{\bfs\beta}, \qquad \Big(\ti k_\xi(z) = -\frac1{(\xi-z)^2}\Big).
\end{equation}
\end{prop}

\begin{proof}
By definition of the Ward functional and a simple residue calculus, we find 
$$\E\,W_{k_\xi}^+ \XX_{\bfs\beta} = \E\,T_{\bfs\beta}(\xi) \XX_{\bfs\beta}.$$
It follows from Ward's identity (Theorem~\ref{main: Ward identity}) that 
$$\E\,T_{\bfs\beta}(\xi) \XX_{\bfs\beta}= \E\,\PP_{\bfs\beta}^{-1}\LL_{k_\xi}^+\PP_{\bfs\beta}\XX_{\bfs\beta}.$$
Applying Leibniz's rule for Lie derivatives to tensor products,
$$\E\,T_{\bfs\beta}(\xi) \XX_{\bfs\beta} = (\PP_{\bfs\beta}^{-1}\LL_{k_\xi}^+\PP_{\bfs\beta})\,\E\,\XX_{\bfs\beta} + \E\, \LL_{k_\xi}^+\XX_{\bfs\beta}.$$
Considering the trivial string $\XX\equiv 1,$ we have $\E\,T_{\bfs\beta}(\xi) = \PP_{\bfs\beta}^{-1}\LL_{k_\xi}^+\PP_{\bfs\beta}.$
The second equation~\eqref{eq: Ward equation0'}  can be obtained using the same argument. 
\end{proof}
\begin{eg*}
We now apply this proposition to the previous example. 
By Wick's calculus, 
\begin{align*}
\E\,A_{\bfs\beta}(\xi) J_{\bfs\beta}(z) &=\E\,\big(A(\xi) + ib\pa J(\xi) - j_{\bfs\beta}(\xi)J(\xi)\big) \big(J(z) + j_{\bfs\beta}(z)\big)\\
&= ib\pa_\xi \E\,J(\xi)J(z) - j_{\bfs\beta}(\xi)\E\,J(\xi)J(z)\\
&=\frac{2ib}{(\xi-z)^3} - \frac{i}{(\xi-z)^2}\sum_k\frac{\beta_k}{\xi-q_k}
\end{align*}
in the $\wh \C$-uniformization. 
As $j_{\bfs\beta}(z)$ has no contribution to  $\E\,A_{\bfs\beta}(\xi) J_{\bfs\beta}(z),$ 
\begin{align*}
\E\,A_{\bfs\beta}(\xi) (J_{\bfs\beta})_{q_l} &= \frac{2ib}{(\xi-q_l)^3} - \frac{i}{(\xi-q_l)^2}\sum_k\frac{\beta_k}{\xi-q_k} \\
&= \frac{2ib-i\beta_l}{(\xi-q_l)^3} - \frac{i}{(\xi-q_l)^2}\sum_{k\ne l}\frac{\beta_k}{\xi-q_k}.
\end{align*}
From the previous example, we find $(j_{\bfs\beta})_{q_l}$ is a PS form of order $i(b-\beta_l/2)$ and 
$$(j_{\bfs\beta})_{q_l} = -i\sum_{k\ne j} \frac{\beta_k}{q_l-q_k},\qquad \pa_{q_l}(j_{\bfs\beta})_{q_l} = i\sum_{k\ne j} \frac{\beta_k}{(q_l-q_k)^2}$$
in the $\wh \C$-uniformization. 
Next we compute $\LL_{k_\xi}^+ (j_{\bfs\beta})_{q_l}:$
\begin{align*}
\LL_{k_\xi}^+ (j_{\bfs\beta})_{q_l} &= i\big(b-\frac{\beta_l}2\big) k_\xi''(q_l) + k_\xi(q_l)\pa_{q_l}(j_{\bfs\beta})_{q_l}+k_\xi'(q_l)(j_{\bfs\beta})_{q_l} + \sum_{k\ne l}k_\xi(q_k)\pa_{q_k}(j_{\bfs\beta})_{q_k}\\
&=\frac{2ib-i\beta_l}{(\xi-q_l)^3} +i \frac1{\xi-q_l}\sum_{k\ne j} \frac{\beta_k}{(q_l-q_k)^2}\\
&-i\frac1{(\xi-q_l)^2}\sum_{k\ne j} \frac{\beta_k}{q_l-q_k}-i \sum_{k\ne j} \frac1{\xi-q_k}\frac{\beta_k}{(q_l-q_k)^2} 
\end{align*}
in the $\wh \C$-uniformization. 
We now verify \eqref{eq: Ward equation0} for $\XX_{\bfs\beta} = (J_{\bfs\beta})(q_l):$
\begin{align*}
\frac1{\xi-q_l} \frac1{(q_l-q_k)^2} &-\frac1{(\xi-q_l)^2}\frac1{q_l-q_k} - \frac1{\xi-q_k}\frac1{(q_l-q_k)^2}\\
=&-\frac1{(\xi-q_l)^2}\frac1{q_l-q_k} + \frac1{(q_l-q_k)^2}\Big(\frac1{\xi-q_l}- \frac1{\xi-q_k} \Big)\\
=&-\frac1{(\xi-q_l)^2}\frac1{q_l-q_k} + \frac1{(q_l-q_k)(\xi-q_l)(\xi-q_k)} \\
=&-\frac1{\xi-q_l}\frac1{q_l-q_k} \Big(\frac1{\xi-q_l} - \frac1{\xi-q_k}\Big) = -\frac1{(\xi-q_l)^2 (\xi-q_k)}.
\end{align*}
\end{eg*}

\subsubsec{Genus one case} 
Applying a similar argument as in the genus zero case, we obtain Theorem~\ref{main: KM2}.

\begin{eg*}
We consider $\XX_{\bfs\beta} = \OO_{\bfs\beta}[\bfs\tau]$ with $\supp\,\bfs\tau \cap \supp\,\bfs \beta = \emptyset.$
Recall that Ward's equation~\eqref{eq: WardTorus2} without background charge modifications gives rise to the addition theorem~\eqref{eq: addition4zeta} for Weierstrass $\zeta$-function.  
It follows from Wick's calculus that
$$\frac{\E\,\pa A_{\bfs\beta}(\xi)\OO_{\bfs\beta}[\bfs\tau]}{\E\,\OO_{\bfs\beta}[\bfs\tau]} - \frac{\E\,\pa A(\xi)\OO[\bfs\tau]}{\E\,\OO[\bfs\tau]} = -b \pa_\xi^2 \E\,J(\xi)\Phi[\bfs\tau] -i \pa_\xi(j_{\bfs\beta}(\xi)\E\,J(\xi)\Phi[\bfs\tau]).$$
On the other hand, we compute the corresponding $\C$-linear part of Lie derivative as
\begin{align*}
\frac{\LL_{v_\xi}^+\E\,\OO_{\bfs\beta}[\bfs\tau]}{\E\,\OO_{\bfs\beta}[\bfs\tau]}-\frac{\LL_{v_\xi}^+\E\,\OO[\bfs\tau]}{\E\,\OO[\bfs\tau]} &= -b\sum_j \tau_j \wp'(\xi-z_j) - \sum_k\wp(\xi-q_k) \frac{\pa_{q_k} \E\,\OO_{\bfs\beta}[\bfs\tau]}{\E\,\OO_{\bfs\beta}[\bfs\tau]}\\
&- \sum_j\wp(\xi-z_j) \Big(\frac{\pa_j \E\,\OO_{\bfs\beta}[\bfs\tau]}{\E\,\OO_{\bfs\beta}[\bfs\tau]} - \frac{\pa_j \E\,\OO[\bfs\tau]}{\E\,\OO[\bfs\tau]}\Big).
&\end{align*}
One can check Ward's equation for $\XX_{\bfs\beta} = \OO_{\bfs\beta}[\bfs\tau]$ directly using the KZ equations~\eqref{eq: KZ4torus}, the neutrality condition $(\NC_0)$ for $\bfs\tau,\bfs\beta,$ and the pseudo-addition theorem~\eqref{eq: pseudo-addition4zeta} for $\zeta.$ 
\end{eg*}

\begin{eg*}
In this example, we consider $\XX_{\bfs\beta} = J_{\bfs\beta}(z).$
Recall that 
$$J_{\bfs\beta} = J + j_{\bfs\beta}, \quad j_{\bfs\beta}(z) = -i\sum\beta_k \frac{\theta'(z-q_k)}{\theta(z-q_k)} - \frac{2\pi}{\Im\,\tau} \sum \beta_k\Im\,q_k$$
and 
$$\E\,J(\xi)J(z) = -\wp(\xi-z) - 2\E\,T,\quad \E\,T = \eta_1 - \frac\pi{2\,\Im\,\tau}, \quad \eta_1 = -\frac16\frac{\theta'''(0)}{\theta'(0)}.$$
It follows from Wick's calculus that 
\begin{align*}\pa_\xi \E\,A_{\bfs\beta}(\xi) J_{\bfs\beta}(z) &= ib\pa_\xi^2 \E\,J(\xi)J(z) - \pa_\xi\big(j_{\bfs\beta}(\xi)\E\,J(\xi)J(z) \big)\\
&=-ib\wp''(\xi-z) + i\sum\beta_k \big(\wp(\xi-q_k) + 2\eta_1\big)\big(\wp(\xi-z)+2\E\,T\big)\\
&-i\sum \beta_k \big(\zeta(\xi-q_k) -2\eta_1(\xi-q_k)\big)\wp'(\xi-z)\\
&-\frac{2\pi}{\Im\,\tau}\big(\sum\beta_k \Im\,q_k\big)\wp'(\xi-z).
\end{align*}
We now compute the Lie derivative of PS form $j_{\bfs\beta}$ of order $ib.$ 
By the neutrality condition $\sum \beta_k = 0,$ 
\begin{align*}
\LL_{v_\xi} j_{\bfs\beta}(z) &= ib v_\xi''(z) + v_\xi(z)  j_{\bfs\beta}'(z) + v_\xi'(z)  j_{\bfs\beta}(z) + \sum v_\xi(q_k) \pa_{q_k}j_{\bfs\beta}(z)\\
&=-ib\wp''(\xi-z) -\wp(\xi-z)\big(i\sum\beta_k\wp(z-q_k)\big)\\
&+\wp'(\xi-z)\big(-i\sum \beta_k\zeta'(z-q_k) + 2i\eta_1\beta_k(z-q_k) -\frac{2\pi}{\Im\,\tau}\sum\beta_k \Im\,q_k\big)\\
&+\sum \beta_k \wp(\xi-q_k)\big(i\wp(z-q_k)+2i\eta_1-\frac{i\pi}{\Im\,\tau}\big).
\end{align*}
It follows from Ward's equation that
\begin{align*}
\sum_k \beta_k \Big(\wp(\xi-q_k)\wp(z-q_k) &- \big(\wp(\xi-q_k)+\wp(z-q_k)\big)\wp(\xi-z)\\
&+ \big(\zeta(\xi-q_k)-\zeta(z-q_k)\big)\wp'(\xi-z)\Big)=0.
\end{align*}
Using the addition theorem~\eqref{eq: addition4zeta} for Weierstrass $\zeta$-function and the neutrality condition $\sum_k\beta_k=0$, we rewrite the above identity as 
$$\sum_k \beta_k  F_{\xi,z}(q_k) = 0,$$
where 
\begin{align*}
F_{\xi,z}(q) = \wp(\xi-q)\wp(z-q) &- \big(\wp(\xi-q)+\wp(z-q)\big)\wp(\xi-z)\\
&-\frac12 \frac{\wp'(\xi-q)+\wp'(z-q)}{\wp(\xi-q)-\wp(z-q)}\wp'(\xi-z).
\end{align*}
Equivalently,
$$F_{\xi,z}(q_1) = F_{\xi,z}(q_2).$$
Thus we obtain \eqref{eq: new addition0}.
However, $F_{\xi,z}(q) \ne 0$ in general, see \eqref{eq: new addition}.
\end{eg*}

\section{Eguchi-Ooguri equations and BPZ equations}

In this section, we derive Eguchi-Ooguri equations on a torus without resorting to the path integral formalism. 
Eguchi-Ooguri equations reveal the remarkable relation between the differential operator $\pa_\tau$ with respect to the modular parameter $\tau$ and the integral operator $\oint_a 1\cdot A_{\bfs\beta}$ within correlations of fields in the OPE family, where $1$ is the constant vector field. 
Due to Eguchi-Ooguri equations, the insertion of a stress tensor acts (within correlations of fields in the OPE family) in terms of $\pa_\tau$ and the Lie derivative operator along a certain multi-valued vector field. 

Certain OPE exponentials have the property of degeneracy at level two or higher.  
In the genus zero and genus one cases, we combine the level two/three degeneracy equations with Ward's equations to derive Belavin-Polyakov-Zamolodchikov equations (BPZ equations).
In \cite{KM13} we use Cardy's boundary version of BPZ equations to relate chordal SLE theory to conformal field theory in a simply connected domain with two marked boundary points. See \cite{KM12} for its radial counterpart. 
We also derive Cardy's version of BPZ equations in the genus zero and genus one cases.

\subsection{Eguchi-Ooguri equations} \index{Eguchi-Ooguri equations}

We first prove Theorem~\ref{main: KM} in the case $b=0$ and $\bfs\beta = \bfs0.$

\begin{thm} \label{lem: EOKM0} \index{Eguchi-Ooguri equations}
For any tensor product $\XX$ of fields in the OPE family $\FF,$
\begin{equation} \label{eq: int A}
\frac1{2\pi i}\oint_{[0,1]} \E\,A(\xi)\XX\,\dd \xi = \frac{\pa}{\pa\tau}\,\E\,\XX
\end{equation}
in the $\T_\Lambda$-uniformization.
\end{thm}

\begin{proof}
We first show that \eqref{eq: int A} holds for $\XX = \Phi(z,z_0)\Phi(z',z_0')$ ($\{z,z_0\}\cap\{z',z_0'\}=\emptyset$).
By Wick's calculus and \eqref{eq: EJPhi1}, we have 
$$\E\,A(\xi)\XX = -\E\,J(\xi)\Phi(z,z_0)\,\E\,J(\xi)\Phi(z',z_0')$$ 
and
$$\oint_{[0,1]} \E\,A(\xi)\XX\,\dd \xi = \mathrm{I} + \mathrm{II} + 4\pi^2\,\frac{\Im\,(z-z_0)\,\Im\,(z'-z_0')}{(\Im\,\tau)^2},$$
where $\mathrm{I} = \mathrm{I}_{z,z'} - \mathrm{I}_{z_0,z'} - \mathrm{I}_{z,z_0'} + \mathrm{I}_{z_0,z_0'}, \mathrm{II} = \mathrm{II}_{z,z'} - \mathrm{II}_{z_0,z'} - \mathrm{II}_{z,z_0'} + \mathrm{II}_{z_0,z_0'},$
$$\mathrm{I}_{z,z'} =  \frac{2\pi i}{\Im\,\tau} \Big(\Im\,z \int_0^1 \frac{\theta'(\xi-z')}{\theta(\xi-z')}\,\dd \xi+\Im\,z' \int_0^1 \frac{\theta'(\xi-z)}{\theta(\xi-z)}\,\dd \xi\Big),$$ 
and 
$$\mathrm{II}_{z,z'}=-\int_0^1 \frac{\theta'(\xi-z)}{\theta(\xi-z)}\frac{\theta'(\xi-z')}{\theta(\xi-z')}\,\dd \xi.$$
From the periodicity $\theta(z+1) = -\theta(z),$ it follows that 
$$\mathrm{I}_{z,z'} = -2\pi^2\, \frac{\Im\,(z+z')}{\Im\,\tau}$$
and thus $\mathrm{I} = 0.$
To compute $\mathrm{II},$ we use Frobenius-Stickelberger's pseudo-addition theorem for Weierstrass $\zeta$-function (see \cite[20.41]{WW}),  \index{pseudo-addition theorem for Weierstrass $\zeta$-function}
\begin{equation} \label{eq: pseudo-addition4zeta0}
\big(\zeta(z_1) + \zeta(z_2) + \zeta(z_3)\big)^2 +  \zeta'(z_1) + \zeta'(z_2) + \zeta'(z_3) = 0, \qquad (z_1 + z_2 + z_3 = 0).
\end{equation}
By \eqref{eq: eta theta}, the above formula can be rewritten in terms of Jacobi $\theta$-function as 
\begin{equation} \label{eq: pseudo-addition4zeta}
2 \sum_{j<k} \frac{\theta'(z_j)}{\theta(z_j)}\frac{\theta'(z_k)}{\theta(z_k)} + \sum_j \frac{\theta''(z_j)}{\theta(z_j)} + 6\eta_1 = 0, \qquad (z_1 + z_2 + z_3 = 0).
\end{equation}
Setting $z_1 = z-\xi, z_2 = \xi-z', z_3 = -z+z',$ it follows from the above identity that 
\begin{align*}\mathrm{II}_{z,z'} &=  \frac{\theta'(z-z')}{\theta(z-z')}\int_0^1\Big(-\frac{\theta'(\xi-z)}{\theta(\xi-z)}+\frac{\theta'(\xi-z')}{\theta(\xi-z')}\Big)\,\dd \xi\\
& - \frac12  \frac{\theta''(z-z')}{\theta(z-z')} -\frac12\int_0^1 \frac{\theta''(\xi-z)}{\theta(\xi-z)}\,\dd \xi-\frac12\int_0^1 \frac{\theta''(\xi-z')}{\theta(\xi-z')}\,\dd \xi-6\eta_1.
\end{align*}
Again by the periodicity $\theta(z+1) = -\theta(z),$ the first integral in the right-hand side of the above is zero. On the other hand, the last three terms of the above cancel out with the corresponding terms in $- \mathrm{II}_{z_0,z'} - \mathrm{II}_{z,z_0'} + \mathrm{II}_{z_0,z_0'}.$
Thus we have   
$$\mathrm{II} = -\frac12\Big(\frac{\theta''(z-z')}{\theta(z-z')} - \frac{\theta''(z_0-z')}{\theta(z_0-z')} - \frac{\theta''(z-z_0')}{\theta(z-z_0')} + \frac{\theta''(z_0-z_0')}{\theta(z_0-z_0')}\Big).$$
Differentiating \eqref{eq: EPhi1} with respect to the modular parameter $\tau,$ we obtain
\begin{align} \label{eq: Ooguri4cross-ratio}
2\pi i \,\frac\pa{\pa\tau}\, &\E\ \Phi(z,z_0)\Phi(z',z_0') = 4\pi^2\,\frac{\Im\,(z-z_0)\,\Im\,(z'-z_0')}{(\Im\,\tau)^2}\\
&  -\frac12\Big(\frac{\theta''(z-z')}{\theta(z-z')} - \frac{\theta''(z_0-z')}{\theta(z_0-z')} - \frac{\theta''(z-z_0')}{\theta(z-z_0')} + \frac{\theta''(z_0-z_0')}{\theta(z_0-z_0')}\Big).\nonumber
\end{align}
Here, we use the property of Jacobi-theta function that it satisfies the heat equation, 
\begin{equation} \label{eq: heat4theta} 
2\pi i \, \frac\pa{\pa\tau}\,\theta = \frac12 \theta''.
\end{equation}

Let $\WW$ be the collection of (multi-variable) Fock space fields $\XX$ satisfying \eqref{eq: int A}.
So far we have shown that $\XX = \Phi(z,z_0)\Phi(z',z_0')$ is in $\WW.$
Using Wick's formula, Leibniz's rule for $\pa_\tau$ and an induction argument, it is easy to show that \eqref{eq: int A} holds for $\XX = \prod_{j=1}^{2n}\Phi(z_j,z_j')$ (without overlapping nodes in the tensor products).
Of course, there is nothing to prove for $\XX = \prod_{j=1}^{2n+1}\Phi(z_j,z_j').$
Thus the collection $\WW$ contains tensor products of the Gaussian free field.
It is easy to show that $\WW$ is closed under linear combinations over $\C,$ differentiations, and OPE products $*_n.$ For example, if $X$ is holomorphic and $X,Y$ are in $\WW,$ then
$$X(\zeta)Y(z) = \sum C_n(z)(\zeta-z)^n$$
and
$$\oint_{[0,1]} \E\,A(\xi)X(\zeta)Y(z)\,\dd \xi = 2\pi i \pa_\tau \E\,X(\zeta)Y(z).$$
Term by term integration gives
$$\oint_{[0,1]} \E\,A(\xi)X(\zeta)Y(z)\,\dd \xi = \sum \oint_{[0,1]} \E\,A(\xi) C_n(z)\,\dd \xi (\zeta-z)^n$$
and term by term differentiation gives
$$ \pa_\tau \E\,X(\zeta)Y(z) = \sum \pa_\tau \E\,C_n(z)(\zeta-z)^n.$$
Comparing the coefficients, we conclude the OPE coefficients $C_n$ belong to $\WW.$
Thus $\FF \subseteq \WW$ and we complete the proof of claim.
\end{proof}

\begin{eg*}
We present a conformal field theoretic proof for \eqref{eq: eta1} 
$$\eta_1 = -\frac16\frac{\theta'''(0)}{\theta'(0)},$$
where $\eta_1 = \zeta(1/2).$
For $\XX = J(z)\overline{J(z)},$ it follows from Wick's formula that
$$\E\,A(\xi)\XX = -\E\,J(\xi)J(z) \E\,J(\xi)\overline{J(z)} =-\Big(\wp(\xi-z) - \frac13\frac{\theta'''(0)}{\theta'(0)}- \frac{\pi}{\Im\,\tau}\Big) \frac{\pi}{\Im\,\tau}.$$
Since $\wp = -\zeta'$ and $\zeta(z+1) = \zeta(z) + 2\eta_1,$ we have 
$$\oint_{[0,1]} \E\,A(\xi)\XX\,\dd \xi =\Big(2\eta_1 + \frac13\frac{\theta'''(0)}{\theta'(0)}\Big)\frac{\pi}{\Im\,\tau} +  \Big(\frac{\pi}{\Im\,\tau}\Big)^2.$$
On the other hand, we have
$$ 2\pi i\,\pa_\tau\, \E\,\XX=-2\pi i \frac\pa{\pa\tau} \frac{\pi}{\Im\,\tau}  =  \Big(\frac{\pi}{\Im\,\tau}\Big)^2.$$
Now we use Theorem~\ref{lem: EOKM0} to obtain \eqref{eq: eta1}.
\end{eg*}

\begin{eg*}
Applying Theorem~\ref{lem: EOKM0} to $\XX = \Phi*\Phi(z,z_0),$ we give a field theoretic proof for 
\begin{equation}\label{eq: theta'0}
2\pi i \, \frac {\pa_\tau \theta'(0)} {\theta'(0)} =-3\eta_1. 
\end{equation}
Taking the limits in both sides of \eqref{eq: Ooguri4cross-ratio} as $z'\to z$ and $z_0'\to z_0,$ we have 
$$2\pi i  \frac\pa{\pa\tau} \log\frac{\theta(z-z_0)}{\theta'(0)}=\frac12 \Big( \frac{\theta''(z-z_0)}{\theta(z-z_0)}-\lim_{z\to0} \frac{\theta''(z)}{\theta(z)}\Big).$$ 
Thus it follows from \eqref{eq: heat4theta} and \eqref{eq: eta1} that 
$$2\pi i  \frac\pa{\pa\tau} {\theta'(0)} = \frac12\, {\theta'''(0)}= -3\eta_1 {\theta'(0)}.$$
\end{eg*}

\begin{eg*}
Integrating the differential equation \eqref{eq: wpwp'} $\wp\wp' = \frac1{12}\wp''',$ we have 
\begin{equation} \label{eq: wp2wp''}
\wp^2 = \frac16\wp'' +C
\end{equation}
for some constant $C$ which of course depends on $\tau.$ 
The Laurent expansion \eqref{eq: wp expansion} of $\wp$ around $z=0$
gives $C =  \frac1{12} g_2.$
We now use Theorem~\ref{lem: EOKM0} to derive
\begin{equation}\label{eq: g2}
g_2 = 48\big(\eta_1^2 + \pi i\, \frac{\pa}{\pa\tau}\eta_1\big).
\end{equation}
By Wick's calculus,
$$\E\,A(\xi)T(z) = \frac12(\E\,J(\xi)J(z))^2 = \frac12\big(\wp(\xi-z)+2\eta_1 - \frac\pi{\Im\,\tau}\big)^2.$$
We compute 
$$\int_0^1 \wp(\xi-z)^2\,\dd\xi = \int_0^1\frac16\wp''(\xi-z) + \frac1{12}g_2\,dd\xi = \frac1{12}g_2$$
and 
\begin{equation} \label{eq: int wp over a}
\int_0^1 \wp(\xi-z)\,\dd\xi = \int_0^1 -\zeta'(\xi-z)\,\dd\xi = -\zeta(1-z)+\zeta(-z) = -2\eta_1.
\end{equation}
Combining the above, we have 
$$\int_0^1 \E\,A(\xi)T(z)\,\dd\xi = \frac1{24}g_2-2\eta_1\big(2\eta_1 - \frac\pi{\Im\,\tau}\big)+\frac12 \big(2\eta_1 - \frac\pi{\Im\,\tau}\big)^2.$$
On the other hand, we compute 
$$2\pi i\frac\pa{\pa\tau} \E\,T(z) = 2\pi i \frac\pa{\pa\tau} \big(\eta_1 - \frac\pi{2\, \Im\,\tau}\big) = 2\pi i\frac{\pa}{\pa\tau}\eta_1 +\frac{\pi^2}{2\,(\Im\,\tau)^2}.$$
By Theorem~\ref{lem: EOKM0}, we get
$$ \frac1{24}g_2 -2\eta_1^2 = 2\pi i\frac{\pa}{\pa\tau}\eta_1.$$
\end{eg*}

\begin{eg*}
We consider $\XX = J(z)J(z_0)$ with $z\ne z_0.$
Recall that 
$$\E\, J(z)J(z_0) = -\wp(z-z_0) - 2\eta_1 + \frac\pi{\Im\,\tau}.$$
It follows from \eqref{eq: g2} that
$$2\pi i\pa_\tau \E\, J(z)J(z_0) = -2\pi i \pa_\tau \wp(z-z_0) - \frac1{12}g_2 + 4\eta_1^2 - \frac{\pi^2}{(\Im\,\tau)^2}.$$
By Wick's calculus, 
\begin{align} \label{eq: EAJJg1}
\E\,A(\xi) &J(z)J(z_0)  = - \E\, J(\xi)J(z)\E\, J(\xi)J(z_0) \\
= &-\wp(\xi-z)\wp(\xi-z_0) - \big(2\eta_1 - \frac\pi{\Im\,\tau}\big)\big(\wp(\xi-z)+\wp(\xi-z_0)\big) \nonumber\\
&-4\eta_1^2 + \frac{4\eta_1\pi}{\Im\,\tau}- \frac{\pi^2}{(\Im\,\tau)^2}. \nonumber
 \end{align}
Using \eqref{eq: int A} and \eqref{eq: int wp over a}, we have 
\begin{equation} \label{eq: int wp wp}
\int_0^1 \wp(\xi-z)\wp(\xi-z_0)\,\dd\xi = 2\pi i \, \frac{\pa}{\pa\tau} \, \wp(z-z_0) + \frac1{12}g_2.
\end{equation}
Let us compute
$$\mathcal{I}(z,z_0) :=  \int_0^1 \wp(\xi-z)\wp(\xi-z_0)\,\dd\xi = \pa_z\pa_{z_0}\int_0^1 \zeta(\xi-z)\zeta(\xi-z_0)\,\dd\xi.$$
By the pseudo-addition theorem \eqref{eq: pseudo-addition4zeta0} for the $\zeta$-function, we have 
\begin{equation} \label{eq: pseudo-addition4zeta1}
\big(\zeta(\xi-z)-\zeta(\xi-z_0)+\zeta(z-z_0)\big)^2 = \wp(\xi-z)+\wp(\xi-z_0)+\wp(z-z_0).
\end{equation}
Taking $\pa_z\pa_{z_0}$-derivative and then integrating over $[0,1]$ with respect to $\xi,$
\begin{align*}
\pa_z\pa_{z_0} \zeta(z-z_0)^2 - 2\,\mathcal{I}(z,z_0) - 4\eta_1 \pa_z\pa_{z_0} \big((z-z_0)\zeta(z-z_0)\big) = \pa_z\pa_{z_0} \wp(z-z_0).
\end{align*}
Here, we use
\begin{align} \label{eq: int diff zeta's}
\int_0^1 &\big(\zeta(\xi-z) - \zeta(\xi-z_0)\big)\,\dd\xi \\
&= \int_0^1 \Big(\frac{\theta'(\xi-z)}{\theta(\xi-z)}  - \frac{\theta'(\xi-z_0)}{\theta(\xi-z_0)} - 2\eta_1(z-z_0)\Big)\,\dd\xi= -2\eta_1(z-z_0). \nonumber
\end{align}
By differentiation,
\begin{align*}
\pa_z\pa_{z_0} \zeta(z-z_0)^2 &= -2\wp(z-z_0)^2 + 2\zeta(z-z_0) \wp'(z-z_0),\\
\pa_z\pa_{z_0} \big((z-z_0)\zeta(z-z_0)\big) &= 2\wp(z-z_0) + (z-z_0)\wp'(z-z_0), \\
\pa_z\pa_{z_0} \wp(z-z_0) & = -\wp''(z-z_0).
\end{align*}
Using \eqref{eq: wp2wp''} $\wp^2 = \frac16\wp'' + \frac1{12}g_2,$
$$\mathcal{I}(z) = 2\wp(z)^2 + \zeta(z)\wp'(z) - \frac{g_2}4 - 2\eta_1(2\wp(z)+z\wp'(z)).$$
Combining \eqref{eq: int wp wp} with the above, we arrive to the following formula 
\begin{equation} \label{eq: d tau wp}
2\pi i \, \frac{\pa}{\pa\tau} \, \wp(z) = 2\wp(z)^2 + \zeta(z)\wp'(z) - \frac{g_2}3 - 2\eta_1(2\wp(z)+z\wp'(z)).
\end{equation}
\end{eg*}

We now extend Theorem~\ref{lem: EOKM0} to the case with background charges. 
(Recall that $T_{\bfs\beta} = A_{\bfs\beta} + \E\,T_{\bfs\beta},$ where a stress tensor $A_{\bfs\beta} = A + ib\pa J - j_{\bfs\beta}J,$ $( j_{\bfs\beta}=\E\, J_{\bfs\beta})$ is chosen such that $\E\,A_{\bfs\beta} = 0.$) 

\begin{thm}[Theorem~\ref{main: KM}] \label{lem: EOKM} \index{Eguchi-Ooguri equations}
For any tensor product $\XX_{\bfs\beta}$ of fields in the OPE family $\FF_{\bfs\beta},$
\begin{equation} \label{eq: int A_beta}
\frac1{2\pi i}\oint_{[0,1]} \E\,A_{\bfs\beta}(\xi)\XX_{\bfs\beta}\,\dd \xi = \frac\pa{\pa\tau}\, \E\,\XX_{\bfs\beta}
\end{equation}
in the $\T_\Lambda$-uniformization.
\end{thm}

\begin{proof}
Let us consider the special case $\bfs\beta = \bfs 0, (b\ne0).$
In this case, $A_{\bfs 0} - A = ib\pa J$ is exact in the $\T_\Lambda$-uniformization and thus we have 
\begin{equation} \label{eq: int A0}
\oint_{[0,1]}\E\,A_{\bfs 0}(\xi) \XX_{\bfs 0}\,\dd \xi =  \oint_{[0,1]}\E\,A(\xi) \XX_{\bfs 0}\,\dd \xi = 2\pi i \,\frac{\pa}{\pa \tau}\,\E\,\XX_{\bfs 0}
\end{equation}
for any tensor product $\XX_{\bfs 0}$ of $\Phi_{\bfs 0}(z_j,z_j').$  
Note that $\Phi_{\bfs 0}(z_j,z_j')=\Phi(z_j,z_j')$ in the $\T_\Lambda$-uniformization. 
Using a similar argument in the proof of Theorem~\ref{lem: EOKM0}, the above identity holds for any tensor product $\XX_{\bfs 0}$ of fields in the OPE family $\FF_{\bfs 0}.$ 

Let us consider the case $\bfs\beta\ne\bfs 0$ with $z_j \notin \supp\,\bfs\beta.$
We define the correspondence $\XX_{\bfs 0} \mapsto \wh\XX_{\bfs 0}$ by the formula $\Phi_{\bfs0}[\bfs\tau]\mapsto\Phi_{\bfs\beta}[\bfs\tau]$ and the rules \eqref{eq: hat rules}.
We denote by $\wh\E[\XX_{\bfs 0}]:=\E\,\VV^{\odot}[\bfs\beta]\XX_{\bfs 0}.$
It is easy to see that $\wh A_{\bfs 0} = T_{\bfs \beta}-\E\,T$ in the $\T_\Lambda$-uniformization and $\PP_{\bfs \beta} =\CC_{(b)}[\bfs\beta]=  \E\,\OO_{\bfs 0} [{\bfs \beta}]$ (up to a multiplicative constant).
Indeed, since $\Phi_{\bfs0} = \Phi, J_{\bfs0} =J, \wh J = J + j_{\bfs \beta}$ in the $\T_\Lambda$-uniformization,
$$\wh A_{\bfs 0} =-\frac12 \wh J\odot \wh J + ib \pa \wh J = A + ib\pa J-j_{\bfs \beta}J - \frac12 j_{\bfs \beta}^2 + ib\pa j_{\bfs \beta} = A_{\bfs \beta} + \E\,T_{\bfs \beta} - \E\,T= T_{\bfs \beta}- \E\,T.$$
By definition, $\PP_{\bfs \beta} = \CC_{(b)}[\bfs\beta]=  \E\,\OO_{\bfs \beta_*} [{\bfs \beta - \bfs \beta_*}].$
Here we choose $\bfs \beta_* = \bfs 0.$ 

By Theorem~\ref{main: Insertion} and \eqref{eq: int A0}, we find 
\begin{align*}
\oint_{[0,1]} \E\, T_{\bfs \beta}(\xi)\XX_{\bfs \beta}\,\dd \xi &-\E\,T\, \E\,\XX_{\bfs \beta} = \oint_{[0,1]} \E\, \wh A_{\bfs 0}(\xi)\wh \XX_{\bfs 0}\,\dd \xi \\
&= \frac1{\E\,\OO_{\bfs 0} [{\bfs \beta}]}\oint_{[0,1]} \E\, A_{\bfs 0}(\xi)\XX_{\bfs 0} \OO_{\bfs 0}[\bfs \beta]\,\dd \xi \\
&=  \frac{2\pi i}{\E\,\OO_{\bfs 0} [{\bfs \beta}]} {\pa_\tau}\,\E\,\XX_{\bfs 0} \OO_{\bfs 0}[\bfs \beta]
=2\pi i\big(\pa_\tau \E\,\XX_{\bfs\beta} + \E\,\XX_{\bfs\beta}\, \frac{\pa_\tau\E\,\OO_{\bfs 0}[\bfs \beta]}{\E\,\OO_{\bfs 0}[\bfs \beta]}\big).
\end{align*}
Applying the trivial string $\XX_{\bfs\beta} \equiv 1$ to the above, we have 
\begin{equation}\label{eq: dtau log Pbeta}
\frac1{2\pi i}\oint_{[0,1]} \E\, T_{\bfs \beta}(\xi)\,\dd \xi -\E\,T=  \frac{\pa_\tau\E\,\OO_{\bfs 0}[\bfs \beta]}{\E\,\OO_{\bfs 0}[\bfs \beta]} \Big(= \frac\pa{\pa\tau} \log \PP_{\bfs \beta} \Big)
\end{equation}
and we complete the proof. 
\end{proof}

Combining \eqref{eq: dtau log Pbeta} and \eqref{eq: Zbeta=ZPbeta} $Z_{\bfs\beta} = Z\, \PP_{\bfs \beta}$ with 
$$2\pi i\, \frac{\pa}{\pa\tau} \log Z(\tau) =2\pi i\, \frac{\pa}{\pa\tau} \log \frac1{\sqrt{\Im\,\tau}|\eta(\tau)^2|} =\eta_1 - \frac\pi{2\,\Im\,\tau}=  \E\,T,$$
we obtain the following corollary. 
\begin{cor} \label{dtau log Zbeta}
We have
$$2\pi i\, \frac{\pa}{\pa\tau} \log Z_{\bfs\beta}(\tau)  =  \oint_{[0,1]} \E\,T_{\bfs\beta} (\xi)\,d\xi$$
in the $\T_\Lambda$-uniformization. 
\end{cor}

\begin{eg*}
It follows from \eqref{eq: jbeta4torus}, \eqref{eq: eta theta} and \eqref{eq: g2} that
\begin{align}\label{eq: tau j}
2\pi i \, \frac\pa{\pa\tau}\, j_{\bfs\beta}(z) &= 2\pi \sum_k \beta_k \frac\pa{\pa\tau}\zeta(z-q_k) - (\frac {g_2}{12} -4\eta_1^2)i\sum_k\beta_kq_k \\
&+ \frac{2\pi^2}{(\Im\,\tau)^2}\sum_k\beta_k\Im\,q_k. \nonumber
\end{align}
By Wick's calculus, in the $\T_\Lambda$-uniformization, 
\begin{align*}
\E\,A_{\bfs\beta}(\xi)J_{\bfs\beta}(z) &= ib\pa_\xi\E\,J(\xi)J(z) - j_{\bfs\beta}(\xi)\,\E\,J(\xi)J(z) \\
&=-ib\wp'(\xi-z) \\
&+ \big(\wp(\xi-z) + 2\E\,T\big)\big(-i\sum_k\beta_k \frac{\theta'(\xi-q_k)}{\theta(\xi-q_k)}-\frac{2\pi}{\Im\,\tau}\sum_k\beta_k\Im\,q_k\big),
\end{align*}
where $\E\,T = \eta_1  -\pi/(2\,\Im\,\tau).$
Since $\wp$ is elliptic, the first term in the right-hand side does not contribute to the integral in \eqref{eq: int A_beta}.
Due to the quasi-periodicity \eqref{eq: zeta periodicities} of $\zeta$-function,
$$\int_0^1 \big(\wp(\xi-z) + 2\,\E\,T\big)\big(-\frac{2\pi}{\Im\,\tau}\sum_k\beta_k\Im\,q_k\big)\,\dd\xi=\frac{2\pi^2}{(\Im\,\tau)^2}\sum_k\beta_k\Im\,q_k.$$
This coincides with the last term of the right-hand side in \eqref{eq: tau j}.
Using the neutrality condition $\sum_k\beta_k=0$ and the fact that the integral of $\theta'(\xi-q_k)/\theta(\xi-q_k)$ with respect to the $\xi$-variable over the cycle $[0,1]$ is constant, 
$$-2i\,\E\,T \int_0^1\sum_k\beta_k \frac{\theta'(\xi-q_k)}{\theta(\xi-q_k)}\,\dd\xi = 0.$$
It follows from \eqref{eq: eta theta} that 
\begin{align*}
-i\sum_k\beta_k\int_0^1&\wp(\xi-z) \frac{\theta'(\xi-q_k)}{\theta(\xi-q_k)}\,\dd\xi = -i\sum_k\beta_k\int_0^1\wp(\xi-z) \zeta(\xi-q_k)\,\dd\xi\\
&-2\eta_1 i \sum_k \beta_kq_k \int_0^1\wp(\xi-z) \,\dd\xi +2\eta_1 i \sum_k\beta_k\int_0^1\xi\,\wp(\xi-z) \,\dd\xi. 
\end{align*}
Note that the last term of the above vanishes due to the neutrality condition $\sum_k\beta_k=0.$
By \eqref{eq: int wp over a}, we have 
\begin{align} \label{eq: int wp theta wp zeta}
-i\sum_k\beta_k\int_0^1\wp(\xi-z) \frac{\theta'(\xi-q_k)}{\theta(\xi-q_k)}\,\dd\xi &= -i\sum_k\beta_k\int_0^1\wp(\xi-z) \zeta(\xi-q_k)\,\dd\xi\\
&+4\eta_1^2 i \sum_k \beta_kq_k. \nonumber
\end{align}
Let us compute 
$$\mathcal{J}(z,z_0) :=  \int_0^1 \wp(\xi-z)\zeta(\xi-z_0)\,\dd\xi = \pa_z\int_0^1 \zeta(\xi-z)\zeta(\xi-z_0)\,\dd\xi.$$
Taking $\pa_z$-derivative of \eqref{eq: pseudo-addition4zeta1}, letting $z_0 = q_k,$ and then integrating over $[0,1]$ with respect to $\xi,$
\begin{align*}
\mathcal{J}(z,q_k) &= \mathcal{J}(z,z) -\zeta(z-q_k)\wp(z-q_k) +2\eta_1(z-q_k)\wp(z-q_k) \\
&-2\eta_1\zeta(z-q_k) - \frac12\wp'(z-q_k).
\end{align*}
Here we use \eqref{eq: int diff zeta's} and \eqref{eq: int wp over a}.
Note that the term $\mathcal{J}(z,z)$ is independent of $q_k$, hence does not contribute to the sum in \eqref{eq: int wp theta wp zeta}.
Combining all of the above with the neutrality condition $\sum_k\beta_k=0,$ we have 
$$2\pi i\sum_k \beta_k \frac\pa{\pa\tau}\zeta(z-q_k) =\sum\beta_k F(z-q_k),$$
where
$$F(z) = - \frac12\wp'(z)-\zeta(z)\wp(z)  +\frac{g_2}{12} z+2\eta_1\big(z\wp(z) -\zeta(z)\big).$$
It is known that 
\begin{equation}\label{eq: tau zeta}
2\pi i  \,\frac\pa{\pa\tau}\zeta(z) = F(z).
\end{equation}
Differentiating the above with respect to $z$ leads  to \eqref{eq: d tau wp}.
\end{eg*}

\subsection{Eguchi and Ooguri's version of Ward's equations}

We now prove Theorem~\ref{main: KM2'}  in the case $b=0$ and $\bfs\beta=\bfs0.$

\begin{thm}\label{EOKM0} \index{Ward's equations}
For  any tensor product $\XX$ of fields in the OPE family $\FF,$ we have $$ \E\,A(\xi)\XX =  \LL_{\ti v_\xi}^+\E\,\XX+2\pi i \, \frac\pa{\pa\tau} \,\E\,\XX, \quad(\ti v_\xi(z) = \zeta(\xi-z) + 2\eta_1 z)$$
in the $\T_\Lambda$-uniformization. 
\end{thm}

\begin{proof}
Let us consider two vector fields $\ti v^1(z) = z, \ti v_\xi^2 (z) = \zeta(\xi-z).$ 
We remark that these two vector fields have jump discontinuities across the cycles $a = [0,1], b = [0,\tau].$  
Consider a parallelogram $D:=\{z\in\C\,|\, z = x + y\tau, x,y\in[0,1]\}$ and its boundary $\gamma_1+\gamma_2+\gamma_3+\gamma_4$ as a positively oriented curve,  
$$\gamma_1 = [0,1], \quad \gamma_2 = [1,1+\tau], \quad \gamma_3 = [1+\tau,\tau], \quad \gamma_4 = [\tau,0].$$
It follows from Green's formula that
\begin{align*}
-\frac1{\pi} \int_D\bp  \ti v^1(z) \E\,A(z)\XX &= \sum \frac1{2\pi i} \oint_{(z_j)} \ti v^1(z) \E\,A(z)\XX\,\dd z\\
&+\frac{\tau}{2\pi i}\int_{\gamma_1} \E\,A(z)\XX\,\dd z  - \frac1{2\pi i}\int_{\gamma_2} \E\,A(z)\XX\,\dd z.
\end{align*}
We use the periodicity~\eqref{eq: zeta periodicities} of $\zeta$-function and Green's formula to obtain  
\begin{align*}
-\frac1{\pi}\int_D \bp \ti v_\xi^2(z) \E\,A(z)\XX &= \sum \frac1{2\pi i} \oint_{(z_j)} \ti v_\xi^2(z) \E\,A(z)\XX\,\dd z\\
&-\frac{\eta_2}{\pi i}\int_{\gamma_1} \E\,A(z)\XX\,\dd z + \frac{\eta_1}{\pi i}\int_{\gamma_2} \E\,A(z)\XX\,\dd z.\end{align*}
Taking a suitable linear combination of the above two equations, we have 
\begin{align*}
-\frac1{\pi}\int_D \bp \ti v_\xi(z) \E\,A(z)\XX &=  \sum \frac1{2\pi i} \oint_{(z_j)} \ti v_\xi(z) \E\,A(z)\XX\,\dd z\\
&+ \frac{\eta_1\tau-\eta_2}{\pi i}\int_{\gamma_1} \E\,A(z)\XX\,\dd z.
\end{align*}
Note that $\ti v_\xi = \ti v_\xi^2 + 2\eta_1 \ti v^1.$  
Since $\bp \ti v_\xi = -\pi \delta_\xi,$ the left-hand side simplifies to $\E\,A(\xi)\XX.$ 
The first term in the right-hand side is identified with $ \LL_{\ti v_\xi}^+\E\,\XX.$ 
Due to Legendre, $\eta_1\tau-\eta_2=\pi i$ (e.g. see \cite[p.50]{Chandrasekharan}).
Theorem~\ref{lem: EOKM0} says that the second term in the right-hand side is $2\pi i\pa_\tau \,\E\,\XX.$
\end{proof}

The following corollary first appeared in \cite[Eq.~(28)]{EO87}. 
\begin{cor}[Eguchi-Ooguri]\label{Eguchi-Ooguri}
If $X_j$'s are $(\lambda_j,\lambda_{*j})$-differentials in the OPE family $\FF,$ then in the $\T_\Lambda$-uniformization 
\begin{align} \label{eq: EO87Eq28}
\E\,T(\xi)\XX &=  \E\,T(\xi) \E\,\XX+2\pi i \,\frac\pa{\pa\tau}\,\E\,\XX\\
& +\sum_j  \big((\zeta(\xi-z_j) +2\eta_1z_j)\pa_j  + \lambda_j( \wp(\xi-z_j) + 2\eta_1)\big)\E\,\XX,\nonumber
\end{align}
where $\XX = X_1(z_1)\cdots X_n(z_n).$
\end{cor}

\begin{eg*}
We consider $\XX = J(z)J(z_0)$ ($z\ne z_0$). 
Since $J$ is a 1-differential, 
\begin{align*}&\LL_{\ti v_\xi}\E\, J(z)J(z_0) = \pa_z\big(\ti v_\xi(z) \E\, J(z)J(z_0) \big) + \pa_{z_0}\big(\ti v_\xi(z_0) \E\, J(z)J(z_0) \big)\\
&= - \big(\wp(\xi-z)+\wp(\xi-z_0)\big)\wp(z-z_0) - \big(\zeta(\xi-z)-\zeta(\xi-z_0)\big)\wp'(z-z_0)\\
&-(2\eta_1 - \frac\pi{\Im\,\tau})\big(\wp(\xi-z)+\wp(\xi-z_0)\big) -4\eta_1\wp(z-z_0)-8\eta_1^2+\frac{4\eta_1\pi}{\Im\,\pi}\\
&-2\eta_1(z-z_0)\wp'(z-z_0).
\end{align*}
It follows from \eqref{eq: d tau wp} and \eqref{eq: g2} that
\begin{align*}
2\pi i \pa_\tau \E\, J(z)J(z_0) &= -2\wp(z-z_0)^2 - \zeta(z-z_0)\wp'(z-z_0)-\frac{g_2}4 \\
&+4\eta_1\wp(z-z_0)+2\eta_1(z-z_0)\wp'(z-z_0) + 4\eta_1^2 - \frac{\pi^2}{(\Im\,\tau)^2}.
\end{align*}
By the previous corollary and \eqref{eq: EAJJg1}, we have
\begin{align*}
\wp(\xi-z)\wp(\xi-z_0) &- \big(\wp(\xi-z)+\wp(\xi-z_0)\big)\wp(z-z_0)\\
& - \big(\zeta(\xi-z)-\zeta(\xi-z_0)\big)\wp'(z-z_0)\\
&= 2\wp(z-z_0)^2 +\zeta(z-z_0)\wp'(z-z_0)-\frac{g_2}4. 
\end{align*}
Using \eqref{eq: wp2wp''} $\wp^2 = \frac16\wp'' + \frac1{12}g_2,$ we rewrite the previous identity as follows:
\begin{align*}
\big(\wp(\xi-z)- \wp(z-z_0)\big)&\big(\wp(\xi-z_0) - \wp(z-z_0)\big) - \big(\zeta(\xi-z)-\zeta(\xi-z_0)\big)\wp'(z-z_0)\\
&= \frac12\wp''(z-z_0) +\zeta(z-z_0)\wp'(z-z_0). 
\end{align*}
Letting $z_1 = z -\xi, z_2 = \xi-z_0, z_3 = -z + z_0,$ it follows from the addition theorem \eqref{eq: addition4zeta} for Weierstrass $\zeta$-function that
$$\big(\wp(z_1)-\wp(z_3)\big) \big(\wp(z_2)-\wp(z_3)\big) + \frac12\frac{\wp'(z_1)-\wp'(z_2)}{\wp(z_1)-\wp(z_2)}\,\wp'(z_3) = \frac12\wp''(z_3)$$
if $z_1+z_2+z_3 = 0.$  
Thus we have \eqref{eq: new addition}.
\end{eg*}

\begin{comment}
\begin{proof}
Since $\wp = -\zeta',$ it follows from \eqref{eq: zeta periodicities} that 
\begin{align*}
\int_0^1 v_{\xi,\xi_0}'(z) \,\dd \xi_0 
&= \int_0^1 \big(\wp(\xi-z) - \wp(\xi_0-z)\big)\,\dd \xi_0 \\
&= \wp(\xi-z) +\zeta(1-z)-\zeta(-z) =  \wp(\xi-z) + 2\eta_1.
\end{align*}
On the other hand, by \eqref{eq: sigma periodicities}, \eqref{eq: sigma theta}, and the relation $\zeta = \sigma'/\sigma,$ we have  
\begin{align*}
\int_0^1 v_{\xi,\xi_0}(z) \,\dd \xi_0 
&= \int_0^1 \big(\zeta(\xi-z) - \zeta(\xi_0-z)-2\eta_1(\xi-\xi_0)\big)\,\dd \xi_0 \\
&= \zeta(\xi-z) -\log \frac{\sigma(1-z)}{\sigma(-z)} -2\eta_1\xi +\eta_1\\
&=  \zeta(\xi-z) +2\eta_1z -2\eta_1\xi-i\pi.
\end{align*}
The term $-2\eta_1\xi-i\pi$ has no contribution to Ward's equation since $\sum \pa_j \E\,\XX = 0.$ 
(Considering a constant vector field, it follows from Corollary~\ref{cor: Ward} immediately.) 
\end{proof}
\end{comment}

We use Theorem~\ref{lem: EOKM} to extend Theorem~\ref{EOKM0} to Theorem~\ref{main: KM2'} in the case $\bfs\beta\ne\bfs0.$
 
\begin{thm}\label{EOKM} \index{Ward's equations}
For  any tensor product $\XX_{\bfs\beta}$ of fields in the OPE family $\FF_{\bfs\beta},$ we have 
$$ \E\,T_{\bfs\beta}(\xi)\XX_{\bfs\beta} =   \E\,T_{\bfs\beta}(\xi)\E\XX_{\bfs\beta}+ \LL_{\ti v_\xi}^+\E\,\XX_{\bfs\beta}+2\pi i \, \frac\pa{\pa\tau} \,\E\,\XX_{\bfs\beta}, \quad(\ti v_\xi(z) = \zeta(\xi-z) + 2\eta_1 z)$$
in the $\T_\Lambda$-uniformization. 
\end{thm}
\begin{proof}
Using the same argument in the proof of Theorem~\ref{EOKM0}, we have 
\begin{align*}
-\frac1{\pi}\int_D \bp \ti v_\xi(z) \E\,A_{\bfs\beta}(z)\XX_{\bfs\beta} &=  \sum_{z_j\in S_{\XX_{\bfs\beta}}\setminus\bfs q} \frac1{2\pi i} 
\oint_{(z_j)} \ti v_\xi(z) \E\,A_{\bfs\beta}(z)\XX_{\bfs\beta}\,\dd z\\
& +\sum_{q_k\in \bfs q} \frac1{2\pi i} \oint_{(z_j)} \ti v_\xi(z) \E\,A_{\bfs\beta}(z)\XX_{\bfs\beta}\,\dd z + \int_{\gamma_1} \E\,A_{\bfs\beta}(z)\XX_{\bfs\beta}\,\dd z.
\end{align*}
Since $\bp \ti v_\xi = -\pi \delta_\xi,$ the left-hand side simplifies to $\E\,A_{\bfs\beta}(\xi)\XX_{\bfs\beta}=\E\,T_{\bfs\beta}(\xi)\XX_{\bfs\beta}-\E\,T_{\bfs\beta}(\xi)\E\,\XX_{\bfs\beta}.$ 
The first two terms in the right-hand side are identified with $ \LL_{\ti v_\xi}^+\E\,\XX_{\bfs\beta}.$
By Theorem~\ref{lem: EOKM}, the last term in the right-hand side is $2\pi i\pa_\tau \,\E\,\XX_{\bfs\beta}.$
\end{proof}

\begin{cor}
If $X_j$'s are $(\lambda_j,\lambda_{*j})$-differentials in the extended OPE family $\FF_{\bfs\beta},$ then in the $\T_\Lambda$-uniformization 
\begin{align} \label{eq: EO87Eq28beta}
\E\,T_{\bfs\beta}(\xi)\XX_{\bfs\beta} &=  \E\,T_{\bfs\beta}(\xi) \E\,\XX_{\bfs\beta}+2\pi i \,\frac\pa{\pa\tau}\,\E\,\XX_{\bfs\beta}\\
& + \sum_j  \big((\zeta(\xi-z_j) +2\eta_1z_j)\pa_j  + \lambda_j( \wp(\xi-z_j) + 2\eta_1)\big)\E\,\XX_{\bfs\beta}
\nonumber\\
&+ \sum_k  (\zeta(\xi-q_k) +2\eta_1q_k)\pa_{q_k} \E\,\XX_{\bfs\beta},\nonumber
\end{align}
where $\XX_{\bfs\beta} = X_1(z_1)\cdots X_n(z_n).$
\end{cor}

\begin{eg*}
Consider the PS form $\XX_{\bfs\beta} = J_{\bfs\beta}(z)$ of order $ib.$ 
Then Ward's equation for $J_{\bfs\beta}(z)$ reads as 
\begin{align*} 
\E\,A_{\bfs\beta}(\xi)J_{\bfs\beta}(z) &=  2\pi i \,\frac\pa{\pa\tau}\,j_{\bfs\beta}(z)+ \sum_k (\zeta(\xi-q_k) +2\eta_1q_k)\pa_{q_k}j_{\bfs\beta}(z) \\
&-ib\,\wp'(\xi-z) + \big((\zeta(\xi-z) +2\eta_1z)\pa_z  + (\wp(\xi-z) + 2\eta_1)\big)j_{\bfs\beta}(z).
\nonumber
\end{align*}
It follows that 
$$\sum_k\beta_k \Big(\frac{1}2\wp'(z-q_k)+\big(\zeta(z-q_k)+\zeta(\xi-z)-\zeta(\xi-q_k)\big)\big(\wp(z-q_k) - \wp(\xi-z)\big)\\
\Big)=0.$$
One can check it directly using the addition theorem \eqref{eq: addition4zeta} for Weierstrass $\zeta$-function and the neutrality condition $\sum_k\beta_k = 0.$

\end{eg*}

\subsection{Level two degeneracy equations and BPZ equations} \label{ss: Cardy2}
Let $\{J_n\}$ and $\{L_n\}$ denote the modes of the current field $J_{\bfs\beta}$ and the Virasoro field $T_{\bfs\beta}$ in $\FF_{\bfs\beta}$ theory, respectively: 
$$J_n(z):=\frac1{2\pi i}\oint_{(z)}(\zeta-z)^{n} J_{\bfs\beta}(\zeta)~\dd\zeta, \qquad L_n(z):=\frac1{2\pi i}\oint_{(z)}(\zeta-z)^{n+1} T_{\bfs\beta}(\zeta)~\dd\zeta.$$
Recall that a field $X$ is called \emph{Virasoro primary} if $X$ is a differential and if $X$ is in the family $\FF_{\bfs\beta}.$ \index{Virasoro primary} 
It is well known that $X$ in $\FF_{\bfs\beta}$ is Virasoro primary if and only if $L_nX = L_n \bar X=0$ for all $n\ge1$ and 
\begin{equation} \label{eq: T primary}
\begin{cases}
L_{-1}X=\pa X, \\
L_{-1}\bar X=\pa \bar X, 
\end{cases}
\qquad
\begin{cases}
L_0X=\lambda X, \\
L_0\bar X=\bar\lambda_* \bar X,
\end{cases}
\end{equation}
for some numbers $\lambda$ and $\lambda_*.$ 
(These numbers are called conformal dimensions of $X$.)
See \cite[Proposition~7.5]{KM13} for this.  
A Virasoro primary field $X$ is called \emph{current primary} \index{current primary} if 
$J_nX = J_n\bar X = 0$
for all $n\ge 1$ and
\begin{equation} \label{eq: J0}
J_0X = -iqX, \quad J_0\bar X = i\bar q_*\bar X
\end{equation}
for some numbers $q$ and $q_*$ 
(These numbers are called charges of $X$.)
It is well known that current primary fields with specific charges satisfy the level two degeneracy equations.
\index{level two degeneracy equations}

\begin{prop} \label{level2degeneracy} 
For a current primary field $\OO$ in $\FF_{\bfs\beta}$ with charges $q,q_*$ at $z\in S_\OO,$ we have 
\begin{equation} \label{eq: level2degeneracy} 
\big(L_{-2}(z)+\eta L_{-1}^2(z)\big)\OO =0 
\end{equation}
if $2q(b+q) = 1$ and if $\eta = -1/(2q^2).$
\end{prop}
See \cite[Proposition~11.2]{KM13}) for its proof.

\subsubsec{Genus zero case} We now derive BPZ equations (Belavin-Polyakov-Zamolodchikov equations) on the Riemann sphere. We remark that in the upper half-plane $\H$ with $\bfs\beta = 2b\delta_\infty,$ a different type of BPZ equations plays an important role in the context of SLE martingales, see \cite{KM13}.

Given $b,$ let $a$ be one of the solutions to the quadratic equation $2x(x+b) = 1$ for $x$ 
and let $\OO_{\bfs\beta}(z)\equiv \OO_{\bfs\beta}^{(a,\bfs\tau)}(z) := \OO_{\bfs\beta}[a\cdot z + \bfs \tau]$ $(\bfs\tau = \sum \tau_j \cdot z_j, j\le0)$ with the neutrality condition $(\NC_0): a + \sum \tau_j = 0.$

\begin{thm} \label{BPZ0} \index{BPZ equations}
If $z\notin\supp\,\bfs\beta \cup \supp\,\bfs\tau,$ then  for any tensor product $\XX_{\bfs\beta}$ of fields $X_j$ in $\FF_{\bfs\beta},$ we have 
in the $\wh\C$-uniformization, 
$$\frac1{2a^2}\pa_z^2\E\, \OO_{\bfs\beta}(z) \XX_{\bfs\beta} = \E\,T_{\bfs\beta}(z) \E\,\OO_{\bfs\beta}(z) \XX_{\bfs\beta} +\E\, \check\LL_{k_z}^+\OO_{\bfs\beta}(z)\XX_{\bfs\beta},$$
where the vector field $k_z$ is given by $k_z(\zeta) = 1/(z-\zeta)$ in the identity chart of $\C$ and the Lie derivative operator $\check\LL_{k_z}^+$ does not apply to the $z$-variable. 
\end{thm}

\begin{proof}
It follows from Ward's equations (Proposition~\ref{Ward equation0}) that 
\begin{align*}
\E\,T_{\bfs\beta}(\xi) \OO_{\bfs\beta}(z) \XX_{\bfs\beta}  &= \E\,T_{\bfs\beta}(\xi) \E\,\OO_{\bfs\beta}(z) \XX_{\bfs\beta} \\
&+ \E\,\LL_{k_\xi}^+(z)\OO_{\bfs\beta}(z)\XX_{\bfs\beta}+ \E\,\check\LL_{k_\xi}^+\OO_{\bfs\beta}(z) \XX_{\bfs\beta}.
\end{align*}
Note that $\LL_{k_\xi}^+(z)\OO_{\bfs\beta}(z)$ is the singular part of the operator product expansion of $T_{\bfs\beta}(\xi)$ and $\OO_{\bfs\beta}(z)$ as $\xi\to z.$ 
Subtracting $\E\,\LL_{k_\xi}^+(z)\OO_{\bfs\beta}(z)\XX_{\bfs\beta}$ from both sides and then taking the limits as $\xi\to z,$ we obtain 
$$\E\,T_{\bfs\beta}*\OO_{\bfs\beta}(z) \XX_{\bfs\beta} = \E\,T_{\bfs\beta}(z) \E\,\OO_{\bfs\beta}(z) \XX_{\bfs\beta} +\E\, \check\LL_{k_z}^+\OO_{\bfs\beta}(z)\XX_{\bfs\beta},$$
where the OPE product is taken with respect to $z.$
Theorem now follows from the level two degeneracy equations~\eqref{eq: level2degeneracy} for $\OO_{\bfs\beta}.$
\end{proof}

Let $Z_{\bfs\beta} = \E\,\OO_{\bfs\beta}.$
The next corollary is immediate from Theorem~\ref{BPZ0} with $\XX \equiv 1.$ 
\begin{cor} \index{null vector equations}
The functions $Z_{\bfs\beta}$ satisfy the null vector equations
$$\frac1{2a^2} \pa_z^2 Z_{\bfs\beta}  = \E\,T_{\bfs\beta}(z) Z_{\bfs\beta} + \check\LL_{k_z}^+Z_{\bfs\beta}.$$
\end{cor}
In the previous corollary, $\check\LL_{k_z}^+Z_{\bfs\beta}$ reads as 
$$\check\LL_{k_z}^+Z_{\bfs\beta} = \sum_j \Big(\frac{\lambda_j}{(z-z_j)^2} + \frac{\pa_j}{z-z_j} \Big)Z_{\bfs\beta}  +\sum_k \frac{\pa_{q_k}}{z-q_k}Z_{\bfs\beta}.$$

We denote
$$\wh\E_{\bfs\beta} \XX_{\bfs\beta}  := \frac{\E\,\OO_{\bfs\beta}(z) \XX_{\bfs\beta} }{\E\,\OO_{\bfs\beta}(z)}.$$
Then we have 
\begin{equation} \label{eq: hat'}
\frac{\pa_z\E\,\OO_{\bfs\beta}(z) \XX_{\bfs\beta} } {\E\,\OO_{\bfs\beta}(z)} = \pa_z \wh\E_{\bfs\beta} \XX_{\bfs\beta}   +  \frac{\pa_z Z_{\bfs\beta}}{Z_{\bfs\beta}} \wh\E_{\bfs\beta} \XX_{\bfs\beta}  
\end{equation}
and similar statements hold for $z_j$-derivative and $q_k$-derivative. 
Also we have 
\begin{equation} \label{eq: hat''}
\frac{\pa_z^2\E\,\OO_{\bfs\beta}(z) \XX_{\bfs\beta} } {\E\,\OO_{\bfs\beta}(z)} = \pa_z^2 \wh\E_{\bfs\beta} \XX_{\bfs\beta}   +2 \frac{\pa_z Z_{\bfs\beta}}{Z_{\bfs\beta}}\pa_z \wh\E_{\bfs\beta} \XX_{\bfs\beta} + \frac{\pa_z^2 Z_{\bfs\beta}}{Z_{\bfs\beta}} \wh\E_{\bfs\beta} \XX_{\bfs\beta} .
\end{equation}
Using the above relations to rewrite the BPZ equations in terms of $\wh\E_{\bfs\beta} \XX$ and its derivatives, we arrive to the BPZ-Cardy equations, see the next corollary. 
We remark that the coefficient of $\wh\E_{\bfs\beta} \XX$ vanishes in consequence of the null vector equations for $Z_{\bfs\beta.}$

\begin{cor} \index{BPZ-Cardy equations}
For any tensor product $\XX_{\bfs\beta}$ of fields in the extended OPE family $\FF_{\bfs\beta},$ 
$$\frac1{2a^2}\Big(\pa_z^2 \wh\E_{\bfs\beta} \XX_{\bfs\beta}  + 2\frac{\pa_z Z_{\bfs\beta}}{Z_{\bfs\beta}} \pa_z \wh\E_{\bfs\beta} \XX_{\bfs\beta} \Big) = \check\LL_{k_z}^+\wh\E_{\bfs\beta} \XX_{\bfs\beta}$$
in the $\wh\C$-uniformization.
\end{cor}

If $\XX_{\bfs\beta} = X_1(z_1)\cdots X_n(z_n),$ then $\check\LL_{k_z}^+\wh\E_{\bfs\beta} \XX_{\bfs\beta}$ in the previous corollary reads as 
$$\check\LL_{k_z}^+\wh\E_{\bfs\beta} \XX_{\bfs\beta} = \sum_{j>0} \LL_{k_z}^+(z_j) \wh\E_{\bfs\beta} \XX_{\bfs\beta} + \Big(\sum_{j\le 0} \frac{\pa_j}{z-z_j}  +\sum_k \frac{\pa_{q_k}}{z-q_k} \Big)\wh\E_{\bfs\beta} \XX_{\bfs\beta}.$$

\subsubsec{Genus one case} 
We now prove Theorem~\ref{BPZ1}, the BPZ equations in $\T_\Lambda$ associated with the level two degeneracy equations.  
Let $\ti v_\xi(z) = \zeta(\xi-z) + 2\eta_1 z.$
Given $b,$ let $a$ be one of the solutions to the quadratic equation $2x(x+b) = 1$ for $x$ 
and let $\OO_{\bfs\beta}(z)\equiv \OO_{\bfs\beta}^{(a,\bfs\tau)}(z) := \OO_{\bfs\beta}[a\cdot z + \bfs \tau]$  with the neutrality condition $(\NC_0).$

\begin{proof}[Proof of Theorem~\ref{BPZ1}] \index{BPZ equations}
It follows from Ward's equations~\eqref{eq: EOKM_main} that 
$$\E\,T_{\bfs\beta}(\xi)\OO_{\bfs\beta}(z)\XX_{\bfs\beta} = \E\,T_{\bfs\beta}(\xi)\,\E\,\OO_{\bfs\beta}(z)\XX_{\bfs\beta}  +2\pi i\,\pa_\tau \E\,\OO_{\bfs\beta}(z)\XX_{\bfs\beta} + \LL_{\ti v_\xi}^+\E\,\OO_{\bfs\beta}(z)\XX_{\bfs\beta}.$$
Subtracting $\Sing_{\xi\to z}\E\, T_{\bfs\beta}(\xi) \OO_{\bfs\beta}(z)\XX_{\bfs\beta}$ from both sides and then taking the limit as $\xi\to z,$ we find
\begin{align*}
\E\,L_{-2}(z)\OO_{\bfs\beta}(z)\XX_{\bfs\beta} &= \E\,T_{\bfs\beta}(z)\,\E\,\OO_{\bfs\beta}(z)\XX_{\bfs\beta}  +2\pi i\,\pa_\tau \E\,\OO_{\bfs\beta}(z)\XX_{\bfs\beta} +\check\LL_{\ti v_z}^+\E\,\OO_{\bfs\beta}(z)\XX_{\bfs\beta}\\
&+ \lim_{\xi\to z} \big(\zeta(\xi-z)-\frac1{\xi-z} +2\eta_1z\big)\pa_z\E\,\OO_{\bfs\beta}(z)\XX_{\bfs\beta}\\
&+ \lim_{\xi\to z} \lambda\big( \wp(\xi-z) - \frac1{(\xi-z)^2}+ 2\eta_1\big)\E\,\OO_{\bfs\beta}(z)\XX_{\bfs\beta}.
\end{align*}
Theorem is immediate from the level two degeneracy equation for $\OO_{\bfs\beta}.$
\end{proof}

Let $Z_{\bfs\beta} =\E\,\OO_{\bfs\beta}.$
\begin{cor} \label{null1} \index{null vector equations}
The functions $Z_{\bfs\beta}$ satisfy the null vector equations
$$
\frac1{2a^2}\pa_z^2 Z_{\bfs\beta} =  \big( \E\,T_{\bfs\beta}(z) +2\pi i\,\pa_\tau + 2\eta_1z\pa_z +2\lambda\eta_1+\check\LL_{\ti v_z}^+\big)  Z_{\bfs\beta}$$
in the $\T_\Lambda$-uniformization.
\end{cor}

In the previous corollary, $\check\LL_{\ti v_z}^+ Z_{\bfs\beta}$ reads as 
\begin{align*}
\check\LL_{\ti v_z}^+ Z_{\bfs\beta} &= \sum_j \big(\zeta(z-z_j) + 2\eta_1 z_j\big)\pa_j Z_{\bfs\beta} + \lambda_j\big(\wp(z-z_j) + 2\eta_1\big)Z_{\bfs\beta}\\
&+ \sum_k \big(\zeta(z-q_k) + 2\eta_1 q_k\big)\pa_{q_k} Z_{\bfs\beta}.
\end{align*}

\begin{eg*}
Let us consider the case $\bfs\beta = \bfs 0$ and $Z(z,z_0) = \E\,\OO[a\cdot z - a\cdot z_0].$ We verify the null vector equation for $Z:$
$$Z(z,z_0) = |\theta'(0)|^{2a^2}|\theta(z-z_0)|^{-2a^2}\exp\big(2\pi a^2\frac{(\Im(z-z_0))^2}{\Im\,\tau}\big).$$
Let $w = z-z_0.$
Differentiating $Z$ with respect to $z,$ we have 
\begin{equation}\label{eg: pa Z}
\frac{\pa_zZ}{Z} = -a^2\big(\zeta(w)-2\eta_1w+2\pi i \,\frac{\Im\,w}{\Im\,\tau}\big).
\end{equation}
A further differentiation gives 
\begin{align*}
\frac1{2a^2} \frac{\pa_z^2Z}{Z} & =  \frac1{2a^2} \pa_z \frac{\pa_zZ}{Z} + \frac1{2a^2}\Big(\frac{\pa_zZ}{Z} \Big)^2\\
&=\frac12\wp(w) + \eta_1 -\frac\pi{2\,\Im\,\tau}+\frac{a^2}2\Big(\zeta(w)^2 + 4\eta_1^2w^2 -4\pi^2\frac{(\Im\,w)^2}{(\Im\,\tau)^2}\Big)\\
&-a^2\Big(2\eta_1w\zeta(w) -2\pi i \,\frac{\Im\,w}{\Im\,\tau}\, \zeta(w)+4\pi \eta_1i \,w\,\frac{\Im\,w}{\Im\,\tau}  \Big).
\end{align*}
Differentiating $Z$ with respect to $\tau,$ it follows from \eqref{eq: heat4theta} and \eqref{eq: theta'0} that 
\begin{align*}
2\pi i \,\frac{\pa_\tau Z}{Z} &=  2\pi i a^2 \pa_\tau \big(\log\theta'(0) -\log\theta(w) + \frac{2\pi}{\Im\,\tau}(\Im\,w)^2 \big)\\
& = a^2\big(-3\eta_1 -\frac12 \frac{\theta''(w)}{\theta(w)} - 2\pi^2\frac{(\Im\,w)^2}{(\Im\,\tau)^2} \big)\\
& = a^2\big(-2\eta_1 -\frac12\zeta^2(w) + 2\eta_1w\zeta(w) +\frac12\wp(w) - 2\eta_1^2 w^2 - 2\pi^2\frac{(\Im\,w)^2}{(\Im\,\tau)^2} \big).\\
\end{align*}
Using \eqref{eg: pa Z}, we have 
\begin{align*}
\frac{2\eta_1}Z\big( z\pa_z + z_0\pa_{z_0}\big)Z = 2\eta_1a^2 \big(-w\zeta(w)+2\eta_1w^2-2\pi i \,w\,\frac{\Im\,w}{\Im\,\tau}\big)
\end{align*}
and
\begin{align*}
\zeta(w)\frac{\pa_{z_0}Z}Z = a^2 \big(\zeta^2(w) -2\eta_1w\zeta(w) + 2\pi i \,\frac{\Im\,w}{\Im\,\tau}\zeta(w)\big).
\end{align*}
Combining all of the above, 
\begin{align*}
\frac1{2a^2} \frac{\pa_z^2Z}{Z} &=\E\,T+2\pi i \,\frac{\pa_\tau Z}{Z} + \frac{2\eta_1}Z\big( z\pa_z + z_0\pa_{z_0}\big)Z  \\&+\zeta(w)\frac{\pa_{z_0}Z}Z  + 2(\lambda+\lambda_0)\eta_1+\lambda_0\wp(w).
\end{align*}
\end{eg*}

\begin{eg*}
We consider $Z = \E\,\OO(z,\bfs z)$ with $\bfs\beta = \bfs 0$ and   
$\bfs\sigma = a\cdot z + \sum\sigma_j \cdot z_j.$
Recall that (see \eqref{eq: C4torus})
\begin{align*}
Z &=  |\theta'(0)|^{a^2 + \sum\sigma_j^2}
\prod_j|\theta(z-z_j)|^{2a\sigma_j}\prod_{j<k}|\theta(z_j-z_k)|^{2\sigma_j\sigma_k}e^{2\pi\frac{(\Im\,\bfs\sigma)^2}{\Im\,\tau}}
\end{align*}
in the $\T_\Lambda$-uniformization. 
In this example, the null vector equation for $Z$ reads as 
\begin{align*}
\frac1{2a^2}\frac{\pa_z^2 Z}{Z}&=\E\,T + 2\pi i\, \frac{\pa_\tau Z}{Z} + \frac{2\eta_1}{Z} \big(z\pa_z + \sum z_j\pa_j\big)Z  \\
&+ \sum_j\zeta(z-z_j)\frac{\pa_jZ}Z + 2\big(\lambda + \sum_j \lambda_j\big)\eta_1+ \sum_j\lambda_j\wp(z-z_j).
\end{align*}
One can check it directly using the pseudo-addition theorem~\eqref{eq: pseudo-addition4zeta0} for the $\zeta$-function.
\end{eg*}

\begin{cor} \index{BPZ-Cardy equations}
For any tensor product $\XX_{\bfs\beta}$ of fields in $\FF_{\bfs\beta},$ 
$$\frac1{2a^2}\Big(\pa_z^2 \wh\E_{\bfs\beta} \XX_{\bfs\beta} + 2\frac{\pa_z Z_{\bfs\beta}}{Z_{\bfs\beta}} \pa_z\wh\E_{\bfs\beta}\XX_{\bfs\beta} \Big)
= \big(2\pi i\,\pa_\tau + 2\eta_1z\pa_z + \check\LL_{\ti v_z}^+\big) \wh\E_{\bfs\beta}\XX_{\bfs\beta}$$
in the $\T_\Lambda$-uniformization.
\end{cor}

\begin{proof}
By Theorem~\ref{BPZ1}, \eqref{eq: hat'}, and similar equations for $\tau$-derivative and Lie derivative $\check\LL_{\ti v_z}$, we have 
\begin{align*}
\frac1{2a^2} \frac{\pa_z^2\E\,\OO_{\bfs\beta}(z) \XX_{\bfs\beta}} {\E\,\OO_{\bfs\beta}(z)} & =
\E\,T_{\bfs\beta}(z) \wh\E_{\bfs\beta}\XX_{\bfs\beta} + 2\pi i \big( \pa_\tau \wh\E_{\bfs\beta}\XX_{\bfs\beta}  + \frac{\pa_\tau Z_{\bfs\beta}}{Z_{\bfs\beta}} \wh\E_{\bfs\beta}\XX_{\bfs\beta} \big)\\
&+ 2\eta_1z\big(\pa_z \wh\E_{\bfs\beta}\XX_{\bfs\beta} +  \frac{\pa_z Z_{\bfs\beta}}{Z_{\bfs\beta}}  \wh\E_{\bfs\beta}\XX_{\bfs\beta}\big) + 2\lambda\eta_1 \wh\E_{\bfs\beta}\XX_{\bfs\beta}\\
&+ \check\LL_{\ti v_z}^+ \wh\E_{\bfs\beta}\XX_{\bfs\beta} +  \frac{\check\LL_{\ti v_z}^+  Z_{\bfs\beta}}{Z_{\bfs\beta}}  \wh\E_{\bfs\beta}\XX_{\bfs\beta}. 
\end{align*}
Corollary now follows from \eqref{eq: hat''} and the null vector equations (Corollary~\ref{null1}) for $Z_{\bfs\beta}.$
\end{proof}

If $\XX_{\bfs\beta} = X_1(z_1)\cdots X_n(z_n),$ then $\check\LL_{k_z}^+\wh\E_{\bfs\beta} \XX_{\bfs\beta}$ in the previous corollary reads as 
\begin{align*}
\check\LL_{k_z}^+\wh\E_{\bfs\beta} \XX_{\bfs\beta} &= \sum_{j>0} \LL_{k_z}^+(z_j) \wh\E_{\bfs\beta} \XX_{\bfs\beta}\\
&+\sum_{j\le0} \big(\zeta(z-z_j) + 2\eta_1 z_j\big)\pa_j  \wh\E_{\bfs\beta} \XX_{\bfs\beta}
+ \sum_k \big(\zeta(z-q_k) + 2\eta_1 q_k\big)\pa_{q_k} \wh\E_{\bfs\beta} \XX_{\bfs\beta}.
\end{align*}

\subsection{Level three degeneracy equations and BPZ equations}  \label{ss: Cardy3}
In this subsection, we derive the level three degeneracy equations for certain current primary fields with specific charges and combine these with Ward's equations to obtain certain types of BPZ equations.
\begin{prop} \label{level3degeneracy} \index{level three degeneracy equations}
Let $\OO := \OO_{\bfs\beta}(z,\bfs z)$ be a current primary field in $\FF_{\bfs\beta}$ with charges $q, q_*$ at $z$ and $q= 2a, (2a(a+b)=1).$
Then we have  
\begin{equation} \label{eq: level3degeneracy} 
\big(L_{-1}^3 -4\tau L_{-2}L_{-1} + (4\tau^2-2\tau)L_{-3}\big)(z)\OO = 0,
\end{equation}
where $\tau = 2a^2.$ 
\end{prop}
\begin{proof} We recall some basic properties (e.g. see Subsections 7.3 and 11.1 in \cite{KM13}) of $J_n\equiv J_n(z)$ and $L_n\equiv L_n(z):$
\begin{equation}\label{eq: [L,J]}
[L_m,J_n] = -nJ_{m+n}+ibm(m+1)\delta_{m+n,0};
\end{equation}
\begin{equation}\label{eq: [L,L]}
[L_m, L_n]=(m-n)L_{m+n}+\frac c{12}m(m^2-1)\delta_{m+n,0};
\end{equation}
and
\begin{equation}\label{eq: LJ}
L_n = -\frac12\sum_{k=-\infty}^\infty\!:J_{-k}J_{k+n}:\!\,-ib(n+1)J_n,
\end{equation}
where 
$$\!:J_mJ_n\!: ~= \begin{cases} J_mJ_n \quad &\textrm{if } m\le n;\\  J_nJ_m &\textrm{otherwise.}\end{cases} $$
By \eqref{eq: LJ}, we find 
$L_{-1}\OO = -J_{-1}J_0\OO$ and $L_0\OO = -\frac12J_0^2\OO-ibJ_0\OO.$
It follows from \eqref{eq: T primary}, \eqref{eq: J0}, and the previous finding that
\begin{equation}\label{eq: J-1}
\pa_z \OO = L_{-1}\OO= 2ia\, J_{-1}\OO.
\end{equation}
Since $\OO$ is a Virasoro primary field, $L_{-1}\OO = \pa_z \OO$ by definition, see \eqref{eq: T primary}.
By \eqref{eq: J-1}, we have
\begin{equation} \label{eq: L-1 3}
L_{-1}^3\OO = L_{-1}^2\pa \OO = 2ia\, L_{-1}^2J_{-1}\OO.
\end{equation}
It follows from \eqref{eq: LJ}, \eqref{eq: J0}, and \eqref{eq: J-1} that
$$L_{-2}\OO = -J_{-2}J_0\OO -\frac12J_{-1}^2\OO+ibJ_{-2}\OO = i(2a+b)J_{-2}\OO+\frac i{4a}J_{-1}L_{-1}\OO.$$
Using \eqref{eq: [L,J]}, the above equation simplifies to 
\begin{equation}\label{eq: L-2}
L_{-2}\OO = i(2a+b)L_{-1}J_{-1}\OO - i\big(a + \frac1{4a}\big)J_{-1}L_{-1}\OO.
\end{equation}
Again, it follows from \eqref{eq: LJ}, \eqref{eq: J0}, and \eqref{eq: J-1} that 
$$L_{-3}\OO = \big(-(J_{-3}J_0 + J_{-2}J_{-1}) + 2ibJ_{-3}\big)\OO = \big(2i(a+b)J_{-3} + \frac i{2a}J_{-2}L_{-1} \big)\OO.$$
By \eqref{eq: [L,J]}, we find
\begin{align} \label{eq: L-3}
L_{-3}\OO &=  \frac i{2a}(L_{-1}J_{-2}-J_{-2}L_{-1})\OO + \frac i{2a}J_{-2}L_{-1} \OO\\
 &= \frac i{2a}L_{-1}J_{-2}\OO = \frac i{2a}(L_{-1}^2J_{-1}-L_{-1}J_{-1}L_{-1})\OO. \nonumber
\end{align}
Combining \eqref{eq: L-1 3}~--~\eqref{eq: L-3} with \eqref{eq: [L,L]}, we find
\begin{align*}
\big(L_{-1}^3 &-4\tau L_{-1}L_{-2} + (4\tau^2+2\tau)L_{-3}\big)\OO \\
&=i\Big(2a - 4\tau(a+\frac1{2a}) + (4\tau^2+2\tau)\frac1{2a}\Big)L_{-1}^2J_{-1}\OO \\
&+i\Big(4\tau(a+\frac1{4a}) - (4\tau^2+2\tau)\frac1{2a})\Big)L_{-1}J_{-1}L_{-1}\OO = 0
\end{align*}
if $\tau = 2a^2.$ 
Proposition now follows from the equation $L_{-3} = [L_{-1},L_{-2}],$ see \eqref{eq: [L,L]}.
\end{proof}

\subsubsec{Genus zero case} 
Let $\OO_{\bfs\beta}(z)\equiv\OO_{\bfs\beta}(z)^{(2a,\bfs\tau)}:=\OO_{\bfs\beta}[2a\cdot z+\bfs\tau]$ with the neutrality condition $(\NC_0).$
We combine the level three degeneracy equations for $\OO_{\bfs\beta}(z)$ with Ward's equations to derive the following theorem, a version of BPZ equations on the Riemann sphere.

\begin{thm} \label{BPZl3g0} \index{BPZ equations}
If $z\notin \supp\,\bfs\beta\cup\supp\,\bfs\tau,$ then for any tensor product $\XX_{\bfs\beta}$ of fields $X_j$ in $\FF_{\bfs\beta},$ 
\begin{align*}
\frac1{8a^2}\pa_z^3\, \E\,  \OO_{\bfs\beta}(z)\XX_{\bfs\beta} &= \big(\E\,T_{\bfs\beta}(z) + \check\LL_{k_z}^+ \big) \pa_z\E\,\OO_{\bfs\beta}(z)\XX_{\bfs\beta}\\
& - (2a^2-\frac12) \big(\E\,\pa T_{\bfs\beta}(z) + \check\LL_{\ti k_z}^+ \big)\E\,\OO_{\bfs\beta}(z)\XX_{\bfs\beta}
\end{align*}
in the $\wh\C$-uniformization. Here, the vector fields $k_\xi$ and $\ti k_\xi$ are given by
$$k_\xi(z) = \frac1{\xi-z}, \qquad  \ti k_\xi(z) = -\frac1{(\xi-z)^2}$$
in the $\wh\C$-uniformization and the Lie derivatives operators $\check\LL_{k_z}^+,\check\LL_{\ti k_z}^+$ do not apply to the $z$-variable. 
\end{thm} 

\begin{proof}
By Ward's equations, 
$$\E\,\pa T_{\bfs\beta}(\xi) \OO_{\bfs\beta}(z)\XX_{\bfs\beta} = \E\,\pa T_{\bfs\beta}(\xi) \E\,\OO_{\bfs\beta}(z)\XX_{\bfs\beta} +\LL_{\ti k_\xi}^+ \E\,\OO_{\bfs\beta}(z)\XX_{\bfs\beta}.$$
We subtract $\Sing_{\xi\to z}\E\, \pa T_{\bfs\beta}(\xi) \OO_{\bfs\beta}(z)\XX_{\bfs\beta}$ from both sides and take the limit as $\xi\to z$ to obtain
$$\E\, L_{-3}(z) \OO_{\bfs\beta}(z)\XX_{\bfs\beta} = \big(\E\,\pa T_{\bfs\beta}(z) + \check\LL_{\ti k_z}^+ \big)\E\, \OO_{\bfs\beta}(z)\XX_{\bfs\beta} .$$
Differentiating \eqref{eq: Ward equation0} with respect $z,$
\begin{align*}
\E\, T_{\bfs\beta}(\xi) \pa_z\OO_{\bfs\beta}(z)\XX_{\bfs\beta}  &= \E\,T_{\bfs\beta}(\xi) \E\,\pa_z\OO_{\bfs\beta}(z)\XX _{\bfs\beta} +\check\LL_{k_\xi}^+\E\,\pa_z\OO_{\bfs\beta}(z,z_0)\XX_{\bfs\beta}\\
&+ \pa_z\Big(\frac1{\xi-z}\pa_z + \frac\lambda{(\xi-z)^2}\Big)\E\,\OO_{\bfs\beta}(z,z_0)\XX_{\bfs\beta}.
\end{align*}
After using  
$$\check\LL_{k_\xi}^+\E\,\pa_z\OO_{\bfs\beta}(z,z_0)\XX_{\bfs\beta}  = \pa_z\check\LL_{k_\xi}^+\E\,\OO_{\bfs\beta}(z,z_0)\XX_{\bfs\beta},$$ 
we subtract the singular part, $\Sing_{\xi\to z}\E\, T_{\bfs\beta}(\xi) \pa_z\OO_{\bfs\beta}(z,z_0)\XX_{\bfs\beta}$ from both sides and take the limit as $\xi\to z.$ 
Then we have 
$$L_{-2}(z)L_{-1}(z)\E\,\OO_{\bfs\beta}(z,z_0)\XX_{\bfs\beta}= \big(\E\,T_{\bfs\beta}(z) + \check\LL_{k_z}^+ \big) \pa_z\E\,\OO_{\bfs\beta}(z)\XX_{\bfs\beta}.$$
Theorem now follows from the level three degeneracy equation~\eqref{eq: level3degeneracy} for $\OO_{\bfs\beta}.$
\end{proof}

Let $Z_{\bfs\beta}= \E\,\OO_{\bfs\beta}(z)$ and $\wh\E_{\bfs\beta} \XX_{\bfs\beta}  = \E\,\OO_{\bfs\beta}(z)\XX_{\bfs\beta}/Z_{\bfs\beta}.$

\begin{cor} \index{null vector equations}
The functions $Z_{\bfs\beta}$ satisfy the null vector equations
$$
\frac1{8a^2} \pa_z^3 Z_{\bfs\beta} = \big(\E\,T_{\bfs\beta}(z) + \check\LL_{k_z}^+ \big)\pa_zZ_{\bfs\beta}  - (2a^2-\frac12) \big(\E\,\pa T_{\bfs\beta}(z) + \check\LL_{\ti k_z}^+ \big)Z_{\bfs\beta}$$
in the $\wh\C$-uniformization.
\end{cor}

\begin{cor} \label{BPZ-Cardy0} \index{BPZ-Cardy equations}
For any tensor product $\XX$ of fields in $\FF_{\bfs\beta},$ we have 
\begin{align*}
\frac1{8a^2}\Big(\pa_z^3 \wh\E_{\bfs\beta} \XX_{\bfs\beta} &+ 3\frac{\pa_z Z_{\bfs\beta}}{Z_{\bfs\beta}} \pa_z^2 \wh\E_{\bfs\beta}\XX_{\bfs\beta} + 3\frac{\pa_z^2 Z_{\bfs\beta}}{Z_{\bfs\beta}} \pa_z \wh\E_{\bfs\beta} \XX_{\bfs\beta}\Big)\\
&= 
\big( \E\,T_{\bfs\beta}(z)+ \check\LL_{k_z}^+  )\pa_z\wh\E_{\bfs\beta}\XX_{\bfs\beta} - (2a^2-\frac12)\check\LL_{\ti k_z}^+\E_{\bfs\beta}\XX_{\bfs\beta} \\
&+ \frac{\check\LL_{k_z}^+Z_{\bfs\beta}}{Z_{\bfs\beta}} \pa_z\wh\E_{\bfs\beta}\XX_{\bfs\beta} 
+ \frac{\pa_zZ_{\bfs\beta}}{Z_{\bfs\beta}}  \check\LL_{k_z}^+\wh\E_{\bfs\beta}\XX_{\bfs\beta}
\end{align*}
in the $\wh\C$-uniformization.

\end{cor}
\begin{proof}
Differentialting ${\E\,\OO_{\bfs\beta}(z) \XX_{\bfs\beta}}/{\E\,\OO_{\bfs\beta}(z)} = \wh\E_{\bfs\beta} \XX_{\bfs\beta},$ we have 
$$\frac{\pa_z^3\E\,\OO_{\bfs\beta}(z) \XX_{\bfs\beta}}{\E\,\OO_{\bfs\beta}(z)} =  \pa_z^3 \wh\E_{\bfs\beta} \XX + 3\frac{\pa_z Z_{\bfs\beta}}{Z_{\bfs\beta}} \pa_z^2 \wh\E_{\bfs\beta}\XX_{\bfs\beta} + 3\frac{\pa_z^2 Z_{\bfs\beta}}{Z_{\bfs\beta}} \pa_z \wh\E_{\bfs\beta} \XX_{\bfs\beta} +  \frac{\pa_z^3 Z_{\bfs\beta}}{Z_{\bfs\beta}}  \wh\E_{\bfs\beta} \XX_{\bfs\beta}$$
and 
$$\frac{\check\LL_{\ti k_z}^+\E\,\OO_{\bfs\beta}(z) \XX_{\bfs\beta}}{\E\,\OO_{\bfs\beta}(z)} = \check\LL_{\ti k_z}^+ \wh\E_{\bfs\beta} \XX_{\bfs\beta}  + \frac{\check\LL_{\ti k_z}^+Z_{\bfs\beta} } {Z_{\bfs\beta}}\wh\E_{\bfs\beta} \XX_{\bfs\beta} .$$
In addition, we find 
\begin{align*}
\frac{\check\LL_{k_z}^+ \pa_z\E\,\OO_{\bfs\beta}(z,z_0) \XX_{\bfs\beta}}{\E\,\OO_{\bfs\beta}(z,z_0)} &= \check\LL_{k_z}^+\pa_z\wh\E_{\bfs\beta} \XX_{\bfs\beta}  + \frac{\check\LL_{k_z}^+Z_{\bfs\beta} } {Z_{\bfs\beta}}\pa_z\wh\E_{\bfs\beta} \XX_{\bfs\beta}\\
 &+   \frac{\pa_zZ_{\bfs\beta} } {Z_{\bfs\beta}}\check\LL_{k_z}^+\wh\E_{\bfs\beta} \XX_{\bfs\beta} + \frac{\check\LL_{k_z}^+\pa_zZ_{\bfs\beta} } {Z_{\bfs\beta}}\wh\E_{\bfs\beta} \XX_{\bfs\beta}.
\end{align*}
Combing all of the above with the BPZ equations for $\E\,\OO_{\bfs\beta}(z)\XX_{\bfs\beta}$ and the null vector equations for $Z_{\bfs\beta},$  the assertion follows. 
\end{proof}

\subsubsec{Genus one case} 
We now derive the following version of the BPZ equations for $\T_\Lambda$ which comes from the level three degeneracy equations. 

\begin{thm} \label{BPZl3g1} \index{BPZ equations}
Let $\OO_{\bfs\beta}(z) = \OO_{\bfs\beta}^{(2a,\bfs\tau)}(z).$
If $z\notin\supp\,\bfs\beta \cup \supp\,\bfs\tau,$ then for any tensor product $\XX_{\bfs\beta}= X_1(z_1)\cdots X_n(z_n)$ of fields $X_j$ in $\FF_{\bfs\beta},$ we have 
\begin{align*}
\frac1{8a^2}\pa_z^3\, &\E\,\OO_{\bfs\beta}(z)\XX_{\bfs\beta}\\ 
&=  \big(\E\,T_{\bfs\beta}(z)+ 2\pi i\,\pa_\tau+(2\eta_1z\pa_{z} + 2(\lambda+1) \eta_1)+ \check\LL_{\ti v_z}^+\big) \pa_z\E\,\OO_{\bfs\beta}(z)\XX_{\bfs\beta}\\
&- (2a^2-\frac12) \big(\E\,\pa T_{\bfs\beta}(z) + \check\LL_{v_z}^+ \big)\E\,\OO_{\bfs\beta}(z)\XX_{\bfs\beta}
\end{align*}
in the $\T_\Lambda$-uniformization.
Here, the vector fields $\ti v_\xi$ and $v_\xi$ are given by
$$\ti v_\xi(z) = \zeta(\xi-z) + 2\eta_1z, \qquad  v_\xi(z) = -\wp(\xi-z)$$
in the $\T_\Lambda$-uniformization and the Lie derivatives operators $\check\LL_{\ti v_z}^+,\check\LL_{v_z}^+$ do not apply to the $z$-variable. 
\end{thm}

\begin{proof} 
It follows from \eqref{eq: Ward4paT} in Theorem~\ref{main: KM2} that 
$$\E\, \pa T_{\bfs\beta}(\xi) \OO_{\bfs\beta}(z)\XX_{\bfs\beta} = \E\, \pa T_{\bfs\beta}(\xi) \,  \E\, \OO_{\bfs\beta}(z)\XX_{\bfs\beta} +  \LL_{v_\xi}^+ \E\, \OO_{\bfs\beta}(z)\XX_{\bfs\beta}.$$
Subtracting $\Sing_{\xi\to z}\E\, \pa T_{\bfs\beta}(\xi) \OO_{\bfs\beta}(z)\XX_{\bfs\beta}$ from both sides and then taking the limit as $\xi\to z,$ we find
$$\E\, L_{-3}(z) \OO_{\bfs\beta}(z)\XX =  \big(\E\,\pa T_{\bfs\beta}(z) + \check\LL_{v_z}^+ \big)\E\, \OO_{\bfs\beta}(z,z_0)\XX.$$
By Theorem~\ref{EOKM}, we have 
\begin{align*}
 \E\,T_{\bfs\beta}(\xi)&\pa_z\OO_{\bfs\beta}(z)\XX_{\bfs\beta}\\
 &=   \E\,T_{\bfs\beta}(\xi)\E\,\pa_z\OO_{\bfs\beta}(z)\XX_{\bfs\beta}+2\pi i \, \pa_\tau \,\E\,\OO_{\bfs\beta}(z)\XX_{\bfs\beta} + \check\LL_{\ti v_\xi}^+\E\,\pa_z\OO_{\bfs\beta}(z)\XX_{\bfs\beta}\\
&+ \big((\zeta(\xi-z) +2\eta_1z)\pa_{z} + (\lambda+1)( \wp(\xi-z) + 2\eta_1)\big)\E\,\pa_z\OO_{\bfs\beta}(z)\XX_{\bfs\beta}.
\end{align*}
Subtracting $\Sing_{\xi\to z}\E\, T_{\bfs\beta}(\xi) \pa_z\OO_{\bfs\beta}(z,z_0)\XX$ from both sides and then taking the limit as $\xi\to z,$ we find
$$L_{-2}(z)L_{-1}(z)\E\,\OO_{\bfs\beta}(z)\XX_{\bfs\beta}= (2\eta_1z\pa_{z} + 2(\lambda+1) \eta_1+\check\LL_{\ti v_z}^+)\pa_z\E\,\OO_{\bfs\beta}(z)\XX_{\bfs\beta}.$$
Theorem now follows from the level three degeneracy equation~\eqref{eq: level3degeneracy} for $\OO_{\bfs\beta}.$
\end{proof}

Let $Z_{\bfs\beta}= \E\,\OO_{\bfs\beta}(z)$ and $\wh\E_{\bfs\beta} \XX_{\bfs\beta}  = \E\,\OO_{\bfs\beta}(z)\XX_{\bfs\beta}/Z_{\bfs\beta}.$

\begin{cor} \index{null vector equations}
The functions $Z_{\bfs\beta}$ satisfy the null vector equations
\begin{align*}
\frac1{8a^2}\pa_z^3\, Z_{\bfs\beta} &= \big(\E\,T_{\bfs\beta}(z)+2\pi i\,\pa_\tau+2\eta_1z\pa_{z} + 2(\lambda+1) \eta_1 +  \check\LL_{\ti v_z}^+\big)\pa_z Z_{\bfs\beta}  \\
& - (2a^2-\frac12) \big(\E\,\pa T_{\bfs\beta}(z) + \check\LL_{v_z}^+ \big)Z_{\bfs\beta}
\end{align*}
in the $\T_\Lambda$-uniformization.
\end{cor}

The method of proof of Corollary~\ref{BPZ-Cardy0} in the genus zero case carries over to the genus one case.  
\begin{cor} \index{BPZ-Cardy equations}
For any tensor product $\XX$ of fields in $\FF_{\bfs\beta},$ we have 
\begin{align*}
\frac1{8a^2}\Big(\pa_z^3 \wh\E_{\bfs\beta} \XX_{\bfs\beta} &+ 3\frac{\pa_z Z_{\bfs\beta}}{Z_{\bfs\beta}} \pa_z^2 \wh\E_{\bfs\beta}\XX_{\bfs\beta} + 3\frac{\pa_z^2 Z_{\bfs\beta}}{Z_{\bfs\beta}} \pa_z \wh\E_{\bfs\beta} \XX_{\bfs\beta}\Big)\\
&= 
 \big(\E\,T_{\bfs\beta}(z)+ 2\pi i\,\pa_\tau+(2\eta_1z\pa_{z} + 2(\lambda+1) \eta_1)+ \check\LL_{\ti v_z}^+\big) \pa_z \E_{\bfs\beta}\XX_{\bfs\beta}\\
& - (2a^2-\frac12) \check\LL_{v_z}^+\E_{\bfs\beta}\XX_{\bfs\beta} 
\\
&
+ 2\pi i \big(\frac{\pa_\tau Z_{\bfs\beta}}{Z_{\bfs\beta}}\pa_z \wh\E_{\bfs\beta}\XX_{\bfs\beta} + \frac{\pa_z Z_{\bfs\beta}}{Z_{\bfs\beta}}\pa_\tau \wh\E_{\bfs\beta}\XX_{\bfs\beta}\big)+ 4\eta_1z \frac{\pa_z Z_{\bfs\beta}}{Z_{\bfs\beta}}\pa_z \wh\E_{\bfs\beta} \XX_{\bfs\beta} \\
&+ \frac{\check\LL_{\ti v_z}^+Z_{\bfs\beta}}{Z_{\bfs\beta}} \pa_z\wh\E_{\bfs\beta}\XX_{\bfs\beta} 
+ \frac{\pa_zZ_{\bfs\beta}}{Z_{\bfs\beta}}  \check\LL_{\ti v_z}^+\wh\E_{\bfs\beta}\XX_{\bfs\beta}
\end{align*}
in the $\T_\Lambda$-uniformization.

\end{cor}

\renewcommand\sectionname{Appendix}
\renewcommand\thesection{A}

\section{Representation of bipolar Green's function} \label{sec: Theta}

The goal of this appendix is to represent bipolar Green's function on a compact Riemann surface $M$ in terms of the period matrix $\PM$ of $M,$ Riemann's theta function associated with $\PM,$ and the Abel-Jacobi map $\AA.$

\subsection{Period matrix and theta function} \label{ss: Period matrix and theta function}
\subsubsec{Period matrix} 
Let $M$ be a compact Riemann surface of genus $g\ge 1.$ 
Fix a canonical basis $\{a_j,b_j\}$ for the homology $H_1 = H_1(M)$ with the following intersection properties: $a_j\cdot b_j = 1$ and all other intersection numbers are zero.
Let $\Omega(M)$ be the space of all holomorphic 1-differentials on $M$ and let $\{\omega_j\}$ be its basis uniquely determined by the equations 
$$\oint _{a_k}\omega_j = \delta_{jk}.$$
The \emph{period matrix} \index{period matrix} $\PM = \{\PM_{jk}\}$ is defined as 
$$\PM_{jk} =\oint_{b_k}\omega_j.$$
It is well known that the period matrix is symmetric and its imaginary part is positive definite, $\Im\,\PM >0.$
See \cite[III.2.8]{FK}. 

\subsubsec{Theta function}
Let $\PM$ be a symmetric $g\times g$ matrix with $\Im\,\PM>0$ (e.g. the period matrix of a Riemann surface). 
The \emph{theta function} \index{theta function} $\Theta(\cdot \,|\,\PM)$ associated with $\PM$ is the following function of $g$ complex variables $Z = (z_1,\cdots,z_g),$
$$\Theta(Z\,|\,\PM) = \sum_{N\in\Z^g} \ee^{2\pi i (Z\cdot N + \frac12 \PM N\cdot N)} \qquad (Z\in\C^g).$$
The theta function is an \emph{even} entire function on $\C^g$ (or a multivalued function on the Jacobi variety, see the next subsection).
It has the following periodicity properties:
for $N\in\Z^g,$ we have 
$$\Theta(Z+N) = \Theta(Z),\qquad \Theta(Z+\PM N) = \ee^{-2\pi i(Z\cdot N + \frac12 \PM N\cdot N)} \Theta(Z). $$
See \cite[VI.1.2]{FK}. 

\begin{rmk*} 
If $g = 1,$ then $\Theta = \theta_3.$ Its relation to $\theta_1$ (which we used previously) is given by
$$\theta_1(z) = \textrm{const}\, \ee^{-\pi i z}\theta_3\Big(z + \frac{1-\tau}2\Big).$$
\end{rmk*}

\subsection{Abel-Jacobi map} \label{ss: Abel-Jacobi} 
Let $\PM$ be the period matrix of $M$ (with a fixed canonical basis and therefore the basis $\{\omega_j\}$ of $\Omega(M)$).
We consider the lattice $\Lambda = \Z^g + \PM \Z^g$ in $\C^g$ associated with the period matrix $\PM$ of $M$ and set 
$$\T_\Lambda \equiv \T_\Lambda^{g}:=\C^{g}/\Lambda.$$
It is homeomorphic to a topological torus $S^1\times\cdots\times S^1$ ($g$ times).
We also set 
$$\Jac(M):=\T_\Lambda.$$
The complex torus $\Jac(M)$ is called the \emph{Jacobi variety} \index{Jacobi variety} or \emph{Jacobian} of $M.$
For each $p_0\in M$ we define the \emph{Abel-Jacobi map} \index{Abel-Jacobi map} by 
$$\AA\equiv \AA_{p_0}:M\to\Jac(M),\quad p\mapsto \int_{p_0}^p \vec\omega,$$
where $ \vec\omega = (\omega_1,\cdots,\omega_g).$
This map is a holomorphic embedding. 
See \cite[III.6.1]{FK}. 

\begin{rmks*} 
\renewcommand{\theenumi}{\alph{enumi}}
{\setlength{\leftmargini}{1.8em}
\begin{enumerate}
\item Up to a torus translation, the Abel-Jacobi map does not depend on the choice of a base point $p_0.$ 
Let 
$$\AA(p) = \int_{p_0}^p \vec\omega, \quad \ti \AA(p) = \int_{\ti p_0}^p \vec\omega.$$ 
Since $\AA(p_0) = 0,$ we have
$$\AA - \ti{\AA} = \AA(\ti p_0) = -\ti \AA(p_0).$$
\item Since $\Jac(M)$ is an Abelian group, we can extend the Abel-Jacobi map to the group $\div(M)$ of divisors on $M:$ 
$$\AA: \div(M) \to \Jac(M).$$
By definition, a \emph{divisor} \index{divisor} $D$ on $M$ is the formal finite linear combination 
$$D = \sum \alpha_j \cdot p_j, \quad \alpha_j\in\Z, \, p_j\in M.$$ 
The sum $\deg D:= \sum \alpha_j$ is called the degree of $D.$
\item For a meromorphic function $f$, Abel's theorem (\cite[III.6.3]{FK}) says that $\AA(f) = 0.$ 
Here, $(f)$ is the divisor of $f:$
$$(f) = \sum_{p\in M} \mathrm{ord}_p f\cdot p,$$
where $ \mathrm{ord}_p f$ is the multiplicity of $p$ if $p$ is a zero of $f$ and $- \mathrm{ord}_p f$ is the multiplicity of $p$ if $p$ is a pole of $f.$
\item If $\xi$ is a meromorphic 1-differential, then $\AA(\xi) = -2 \mathcal{K},$ where $\mathcal{K}$ is the \emph{vector of Riemann constants}, \index{vector of Riemann constants} see Subsection~\ref{ss: K}.
\end{enumerate}
}
\end{rmks*}

A divisor $D$ is called \emph{integral} \index{divisor!integral} if $D = \sum \alpha_j \cdot p_j, \alpha_j\ge0$ and we write $D\ge 0.$
We denote by $M^{(k)} :=\div_+^k(M)$ the set of integral divisors of degree $k$ on $M$ and set $J^{(k)} = \AA(M^{(k)}).$
By Jacobi's theorem (\cite[III.6.6]{FK}), $\Jac(M)$ is stratified as 
$$\{0\} = J^{(0)} \subset J^{(1)} \subset \cdots \subset J^{(g)} = \Jac(M).$$

For a non-special divisor $D\in M^{(g)},$ (see the definition of special divisors in the next subsection) $\AA: M^{(g)} \to \Jac(M)$ is a local homeomorphism at $D.$
See \cite[VI.2.5]{FK}. 
\medskip
 
We write $D_1\ge D_2$ if $D_1-D_2\ge0.$
For $D\in \div(M),$ we denote 
\begin{align*}
\mathcal{R}(D)&= \textrm{the vector space of meromorphic functions } f \textrm{ on } M \textrm{ such that } \\
&\qquad\qquad\qquad\qquad\quad (f) \ge D \textrm{ or } f\equiv0, \\
\mathcal{I}(D)&= \textrm{the vector space of Abelian differentials } \omega\textrm{ on } M \textrm{ such that } \\ &\qquad\qquad\qquad\qquad\quad (\omega) \ge D \textrm{ or } \omega\equiv0,
\end{align*}
and set $r(D)= \dim \mathcal{R}(D),$ $i(D) = \dim\mathcal{I}(D).$
The Riemann-Roch Theorem says that 
$$r(-D) = \deg D - g +1 + i(D)$$
for each divisor $D$ on $M.$ 
 
\subsection{Special divisors and the vector of Riemann constants} \label{ss: K}
\subsubsec{Special divisors}
By definition, an integral divisor $D$ of degree $g$ on $M$ is \emph{special} \index{divisor!special} if $i(D) > 0,$ i.e. if there exists a non-trivial 1-differential $\xi$ such that $(\xi) \ge D.$ 
By the Riemann-Roch Theorem, $D$ is special if and only if there exists a non-constant meromorphic function $f$ such that $(f) + D \ge 0.$

\begin{eg*} 
A point $p\in M$ is called a \emph{Weierstrass point} \index{Weierstrass point} if the divisor $D = g\cdot p$ is special. 
Equivalently, there exists a meromorphic function $f$ with a single pole at $p$ of order $\le g;$ or there is a holomorphic differential with zero of order $\ge g$ at $p.$ 
There are at least $2g+2$ and at most $g^3-g$ points. 
See \cite[III.5.11]{FK}
\end{eg*}

\subsubsec{The vector of Riemann constants}

We now introduce the \emph{vector $\mathcal{K}_{p_0}=(K_1,\cdots,K_g)$ of Riemann constants} \index{vector of Riemann constants} with a base point $p_0:$
$$K_j = \frac{1+\PM_{jj}}2 - \sum_{k\ne j}\oint_{a_k}\omega_k(p) \int_{p_0}^p \omega_j.$$
Recall that $\omega_j$'s are uniquely determined by the equations 
$$\oint _{a_k}\omega_j = \delta_{jk}$$
and they form a basis of $\Omega(M).$
The vector $\mathcal{K}$ of Riemann constants is well-defined as an elements of $\T_\Lambda.$
\begin{rmk*}
This vector $\mathcal{K}_{p_0}$ depends on the base point $p_0:$
$$\mathcal{K}_{p_0} = \mathcal{K}_{q_0} + (g-1)\AA_{q_0}(p_0).$$
Thus $\AA_{p_0} + \mathcal{K}_{p_0}: M^{(g-1)} \to \T_\Lambda$ is independent of the choice of $p_0.$ 
Therefore, a base point $p_0$ can be omitted in the statement of Theorem~\ref{Theta divisor} below. 
\end{rmk*}

\subsection{Theta divisor, theta equation, and periodicity of theta function}
\subsubsec{Theta divisor} 
\index{theta divisor}
The zero set $\Theta_\zero \subset \T_\Lambda$ of the theta functions is well defined on the torus $\T_\Lambda.$

\begin{thm}$($\cite[VI.3.1]{FK}$)$ \label{Theta divisor} 
We have 
$$ \Theta_\zero =J^{(g-1)} + \mathcal{K},$$
i.e. $\Theta(e) = 0$ if and only if there exists $D\in M^{(g-1)}$ such that 
$$e = \AA(D) + \mathcal{K}.$$
\end{thm}
Note that $\Theta_\zero$ does not depend on the choice of a base point $p_0.$ 
See the remark in the previous subsection.

We say that $e\in\T_\Lambda$ is \emph{singular} with respect to $p_0,$ and write $e\in\Theta_\sing(p_0)$ if 
$$\Theta(\AA(p)-e)\equiv 0, \qquad (p\in M),$$
where $\AA$ has a base point at $p_0.$

\begin{thm}$($\cite[VI.3.3]{FK}$)$  \label{sing} 
The following are equivalent:
\renewcommand{\theenumi}{\alph{enumi}}
{\setlength{\leftmargini}{1.8em}
\begin{enumerate} 
\item The point $e$ is singular with respect to $p_0;$
\item there exists $D$ in $M^{(g-1)}$ such that the divisor $(p_0 + D)$ is special and 
$$ e = \AA(D) + \mathcal{K}.$$
\end{enumerate}
}
\end{thm}

\begin{eg*}
If $p_0$ is a Weierstrass point, then $e=\mathcal{K}$ is singular with respect to $p_0.$
\end{eg*}

\subsubsec{Theta equation}
We now describe the location of the zeros of theta function.
\begin{thm}$($\cite[VI.2.4]{FK}$)$  \label{Theta equation}
If $e\in\T_\Lambda$ is not singular with respect to $p_0,$ then the equation 
$$\Theta(\AA(z) - e) = 0, \qquad z\in M$$
has exactly $g$ solutions $z_j$, and 
$$\sum_{j=1}^g \AA(z_j) = e-\mathcal{K}.$$ 
\end{thm}

Let $p,q\in M, p\ne q.$ We choose 
$$ e\in \Theta_\zero\setminus(\Theta_\sing(p)\cup\Theta_\sing(q))\ne\emptyset.$$
Recall that non-singularity of $e$ means
$$\Theta(\AA(z)-\AA(p)-e)\not\equiv 0, \qquad \Theta(\AA(z)-\AA(q)-e) \not\equiv0.$$
Here $\AA:M\to \T_\Lambda$ is the Abel-Jacobi map with any fixed base point $p_0.$

\begin{lem} The ``function''
$$z\mapsto \frac{\Theta(\AA(z)-\AA(q)-e)}{\Theta(\AA(z)-\AA(p)-e)}, \qquad (z\in M)$$
has a well-defined (zeros-poles) divisor, which is exactly $q-p.$
\end{lem}

\begin{proof} It follows from Theorem~\ref{Theta equation} that the numerator and the denominator of our ``function'' have exactly $g$ zeros. 
Let $q + D$ and $p + \ti D$ be the corresponding divisors. 
(Note that $q$ is a zero of the numerator and $p$ is a zero of the denominator.) 
Then we have
$$\AA(q) + \AA(D) = e + \AA(q) - \mathcal{K}, \qquad \AA(p) + \AA(\ti D) =e +\AA(p) - \mathcal{K},$$
and therefore
$$\AA(D) = \AA(\ti D).$$
Let us show that in fact $D = \ti D.$ If $D \ne \ti D,$ by Abel's theorem there exists a
non-constant meromorphic function $f$ on $M$ such that
$$(f) = \ti D - D\ge -D.$$
By the Riemann-Roch theorem, this gives $i(D)\ge 2$ and therefore $i(q+D)\ge1,$ which means
that the divisor $(q + D)$ is special. 
By Theorem~\ref{sing} we have
$$\Theta(\AA(z) - \AA(q + D) -\mathcal K) \equiv 0.$$
Since $\AA(q + D) = e + \AA(q) -\mathcal K,$ we arrive to a contradiction.
\end{proof}

\subsubsec{Periodicity of theta function}
The following two lemmas follow immediately from the periodicity of theta function $\Theta.$

\begin{lem}
Let $P,Q\in \C^g.$ Then the function
$$Z\mapsto F_{P,Q}(Z)=\log\Big|\frac{\Theta(Z-Q)}{\Theta(Z-P)}\Big| - 2\pi (\Im\,\PM)^{-1}\Im(P-Q)\cdot \Im\,Z $$
is $\Lambda$-periodic.
\end{lem}

\begin{proof}
First, it is obvious that $F_{P,Q}(Z+N) = F_{P,Q}(Z).$ 
Secondly, we have
\begin{align*}
F_{P,Q}(Z + \tau N) &= F_{P,Q}(Z) + 2\pi\,\Im (P-Q)\cdot N - 2\pi (\Im\,\PM)^{-1}\Im(P-Q)\cdot \Im\,\PM N \\
&= F_{P,Q}(Z).
\end{align*}
\end{proof}

This function is defined in $\C^g$ wherever we don't have $\frac 00.$ 

\begin{lem}
If $P-\tilde P\in \Lambda$ and $Q-\tilde Q\in\Lambda,$ then
$F_{P,Q}-F_{\tilde P,\tilde Q}$ is a constant (which depends on $P,Q,\tilde P,$ and $\tilde Q).$
\end{lem}
\begin{proof}
For $\tilde P = P - \PM\cdot N, \tilde Q = Q - \PM\cdot \tilde N,$ we have 
\begin{align*}
F_{P,Q}(Z)-F_{\tilde P,\tilde Q}(Z) &= \log\Big|\frac{\Theta(Z-P+\PM\cdot N)}{\Theta(Z-P)}\frac{\Theta(Z-Q)}{\Theta(Z-Q+\PM\cdot \tilde N)}\Big| \\&- 2\pi (\Im\,\PM)^{-1}\Im(\PM(N-\tilde N))\cdot \Im\,Z.
\end{align*}
Using the periodicity of Riemann theta function, $F_{P,Q}(Z)-F_{\tilde P,\tilde Q}(Z)$ simplifies to
$$C_{N,\tilde N} + 2\pi\,\Im\,Q\cdot \tilde N - 2\pi\,\Im\,P\cdot N$$
for some constant $C_{N,\tilde N}$ depending on $N,\tilde N.$
The other cases can be checked similarly. 
\end{proof}

If $P,Q\in\T_\Lambda,$ then we consider $F_{P,Q}(\cdot) + \R$ as an equivalence class of functions on $\T_\Lambda.$

\subsection{Bipolar Green's function on a compact Riemann surface}
Now we have an explicit formula for bipolar Green's function $G_{p,q}$ on $M.$ \index{bipolar Green's function}
\begin{thm} \label{Green M} If $p,q\in M,$ and $e\in\T_\Lambda$ satisfies
$$e\in \Theta_\zero\setminus (\Theta_\sing(p)\cup\Theta_\sing(q)),$$
then up to an additive real constant, we have  
$$G_{p,q}(z) = F_{\AA(p)+e,\AA(q)+e}( \AA(z)).$$

\end{thm}

This theorem now follows from three previous lemmas since the well-defined function $G_{p,q}$ on $M$ (up to an additive real constant) is harmonic on $M\setminus\{p,q\},$ 
and satisfies
\begin{align*}
G_{p,q}(z) &= \phantom{-}\log\frac1{|z-p|} + O(1) \qquad (z\to p),\\
 G_{p,q}(z) &= -\log\frac1{|z-q|} + O(1) \qquad (z\to q)
\end{align*}
in some/any chart.

%\backmatter

%\bibliographystyle{plain}
%\bibliographystyle{alpha}
%\bibliography{KM3}

%\begin{comment}

%\begin{thebibliography}{10}

%\end{thebibliography}
%\end{comment}

%\printindex
%\input{KM3.ind}

\end{document}